\documentclass[12pt,preprint,twoside]{article} %fleqn
 \usepackage{ams}
 \usepackage{dsfont}
 \usepackage{amsmath}
 \usepackage{amssymb}
 \usepackage{dsfont}
 \usepackage{amsthm}
 \usepackage{rotating}
 \usepackage{fancyhdr}
 \usepackage{fancyheadings}
 \RequirePackage[OT1]{fontenc}
 \RequirePackage{amsthm,amsmath}
 \RequirePackage[numbers]{natbib}
 \RequirePackage[colorlinks,citecolor=red,linkcolor=blue,urlcolor=red]{hyperref}

 \usepackage{titlesec}
 \titleformat{\section}{\large\bfseries}{\thesection}{1em}{}
 \titleformat{\subsection}{\normalsize\bfseries}{\thesubsection}{1em}{}
 \titleformat{\subsubsection}{\normalsize\it}{{\rm \thesubsubsection}}{1em}{\vspace{-.7em}}

 \usepackage{mathptmx,bm}
 \usepackage{times}

 \newtheoremstyle{mytheo}%        name
                {\topsep}%    space above
                {\topsep}%    space below
                {\rm}%        body text
                {}%           indent
                {\sc }%  head text
                {.}%          punctuation after head
                {0.5em}%      space after head
                {}%           head spec

 \theoremstyle{mytheo}

 \newtheorem{theo}{Theorem}[section]

 \newtheorem{cor}{Corollary}[section]
 
 \newtheorem{defi}{Definition}[section]
 \newtheorem{exam}{Example}[section]

 \newtheorem{lem}{Lemma}[section]

 \newtheorem{prop}{Proposition}[section]
 
 \newtheorem{rem}{Remark}[section]

 \renewcommand{\thetable}{\thesection.\arabic{table}}

 \numberwithin{equation}{section}
 
 \renewcommand{\Pr}{\mathds{P}}

 \renewcommand{\tablename}{{\sc Table}}

 \newcommand{\be}{\begin{equation}}
 \newcommand{\ee}{\end{equation}}
 \newcommand{\E}{\mathds{E}}
 \newcommand{\Var}{\mbox{\rm \hspace*{.2ex}Var\hspace*{.2ex}}}
 \newcommand{\Cov}{\mbox{\rm \hspace*{.2ex}Cov\hspace*{.2ex}}}

 \renewcommand{\R}{\mathds{R}}
 \renewcommand{\C}{\mathds{C}}
 \newcommand{\essinf}{\mbox{\rm \hspace*{.2ex}ess}\inf}
 \newcommand{\esssup}{\mbox{\rm \hspace*{.2ex}ess}\sup}

 \newcommand{\ud}{\hspace{0.1ex}\textrm{$\rm d$}}
 \newcommand{\IP}{\textrm{${\rm IP}(\mu;\delta,\beta,\gamma)$}}
 \newcommand{\IPq}{\textrm{${\rm IP}(\mu;q)$}}
 \newcommand{\IPqb}{\textrm{${\rm IP}(\mu^*;q^*)$}}
 \newcommand{\IPtilde}{\textrm{${\rm IP}(\widetilde{\mu};\widetilde{q})$}}
 \newcommand{\IPqtilde}{\textrm{${\rm IP}(\widetilde{\mu};\widetilde{\delta},\widetilde{\beta},\widetilde{\gamma})$}}
 \newcommand{\IPqk}{\textrm{${\rm IP}(\mu_k;q_k)$}}
 \newcommand{\IPqkb}{\textrm{${\rm IP}(\mu_k^*;q_k^*)$}}
 \newcommand{\IPkb}{\textrm{${\rm IP}(\mu_k^*;\delta_k^*,\beta_k^*,\gamma_k^*)$}}
 
 \newcommand{\lead}{\textrm{\rm{lead}\hspace{.2ex}}}
 \newcommand{\xA}[2]{\begin{minipage}[A]{18ex}#1\\\vspace{-4.5ex}\begin{flushright}#2\end{flushright}\end{minipage}}
 \newcommand{\xB}[1]{\begin{minipage}[B]{14ex}\begin{flushright}#1\end{flushright}\end{minipage}}
 \newcommand{\xC}[1]{\textrm{\begin{minipage}[C]{13.8cm}#1\end{minipage}}}
 \newcommand{\ds}{\displaystyle}

 \newcommand{\MR}[1]{\textcolor[rgb]{1,0,0}{MR#1}}

 \title{\Large\bf Integrated Pearson family and orthogonality
 of the Rodrigues polynomials:
 A review including new results and
 an alternative classification of the Pearson
 system\textcolor[rgb]{0,0,1}{\footnote{Work partially
 supported by the University of Athens Research Grant 70/4/5637}}\vspace*{-.6em}}
 \newcommand{\titlerunning}{\textsc{Integrated Pearson family and
 Rodrigues polynomials}}

 \author{\large
 G.\ Afendras\textcolor[rgb]{0,0,
 1}{\footnote{{e-mail:}\ \textcolor[rgb]{0.98,0.00,0.00}{g\_afendras@math.uoa.gr}}}
 \ \ and \ \ N.\
 Papadatos\textcolor[rgb]{0,0,1}{\footnote{{\it
 Corresponding author. }{e-mail:}\
 \textcolor[rgb]{0.98,0.00,0.00}{npapadat@math.uoa.gr}, {url:}\
 \textcolor[rgb]{0.98,0.00,0.00}{users.uoa.gr/$\sim$npapadat/}}}}
 \newcommand{\authorrunning}{\textsc{G.\ Afendras, N.\ Papadatos}}
 \vspace{-1em}

 \date{\small\it
 Department of Mathematics, Section of Statistics and
 O.R., University of Athens,
 \\
 Panepistemiopolis, 157 84 Athens, Greece.
 }
 \vspace{-1em}

\begin{document}

 \maketitle
 \fancypagestyle{plain}{
 \fancyhf{}
 \renewcommand{\headrulewidth}{0pt}
 \lhead{\sc Dedicated to Professor V.\ Papathanasiou}%\arXiv{???}}
 \cfoot{\thepage}
 }

 \pagestyle{fancy} \fancyhf{} \fancyhead[RO,LE]{\thepage}
 \renewcommand{\headrulewidth}{.5pt}
 \fancyhead[LO]{\footnotesize\itshape\titlerunning}
 \fancyhead[RE]{\footnotesize\itshape\authorrunning}

 \vspace{-2em}
 \begin{center}
 \begin{minipage}[1]{25em}
 {\small
 {\bf Abstract:} An alternative classification of the Pearson family of
 probability densities is related to the orthogonality of the corresponding
 Rodrigues polynomials. This leads to a subset of the ordinary
 Pearson system, the {\it
 Integrated Pearson Family}.
 Basic properties of this family are discussed
 and reviewed, and some new results are presented.
 A detailed comparison between the integrated Pearson
 family and the ordinary Pearson system is presented, including
 an algorithm that enables to decide whether a given
 %ordinary
 Pearson density belongs to the integrated system, or not.
 %Finally,
 Recurrences between the derivatives
 of the corresponding orthonormal polynomial systems are also given.
 \medskip

 {\bf MSC:} Primary 62E15, 60E05; Secondary 62-00.
 %%%Primary 33C45, 60E05; Secondary 33D45, 42C05.
 %AMS 2000 subject classification:} Primary 60E05,
 \medskip

 {\bf Key words and phrases:} Integrated Pearson Family of
 distributions; Derivatives of orthogonal polynomials;
 Rodrigues polynomials.
 }
 \end{minipage}
 \vspace{-.5em}
 \end{center}

 \section{Introduction}
 \label{sec1}
 Karl Pearson (1895), in the context of fitting curves to real data,
 introduced his famous family of frequency curves by means of
 the differential equation
 \[
 \frac{f'(x)}{f(x)}=\frac{p_1(x)}{p_2(x)},
 \]
 where $f$ is the probability density and $p_i$ is a polynomial in $x$
 of degree at most $i$, $i=1,2$. Since then, a vast bibliography has been
 developed regarding the properties of Pearson distributions.
 The original classification given by Pearson contains twelve types
 (I--XII), although this numbering system does not have a
 clear systematic basis; Johnson et al.\ (1994), p.\ 16.
 Craig (1936) proposed a new exposition and chart for Pearson
 curves; however, a more reasonable and convenient
 classification is included in a review paper by Diaconis
 and Zabell (1991).
 Extensions to discrete distributions have been introduced
 by Ord (1967) and an extensive review can be found in Ord (1972),
 Chapter 1.

 In this paper we present and review a number of properties
 satisfied by the distributions of the Pearson family and the
 associated Rodrigues polynomials,
 the polynomials that are produced by a Rodrigues-type formula.
 Our main focus is on a suitable subset of Pearson distributions,
 the {\it Integrated Pearson Family}, because this class subsumes
 all interesting properties related to the associated
 orthogonal polynomial systems. For example, it will be shown in
 Section~\ref{sec4} that orthogonality of Rodrigues polynomials with
 respect to an ordinary Pearson density $f$ results to
 an equivalent definition of the integrated Pearson system.
 This consideration entails
 an alternative classification of (integrated) Pearson distributions,
 which is essentially the one given in Diaconis and Zabell
 (1991).

 In the context of deriving variance bounds for functions of random variables,
 Afendras et al.\ (2007, 2011) and Afendras and Papadatos (2011)
 have made use of
  the following definition, which provides
 the main framework of the present article.

 \begin{defi}[Integrated Pearson Family]
 \label{def.IP}
 Let $X$ be an absolutely continuous random
 variable with density $f$ and finite mean $\mu=\E X$.
 We say that $X$ (or its density) belongs to the integrated Pearson family
 (or integrated Pearson system) if there exists a quadratic polynomial
 $q(x)=\delta x^2+\beta x+\gamma$
 (with $\delta,\beta,\gamma\in\R$, $|\delta|+|\beta|+|\gamma|>0$) such that
 \be
 \label{qua1}
 \int_{-\infty}^{x}(\mu-t)f(t)\ud{t}=q(x)f(x)
 \ \  \text{for all}  \
 x\in\R.
 \ee
 This fact will be denoted by
 %\be
 %\label{IP}
 $X\sim \IPq$ or $f\sim\IPq$ or, more explicitly, $X$ or $f\sim\IP$.
 %\ee
 \end{defi}

 Despite the fact that the integrated Pearson family is quite restricted,
 compared to the usual
 Pearson system -- see Proposition \ref{prop.1}(iii), below --
 %Since there is a vast bibliography on Pearson distributions,
 %it appears that its properties should be well-known. However,
 we believe that the reader will find here some interesting observations
 that are worth to be highlighted.
 %One
 %important
 %point is that
 The integrated Pearson system satisfies
 many interesting properties, like recurrences on moments and on
 Rodrigues polynomials,
 covariance identities, closeness of each type under particularly
 useful transformations etc.;
 such properties are by far more complicated (if they are, at all, true)
 for distributions
 outside the Integrated Pearson system.
 These features should be combined with the fact that
 the Rodrigues polynomials form an orthogonal system for the
 corresponding Pearson density if and only if the density belongs
 to the Integrated Pearson family. In other words, the
 Rodrigues polynomials and, consequently, the ordinary Pearson
 densities, are useful only if they are considered in the
 framework of the Integrated Pearson system. To our knowledge,
 these facts have not been written explicitly elsewhere.

 The paper is organized as follows: In Section \ref{sec2}
 we provide a detailed classification of the integrated
 Pearson family. It turns out that, up to an affine transformation,
 there are six different types of densities, included in Table
 \ref{table densities}.
 We also provide conditions guaranteeing the existence of moments,
 and we give recurrences as long as these moments exist.
 In Section \ref{sec3}, a detailed comparison
 between the integrated Pearson family and the ordinary
 Pearson system is presented. Interestingly enough, there exist a simple
 algorithm that enables one to decide whether a given ordinary
 Pearson density belongs to the integrated system, or not.
 In Section \ref{sec4}, exploiting a result of Diaconis and Zabell (1991),
 we show that (under natural moment conditions)
 the first three Rodrigues polynomials (of degree $0$, $1$ and $2$)
 are orthogonal with respect to an ordinary Pearson density
 if and only if this density belongs to the
 integrated Pearson system. Finally, in Section \ref{sec5}
 we provide recurrences between the orthonormal polynomials
 and their derivatives; in fact, the derivatives themselves are orthogonal
 polynomials with respect to other integrating Pearson
 densities, having the same quadratic polynomial, up to a scalar multiple.
 Although we do not include any specific applications of these results
 here, we notice that such recurrences are particularly useful in obtaining Fourier
 expansions of the derivatives of a function of a Pearson
 variate. The main result of Section \ref{sec5} is given by
 Corollary \ref{cor.derivatives}. It provides
 an explicit relation (in terms of $\mu$ and $q$) between
 the $m$-th derivative of an orthonormal
 polynomial of degree $k\geq m$  and the corresponding orthonormal
 polynomial of degree $k-m$. That is, it relates
 the orthonormal polynomial system, associated with some
 $f\sim\IPq$, to the corresponding orthonormal polynomial system
 associated with the `target' density $f_m\propto q^m f$.

 In the sequel and elsewhere in this article,
 $X\sim \IP$ means that $X$ has finite
 mean $\mu$, and that $X$ admits a density $f$
 (w.r.t.\ Lebesgue measure on $\R$) such
 that (\ref{qua1}) is fulfilled. Define the open
 (bounded or unbounded) interval
 \be
 \label{eq.J}
 J=J(X):=(\essinf(X),\esssup(X)).
 \ee
 %In other words,
 If $F$ is the distribution function of $X$ then
 $J=(\alpha_F,\omega_F)=(\alpha,\omega)$, say,
 where $\alpha_F:=\inf\{x:F(x)>0\}$,
 $\omega_F:=\sup\{x:F(x)<1\}$. It is clear that (\ref{qua1}) takes the form
 $0=0$ whenever $x=\rho$ is a zero of $q$ that lies outside the interval
 $(\alpha,\omega)$; thus, $f(\rho)$ may assume any value in this case.
 However, in order to be specific, we can redefine $f(\rho)=0$ at such points
 $\rho$, if any, without any loss of generality. Therefore, we shall use this
 convention
 through the whole article without any further reference to it.

 \section{A complete classification of the Integrated Pearson family}
 \label{sec2}
 We show in this section that the Integrated Pearson family contains
 six different types of distributions. These are classified
 in terms of the corresponding quadratic polynomial
 $q(x)=\delta x^2+\beta x+\gamma$
 and its discriminant $\Delta=\beta^2-4\delta\gamma$ as
 follows: Type $1$ (Normal-type, $\delta=\beta=0$);
 type $2$ (Gamma-type, $\delta=0$, $\beta\neq
 0$); type $3$ (Beta-type, $\delta<0$); type $4$
 (Student-type, $\delta>0$, $\Delta<0$); type
 $5$ (Reciprocal Gamma-type, $\delta>0$,
 $\Delta=0$); type $6$ (Snedecor-type,
 $\delta>0$, $\Delta>0$). The first
 three types (with $\delta\leq 0$)
 consist of the well-known Normal, Gamma and Beta random
 variables and their linear  transformations; the last
 three types (with $\delta>0$)
 consist of some less familiar distributions
 (see Table \ref{table densities}, below);
 they have finite moments up to order $1+\frac{1}{\delta}-\epsilon$
 (for any
 $\epsilon>0$) while $\E|X|^{1+1/\delta}=\infty$.
 The proposed classification is very similar to the one given by
 %\cite[Table 2, and pp.\ 294--296]{DZ}.
 Diaconis and Zabell (1991), Table 2 and pp.\ 294--296.

 We start with an easily verified proposition.
 \begin{prop}
 \label{prop.1}
 Let $X\sim \IPq$ and set $J=(\alpha,\omega)=(\essinf(X),\esssup(X))$.
 Then,
 \begin{itemize}
 \item[(i)]
 $f(x)$ is strictly positive for $x$ in $J$ and zero otherwise, i.e.,
 $\{x:f(x)>0\}=J$;
 \item[(ii)]
 $f\in C^{\infty}(J)$, that is, $f$ has derivatives
 of any order in $J$;
 \item[(iii)]
 $X$ is a (usual) Pearson random variable supported in $J$;
 \item[(iv)] $q(x)=\delta x^2+\beta x+\gamma>0$ for
 all $x\in J$;
 \item[(v)]
 if $\alpha>-\infty$ then $q(\alpha)=0$ and, similarly, if $\omega<+\infty$
 then $q(\omega)=0$;
 \item[(vi)]
 for any $\theta,c\in\R$ with $\theta\neq 0$, the random
 variable $\widetilde{X}:=\theta X+c\sim\IPtilde$
 with $\widetilde{\mu}=\theta\mu+c$ and
 $\widetilde{q}(x)=\theta^2 q((x-c)/\theta)$.
 \end{itemize}
 \end{prop}
 \begin{proof} By (\ref{qua1}), $x\mapsto q(x)f(x)$ is continuous. On the other hand,
 from the definition of $J=(\alpha_F,\omega_F)=(\alpha,\omega)$
 it follows that
 $q(x)f(x)$ must vanish for all $x\leq \alpha$ (if any) and for all
 $x\geq \omega$ (if any).
 Also, it must be strictly positive for $x\in J$. Indeed, if $x\in(\mu,\omega)$ then
 $q(x)f(x)=\int_{x}^{\infty}(t-\mu)f(t)\ud{t}\geq (x-\mu)(1-F(x))>0$;
 if $x\in(\alpha,\mu)$ then
 $q(x)f(x)=\int_{-\infty}^{x}(\mu-t)f(t)\ud{t}\geq (\mu-x)F(x)>0$; finally,
 $q(\mu)f(\mu)=\frac{1}{2}\E|X-\mu|>0$. Thus, $q(x)f(x)>0$
 for all $x\in (\alpha,\omega)$. Since $q$ is continuous
 and has no roots in $J$ it follows that both $q(x)$ and $f(x)$
 are strictly positive (and continuous) in $J$.
 The vanishing of $q f$ outside $J$ shows that $f(x)=0$
 for all $x\notin J$, with the possible exception at the points
 $x\notin J$ which are real roots of $q$.
 Clearly, if $\rho\in\R\smallsetminus (\alpha,\omega)$ is a zero of $q$
 we can redefine $f(\rho)=0$, if necessary, so that (i) and (iv) follow.
 On the other hand, $f:(\alpha,\omega)\to (0,\infty)$ is $C^{\infty}(J)$.
 Indeed, writing $p_1(x)=\mu-x-q'(x)$
 (a polynomial of degree at most one) we see from (\ref{qua1}) that
 $f:J\to(0,\infty)$
 is continuous and thus,
 \be
 \label{Pearson}
 f'(x)=f(x)\frac{p_1(x)}{q(x)} \ \ \mbox{or, equivalently,} \ \
 \frac{f'(x)}{f(x)}=\frac{\mu-x-q'(x)}{q(x)}, \ \ \ x\in J.
 \ee
 This proves (iii).
 Moreover, (\ref{Pearson}) shows that $f'$ is continuous in $J$ and, inductively,
 that $f^{(n+1)}:J\to \R$ is continuous, since for $x\in J$,
 \[
 f^{(n+1)}(x)=\sum_{j=0}^{n} {n \choose j} f^{(j)}(x)
 \left(\frac{p_1(x)}{q(x)}\right)^{(n-j)},\quad \ n=0,1,2,\ldots \ .
 \]
 Now (vi) is straightforward and it remains to show (v).
 To this end, assume that $\omega<\infty$. Since
 $q(\omega)=\lim_{x\nearrow\omega}q(x)$ and $q(x)>0$ for
 $x$ in a left neighborhood of $\omega$, it follows that
 $q(\omega)\geq 0$. Assume now that $q(\omega)>0$ and define
 \[
 \lambda_1:=\inf_{x\in[\mu,\omega]}\{q(x)\}>0, \ \ \
 \lambda_2:=\sup_{x\in[\mu,\omega]}|\mu-x-q'(x)|<\infty.
 \]
 Then, for all $x\in [\mu,\omega)$,
 \[
 \left|\int_{\mu}^{x} \frac{\mu-t-q'(t)}{q(t)}\ud{t}\right|\leq\int_{\mu}^{x}
 \frac{|\mu-t-q'(t)|}{q(t)}\ud{t}\leq\int_{\mu}^{\omega} \frac{|\mu-t-q'(t)|}{q(t)}\ud{t}
 \leq(\omega-\mu)\frac{\lambda_2}{\lambda_1}<\infty.
 \]
 Setting  $\lambda:= (\omega-\mu)\frac{\lambda_2}{\lambda_1}<\infty$
 and observing that
 \[
 \ln f(x)=\ln f(\mu)+\int_{\mu}^x \frac{f'(t)}{f(t)}\ud{t}
 =\ln f(\mu)+\int_{\mu}^x \frac{\mu-t-q'(t)}{q(t)}\ud{t},
 \ \ x\in[\mu,\omega),
 \]
 we have
 \[
 |\ln f(x)|\leq |\ln f(\mu)|+\lambda:=c<\infty,\quad \mu\leq x<\omega.
 \]
 Therefore, there exist constants $c_1, c_2$ such that
 $0<c_1\leq f(x)\leq c_2<\infty$ for all $x\in[\mu,\omega)$. Thus,
 \[
 q(\omega)=\lim_{x\nearrow\omega} q(x)
 =\lim_{x\nearrow\omega} \frac{1}{f(x)}\int_{x}^{\omega} (t-\mu)f(t) \ud{t}=0,
 \]
 which contradicts the assumption $q(\omega)>0$.
 The case $\alpha>-\infty$ is reduced to the case
 $\omega<\infty$ if we consider the random variable
 $\widetilde{X}=-X$ with mean $\widetilde{\mu}=-\mu$
 and support $J(\widetilde{X})=(\widetilde{\alpha},\widetilde{\omega})
 =(-\omega,-\alpha)$. According to (vi), its density $\widetilde{f}$ satisfies
 (\ref{qua1}) with quadratic
 polynomial $\widetilde{q}(x)=q(-x)$. Thus, if $\alpha>-\infty$ then
 $\widetilde{\omega}<\infty$ and
 $q(\alpha)=\widetilde{q}(-\alpha)=\widetilde{q}(\widetilde{\omega})=0$.
 \end{proof}
 \begin{cor}
 \label{cor.maximal}
 Let $X\sim \IPq$ and assume that $\alpha=\essinf(X)$ and
 $\omega=\esssup(X)$ are the lower and upper endpoints
 of the distribution function of $X$.
 Then, the support of $X$ (or of its density $f$)
 $S(f)=S(X):=\{x:f(x)>0\}$, equals to the open interval
 $J=J(X)=(\alpha,\omega)$. This interval support has the following
 two properties:
 \begin{itemize}
 \item[(i)]
    $J\subseteq S^+(q):=\{x:q(x)>0\}$ and
 \item[(ii)]
    $J$ is a maximal open interval contained in $S^+(q)$,
    i.e., for any open interval
    $\widetilde{J}\subseteq S^+(q)$ it is true that
    either $\widetilde{J}\subseteq J$ or $\widetilde{J}\cap
    J=\varnothing$.
 \end{itemize}
 \end{cor}
 In other words, the support $J$ of $X$ can be taken to be an open
 interval that
 coincides to some connected component of the open set
 $\{x:q(x)>0\}$. Since $q$ is a polynomial of degree at
 most two, it is clear that the set $\{x:q(x)>0\}$ has at
 most two connected components. For example, if $q(x)=x^2$
 then either $J=(-\infty,0)$ or $J=(0,\infty)$;
 if $q(x)=x^2-1$ then either $J=(-\infty,-1)$ or $J=(1,\infty)$;
 if $q(x)=1-x^2$ then $J=(-1,1)$; if $q(x)=x$ then
 $J=(0,\infty)$; if $q(x)=1+x^2$ or $q(x)\equiv 1$ then
 $J=\R$. Since, however, $\E X=\mu\in J$, any particular
 choice of $\mu\in \{x:q(x)>0\}$ characterizes the support
 $J$ of $X$. We say that $q(x)=\delta x^2+\beta x+\gamma$
 is {\it admissible} if there exists $\mu\in\R$ such that
 $\mu\in \{x:q(x)>0\}$; thus, $\{x:q(x)>0\}\neq\varnothing$
 whenever $q$ is admissible. In the sequel we shall show that for
 any admissible choice of $q$ and for any $\mu\in\{x:q(x)>0\}$
 there exists an absolutely continuous random
 variable $X$ with density $f$ such that $\E X=\mu$ and
 (\ref{qua1}) is fulfilled. Moreover, it will become clear that
 $f$ is characterized by the pair $(\mu;q)$. Therefore,
 the notation $X\sim \IPq$ or, equivalently, $f\sim \IPq$,
 has a well-defined meaning.

 The proposed classification distinguishes between the cases $\delta=0$,
 $\delta<0$ and $\delta>0$, as follows:
 %\vspace{-.3ex}

 \subsection[The case $\delta=0$]{The case $\bm{\delta=0}$}
 \label{ssec2.1}
 We have to further distinguish between the cases $\beta=0$ and $\beta\neq 0$.
 %\vspace{-.3ex}

 \subsubsection[The case $\delta=0$, $\beta=0$]{The subcase $\delta=0$, $\beta=0$}
 \label{sssec2.1.1}
 Since $q(x)\equiv\gamma$ and $q$ is admissible we must
 have $\gamma>0$. Therefore, $J(X)=\R$. Fixing $\mu\in\R$ and solving
 the differential equation (\ref{Pearson}) we get
 %\vspace{-.3ex}
 \[
 f(x)=\frac{1}{\sqrt{2\pi\gamma}}e^{-\frac{(x-\mu)^2}{2\gamma}},\quad x\in\R,
 %\vspace{-.3ex}
 \]
 i.e.\ $X\sim N(\mu,\sigma^2)$ with
 $\sigma^2=\gamma$.
 %\vspace{-.3ex}

 \subsubsection[The case $\delta=0$, $\beta\ne0$]{The subcase $\delta=0$, $\beta\ne0$}
 \label{sssec2.1.2}
 Assume that $q(x)=\beta x+\gamma$ with $\beta\neq 0$ and fix a number
 $\mu\in\{x:q(x)>0\}$; that is, $q(\mu)=\beta\mu+\gamma>0$.
 According to Proposition \ref{prop.1}(vi) we may further assume that
 $\beta>0$, $\gamma=0$ and $\mu>0$;
 otherwise, it suffices to consider the random variable
 $\widetilde{X}=\frac{\beta}{|\beta|}(X+\frac{\gamma}{\beta})$
 with $\widetilde{q}(x)=|\beta|x$ and
 $\E \widetilde{X}=\widetilde{\mu}=\frac{\beta}{|\beta|}(\mu+\frac{\gamma}{\beta})
 =\frac{q(\mu)}{|\beta|}>0$
 since $q(\mu)>0$.
 Now, since $q(x)=\beta x$ with $\beta>0$ we must have $J(X)=(0,\infty)$. Fixing $\mu>0$
 and solving the differential equation (\ref{Pearson}) we
 get
 %\vspace{-.3ex}
 \[
 f(x)=\frac{(1/\beta)^{\mu/\beta}}{\varGamma(\mu/\beta)}x^{\mu/\beta-1}e^{-x/\beta},\quad x>0.
 %\vspace{-.3ex}
 \]
 That is, $X\sim\varGamma(a,\lambda)$ with
 $a=\mu/\beta>0$ and $\lambda=1/\beta>0$. Hence, a linear
 non-constant $q$ corresponds to a linear transformation,
 $\widetilde{X}=\theta X+c$, $\theta\neq 0$,
 of a Gamma random variable $X$, i.e., to {\it Gamma-type}
 distributions.
 %\vspace{-.3ex}

 \subsection[The case $\delta<0$]{The case $\bm{\delta<0}$}
 \label{ssec2.2}
 Since $\delta<0$ and $\{x:q(x)>0\}$ must contain some
 interval it follows that the discriminant
 $\beta^2-4\delta\gamma$ of $q$ must be strictly positive.
 If $\rho_1<\rho_2$ are the real roots of $q$ we can write
 $q(x)=\delta (x-\rho_1)(x-\rho_2)$ so that the support of $X$ is
 the finite interval $J(X)=(\rho_1,\rho_2)$. Now we show
 that for any choice of $\mu\in(\rho_1,\rho_2)$ there exist
 a (unique) random variable $X$ with $X\sim\IPq$. To this end, it
 suffices to examine the particular case $q(x)=-\delta
 x(1-x)$ and $0<\mu<1$;
 the general case is reduced to the particular
 one if we consider the random variable $\widetilde{X}=(X-\rho_1)/(\rho_2-\rho_1)$.
 Fixing $\mu\in(0,1)$, $q(x)=-\delta x(1-x)$ and solving
 the differential equation (\ref{Pearson}) on $J(X)=(0,1)$
 we get
 %\vspace{-.3ex}
 \[
 f(x)=\frac{1}{B(-\mu/\delta,-(1-\mu)/\delta)}x^{-\mu/\delta-1}(1-x)^{-(1-\mu)/\delta-1},\quad 0<x<1,
 %\vspace{-.3ex}
 \]
 that is, $X\sim B(a,b)$ with $a=\mu/|\delta|>0$,
 $b=(1-\mu)/|\delta|>0$. It follows that the case
 $\delta<0$ corresponds to a linear transformation
 of a Beta random variable, the {\it Beta-type} distributions.
 %\vspace{-.3ex}

 \subsection[The case $\delta>0$]{The case $\bm{\delta>0}$}
 \label{ssec2.3}
 We have to further distinguish between the cases where
 the discriminant $\Delta=\beta^2-4\delta\gamma$ is
 positive, zero or negative.
 %\vspace{-.3ex}

 \subsubsection[The subcase $\delta>0$, $\Delta<0$]{The subcase $\delta>0$, $\Delta<0$}
 \label{sssec2.3.1}
 Since $q$ has no real roots, $J(X)=\R$. Thus, $\mu\in\R$ can
 take any arbitrary value. Also, $q$ has the form
 $q(x)=\delta (x-c)^2+\theta$
 with $\delta>0$, $\theta>0$ and $c\in\R$. Without loss of
 generality we further assume that $c=0$; otherwise we can
 consider the random variable $\widetilde{X}=X-c$. Fixing
 $\mu\in\R$, $q(x)=\delta x^2+\theta$
 and solving (\ref{Pearson}) one finds that
 \vspace{-.3ex}
 \[
 f(x)=\frac{C}{(\delta x^2+\theta)^{1+\frac{1}{2\delta}}}
 \exp\left(\frac{\mu}{\sqrt{\delta\theta}}
 \tan^{-1}(x\sqrt{\delta/\theta})\right),
 \quad x\in\R.
 \vspace{-.3ex}
 \]
 The normalizing
 constant $C=C_{\mu}(\delta,\theta)$ can be calculated explicitly when $\mu=0$:
 %\vspace{-.3ex}
 \[
 C_0(\delta,\theta)
 =\frac{\varGamma(1+1/(2\delta))
 \sqrt{\delta\theta^{1+1/\delta}}}{\varGamma(1/2+1/(2\delta))\sqrt{\pi}}.
 %\vspace{-.3ex}
 \]

 Therefore, the quadratic polynomial $q(x)=\delta(x-c)^2+\theta$ with
 $\delta>0$ and $\theta>0$  corresponds to {\it
 Student-type} distributions centered at $c$, provided that $\mu=c$;
 otherwise, i.e., when $\mu\neq c$, it corresponds to some asymmetric, say
 {\it skew Student-type}, distributions.
 %\vspace{-.3ex}

 \subsubsection[The subcase $\delta>0$, $\Delta=0$]{The subcase $\delta>0$, $\Delta=0$}
 \label{sssec2.1.2}
 Since $q$ has a unique real root at $\rho=-\beta/(2\delta)$,
 it follows that $q(x)=\delta (x-\rho)^2$ and, therefore,
 the support $J(X)$ is either $(-\infty,\rho)$
 or $(\rho,\infty)$, according to $\mu<\rho$ or $\mu>\rho$, respectively.
 Without loss of generality we may assume that $q(x)=\delta
 x^2$ with $\delta>0$ and $\mu>0$; otherwise, it suffices to
 consider the random variable
 $\widetilde{X}=\frac{\mu-\rho}{|\mu-\rho|}(X-\rho)$. Now,
 setting $J(X)=(0,\infty)$, $q(x)=\delta x^2$ ($\delta>0$) and $\mu>0$
 in eq.\ (\ref{Pearson}) we get the solution
 %\vspace{-.4ex}
 \[
 f(x)=\frac{\lambda^a}{\varGamma(a)}x^{-a-1}e^{-\lambda/x},\quad x>0,
 %\vspace{-.4ex}
 \]
 where $\lambda=\mu/\delta>0$ and $a=1+1/\delta>1$.
 Observing
 that $1/X\sim\varGamma(a,\lambda)$ it follows that the case
 $\delta>0$, $\Delta=0$ corresponds to
 {\it Reciprocal Gamma-type} distributions.
 %\vspace{-.3ex}

 \subsubsection[The subcase $\delta>0$, $\Delta>0$]{The subcase $\delta>0$, $\Delta>0$}
 \label{sssec2.1.3}
 Assuming that $\rho_1<\rho_2$ are the roots of $q$ we can
 write $q(x)=\delta (x-\rho_1)(x-\rho_2)$ and the support
 $J(X)$ has to be either $(-\infty,\rho_1)$ or
 $(\rho_2,\infty)$, according to $\mu<\rho_1$ or $\mu>\rho_2$, respectively.
 By considering the random variable
 $\widetilde{X}=-(X-\rho_1)$ when $\mu<\rho_1$ and the
 random variable $\widetilde{X}=X-\rho_2$ when $\mu>\rho_2$
 it is easily seen that both cases reduce to
 $\widetilde{\mu}>0$, $J(\widetilde{X})=(0,\infty)$ and
 $\widetilde{q}(x)=\delta x(x+\theta)$ with $\delta>0$
 and $\theta=\rho_2-\rho_1>0$. Thus, there is no loss of
 generality in assuming $\mu>0$, $J(X)=(0,\infty)$ and
 $q(x)=\delta x (x+\theta)$ with $\delta>0$ and $\theta>0$.
 Then, (\ref{Pearson}) yields
 %\vspace{-.4ex}
 \[
 f(x)=\frac{1}{B(a,b)}\theta^{a}x^{b-1}(x+\theta)^{-a-b},\quad x>0,
 %\vspace{-.4ex}
 \]
 with $a=1+\frac{1}{\delta}>1$ and $b=\frac{\mu}{\delta\theta}>0$.
 Equivalently, $\frac{\theta}{X+\theta}\sim B(a,b)$.
 It follows that the case $\delta>0$, $\Delta>0$
 corresponds to {\it Snedecor-type} distributions.
 \medskip

 All the above possibilities are summarized in Table
 \ref{table densities}, below; compare
 with Table 2, p.\ 296,
 in Diaconis and Zabell (1991).
 %\cite[Table 2, p.\ 296]{DZ}.

 \begin{sidewaystable}
 \centering
 \caption[Densities of the Integrated Pearson family]{Densities
 of the Integrated Pearson family
 %\vspace*{5em}
 $\IP\equiv\IPq$.\textcolor[rgb]{0,0,1}{$^*$}}
 \label{table densities}
 \footnotesize
 \begin{tabular}{@{\hspace{0ex}}c@{\hspace{5ex}}c@{\hspace{5ex}}c@{%
 \hspace{5ex}}c@{\hspace{5ex}}c@{\hspace{5ex}}c@{\hspace{5ex}}r@{\hspace{0ex}}}
 \\
 \hline

 \hline

 \hline
 \xA{\bf type}{\bf usual notation} & \bf density $\bm{f(x)}$& \bf
 support & $\bm{q(x)}$ & \bf parameters
 & \bf mean $\bm{\mu}$ & \xB{\bf{\vspace{1em}classification\\{rule\hspace{3.5ex}\vspace{1em}\,}}}
 \\
 \hline
 \\
 % NORMAL-TYPE
 \xA{{\bf 1.} Normal-type}{$X\sim{N(\mu,\sigma^2)}$}
 & $\frac{1}{\sigma\sqrt{2\pi}}e^{-\frac{(x-\mu)^2}{2\sigma^2}}$
 & $\R$ & $\sigma^2$ & $\gamma=\sigma^2>0$ & $\mu\in\R$ & $\delta=\beta=0$
 \\
 [3ex]
 % GAMMA-TYPE
 \xA{{\bf 2.} Gamma-type}{$X\sim\varGamma(a,\lambda)$}
 & $\frac{\lambda^a}{\varGamma(a)}x^{a-1}e^{-\lambda{x}}$
 & $(0,+\infty)$ & $\ds\frac{x}{\lambda}$ & $a$, $\lambda>0$
 & $\ds\frac{a}{\lambda}>0$ & $\delta=0$, $\beta\ne0$
 \\
 [3ex]
 % BETA-TYPE
 \xA{{\bf 3.} Beta-type}{$X\sim B(a,b)$}
 & $\frac{x^{a-1}(1-x)^{b-1}}{B(a,b)}$  & $(0,1)$
 & $\ds\frac{x(1-x)}{a+b}$ & $a$, $b>0$ & $\ds\frac{a}{a+b}>0$
 & $\delta=\ds\frac{-1}{a+b}<0$
 \\
 [3ex]
 % STUDENT-TYPE
 \xA{{\bf 4.} Student-type}{\footnotemark\hspace{5ex}\,}
 & $\frac{C\exp\left(\frac{\mu\tan^{-1}(x\sqrt{{\delta}/{\gamma}})}
 {\sqrt{\delta\gamma}}\right)}{(\delta x^2+\gamma)^{1+\frac{1}{2\delta}}}\footnotemark$
 & $\R$ & $\delta{x^2}+\gamma$ & $\delta$, $\gamma>0$ & $\mu\in\R$
 & \xB{$\delta>0$\\$\beta^2<4\delta\gamma$\vspace{.5em}}
 \\
 [3ex]
 % RECIPROCAL GAMMA-TYPE
 \xA{{\bf 5.} Reciprocal}{Gamma-type\,\,\,\,\,\,\,\,\ }
 & $\frac{\lambda^a}{\varGamma(a)}x^{-a-1}e^{-\frac{\lambda}{x}}$
 & $(0,+\infty)$ & $\ds\frac{x^2}{a-1}$ & $a>1$, $\lambda>0$
 & $\ds\frac{\lambda}{a-1}>0$ & \xB{$\delta=\frac{1}{a-1}>0$
 \\
 $\beta^2=4\delta\gamma$\\$\frac{1}{X}\sim\varGamma(a,\lambda)$\vspace{1em}}
 \\
 [4ex]
 % SNEDECOR GAMMA-TYPE
 \xA{{\bf 6.} Snedecor-type}{\footnotemark\hspace{5ex}\,}
 & $\frac{\theta^a}{B(a,b)}x^{b-1}(x+\theta)^{-a-b}$
 & $(0,+\infty)$ & $\ds\frac{x(x+\theta)}{a-1}$ & ${\ds a>1,\atop\ds b,\ \theta>0}$
 & $\ds\frac{b\theta}{a-1}>0$ & \xB{$\delta=\frac{1}{a-1}>0$
 \\
 $\beta^2>4\delta\gamma$\\$\frac{\theta}{X+\theta}\sim B(a,b)$\vspace{1em}}
 \\
 \hline

 \hline
 \end{tabular}
 \footnotetext{\!\!\!\!$^*$ \scriptsize
 A
 random variable $X$ belongs to the Integrated Pearson family
 if and only if there exist constants $c_1\neq 0$ and $c_2\in\R$
 such that the density of $\widetilde{X}=c_1 X+c_2$ is contained
 in the
 \vspace{.7em}
 table.}
 \footnotetext{\!\!\!\!$^1$ \scriptsize For $n>1$ and if $\mu=0$ and
 $\delta=\frac{1}{n-1}=\frac{\gamma}{n}$ then $X\sim{t}_n$.}
 \footnotetext{\!\!\!\!$^2$ \scriptsize $C=C_{\mu}(\delta,\gamma)>0$,
 with $C_{0}(\delta,\gamma)
 ={\varGamma\left(1+\frac{1}{2\delta}\right)
 \sqrt{\delta\gamma^{1+\frac{1}{\delta}}}}\Big/{\varGamma\left(\frac12
 +\frac{1}{2\delta}\right)\sqrt{\pi}}$.}
 \footnotetext{\!\!\!\!$^3$ \scriptsize For $n>0$, $m>2$ and if
 $a=\frac{m}{2}$, $b=\frac{n}{2}$, $\theta=\frac{m}{n}$ then
 $X\sim F_{n,m}$.}
 \end{sidewaystable}

 \begin{rem}
 \label{rem.extra}
 Since
 \[
 (\mu-x)\frac{\exp\left(\frac{\mu}
 {\sqrt{\delta\gamma}}
 \tan^{-1}(x\sqrt{{\delta}/{\gamma}})
 \right)}{(\delta x^2+\gamma)^{1+\frac{1}{2\delta}}}
 =
 \frac{\ud}{\ud x}
 \frac{\exp\left(\frac{\mu}{\sqrt{\delta\gamma}}
 \tan^{-1}(x\sqrt{{\delta}/{\gamma}})\right)}
 {(\delta x^2+\gamma)^{\frac{1}{2\delta}}},
 \]
 it follows that $\E X=\mu$
 for the Student-type densities (type 4), while
 for all other cases it is evident to check that
 the mean is as displayed in Table
 \ref{table densities}. Next,
 it is easily verified that the densities
 of Table
 \ref{table densities} satisfy the
 assumptions (B) of Proposition
 \ref{prop.2}, below,
 with $\mu=\E X$, $p_2(x)=q(x)$ and
 $p_1(x)=\mu-x-q'(x)$, where $\mu$ and $q$
 are as in the Table.
 Hence, according to
 Proposition \ref{prop.2}, all these densities
 are, indeed, integrated Pearson.
 \end{rem}
 \begin{cor}
 \label{cor.moments}
 Assume that $X\sim\IP$.
 \vspace*{-1ex}
 \begin{itemize}
 \item[(a)] If $\delta\leq 0$ then $\E|X|^\alpha<\infty$ for any $\alpha\in[0,\infty)$.
 \vspace*{-1ex}
 \item[(b)] If $\delta>0$ then $\E|X|^\alpha<\infty$ for any
 $\alpha\in[0,1+1/\delta)$,
 while $\E|X|^{1+1/\delta}=\infty$.
 \vspace*{-.5ex}
 \end{itemize}
 \end{cor}
 \begin{proof}
 If $X\sim\IP$ then we can find constants $c_1\neq 0$ and
 $c_2\in\R$ such that the density of $\widetilde{X}=c_1 X+c_2$
 is contained in Table \ref{table densities}. Then, according to Proposition
 \ref{prop.1}(vi), $\widetilde{X}\sim\IPqtilde$
 with $\widetilde{\delta}=\delta$. The assertion
 follows from the fact that $\E|X|^{\alpha}<\infty$ if and only if
 $\E|c_1 X+c_2|^{\alpha}<\infty$.
 \end{proof}

 Next, we shall obtain a recurrence for the moments and the
 central moments of a random variable $X\sim\IPq$. To this
 end we first prove a simple lemma.
 \begin{lem}
 \label{lem.limits2}
 If $X\sim\IP$ has support $J(X)=(\alpha,\omega)$
 and $\E |X|^n<\infty$ for some $n\geq 1$ (that is, $\delta<\frac{1}{n-1}$)
 then
 \be
 \label{limits2}
 \lim_{x\nearrow\omega}{x^k q(x)f(x)}
 =\lim_{x\searrow\alpha}{x^k q(x)f(x)}=0,\quad
 k=0,1,\ldots,n-1,
 \ee
 and, in general, for any $c\in\R$,
 \be
 \label{limits3}
 \lim_{x\nearrow\omega}{(x-c)^k q(x)f(x)}
 =\lim_{x\searrow\alpha}{(x-c)^k q(x)f(x)}=0,\quad
 k=0,1,\ldots,n-1.
 \ee
 \end{lem}
 \begin{proof}
 Since $x^k q(x)f(x)=x^k\int_{\alpha}^x (\mu-t)f(t)\ud{t}$,
 $\alpha<x<\omega$, the second limit in (\ref{limits2}) is trivial whenever
 $\alpha>-\infty$ and the first one is trivial whenever $\omega<\infty$.
 If $\omega=\infty$ it suffices to verify the first limit
 in (\ref{limits2}) only when $k=n-1$ and $n\geq 2$ (because the case $k=0$ is obvious);
 then,
 since $q(x)f(x)$ is eventually decreasing we have that for
 large enough $x>0$,
 \[
 \begin{split}
 x^{n-1}q(x)f(x)&=q(x)f(x)\frac{(n-1)2^{n-1}}{2^{n-1}-1}\int_{x/2}^x t^{n-2} \ud{t}
 \\
 &\leq \frac{(n-1)2^{n-1}}{2^{n-1}-1}\int_{x/2}^x t^{n-2}q(t)f(t) \ud{t}\\
 &\leq \frac{(n-1)2^{n-1}}{2^{n-1}-1}\int_{x/2}^\infty t^{n-2}q(t)f(t) \ud{t}
 \to 0,\ \ \textrm{as} \ x\to\infty,
 \end{split}
 \]
 because $\deg(q)\leq 2$ and, by assumption,
 $\E q(X)|X|^{n-2}<\infty$. The case $\alpha=-\infty$ is translated to
 the previous one by considering the random variable
 $\widetilde{X}=-X$ with density $\widetilde{f}(x)=f(-x)$, quadratic polynomial
 $\widetilde{q}(x)=q(-x)$ and support
 $J(\widetilde{X})=(\widetilde{\alpha},\widetilde{\omega})=(-\omega,-\alpha)=(-\omega,\infty)$.
 Then $\E|\widetilde{X}|^n=\E|X|^n<\infty$ and
 \[
 \lim_{x\to-\infty}x^{k}q(x)f(x)
 =(-1)^{k}\lim_{x\to\infty}x^{k}q(-x)f(-x)=(-1)^{k}\lim_{x\to\infty}x^{k}\widetilde{q}(x)
 \widetilde{f}(x)=0
 \]
 for all $k\in\{0,1,\ldots,n-1\}$. Now it suffices to observe that all
 limits in (\ref{limits3}) are linear combinations of
 limits in (\ref{limits2}). Indeed, the first limit in
 (\ref{limits3}) is
 \medskip

 \centerline{
 $\lim_{x\nearrow\omega}{(x-c)^k q(x)f(x)}=\sum_{i=0}^k
 {k\choose i}(-c)^{k-i} \lim_{x\nearrow\omega}{x^i q(x)f(x)}=0$}
 \medskip

 \noindent
 and, similarly, the second limit in (\ref{limits3}) is
 \medskip

 \hfill
 $\lim_{x\searrow\alpha}{(x-c)^k q(x)f(x)}=\sum_{i=0}^k
 {k\choose i}(-c)^{k-i} \lim_{x\searrow\alpha}{x^i q(x)f(x)}=0$.
 \end{proof}
 \begin{lem}
 \label{lem.req.moments}
 If $X\sim\IP$
 %has support $J(X)=(\alpha,\omega)$
 and $\E |X|^n<\infty$ for some $n\geq 2$ (that is, $\delta<\frac{1}{n-1}$)
 then for any $c\in\R$, the central moments about $c$ satisfy the recurrence
  \be
 \label{reccur.general}
 \begin{split}
  \E (X-c)^{k+1}=\frac{(\mu-c+ k q'(c))\E (X-c)^{k}+ k q(c) \E (X-c)^{k-1}}{1-k\delta},&\\
 \textrm{\raisebox{1ex}{$\quad k=1,2,\ldots,n-1,$}}
 \end{split}
  \ee
 with initial conditions $\E(X-c)^0=1$, $\E (X-c)^1=\mu-c$,
 where $q(c)=\delta c^2+\beta c+\gamma$, $q'(c)=2\delta c+\beta$.
 In particular,
 %%%\pagebreak
 \begin{itemize}
 \item[(i)] the usual moments ($c=0$) satisfy the recurrence
  \be
  \label{reccur.usual}
  \E X^{k+1}=\frac{(\mu+k\beta)\E X^{k}+k\gamma
 \E X^{k-1}}{1-k\delta}, \ \ k=1,2,\ldots,n-1,
  \ee
 with initial conditions $\E X^0=1$ and $\E X^1=\mu$;
 \item[(ii)] the central moments ($c=\mu$) satisfy the recurrence
  \be
 \label{reccur.central}
 \begin{split}
  \!\!\!
  \E (X-\mu)^{k+1}=\frac{k q'(\mu)\E (X-\mu)^{k}+ k q(\mu) \E (X-\mu)^{k-1}}{1-k\delta},
 \ \ k=1,2,\ldots,n-1,
 \end{split}
  \ee
 with initial conditions $\E(X-\mu)^0=1$ and
 $\E (X-\mu)^1=0$.
 \end{itemize}
 \end{lem}
 \begin{proof}
 If $J(X)=(\alpha,\omega)$ is the support of $X$ and $k\in\{1,2,\ldots,n-1\}$
 then we have
 \[
 \begin{split}
 \E(X-c)^{k+1}&=\E[((\mu-c)-(\mu-X))(X-c)^{k}]\\
              &=(\mu-c)\E(X-c)^{k}-\int_{\alpha}^{\omega}(x-c)^{k}(\mu-x)f(x)\ud{x}.
 \end{split}
 \]
 Using (\ref{limits3}) and the fact that $q(X)=\delta(X-c)^2+q'(c)(X-c)+q(c)$ we see that
 \[
 \begin{split}
 -\int_{\alpha}^{\omega} (x-c)^{k}(\mu-x)f(x)\ud{x}
 &= -\int_{\alpha}^{\omega} (x-c)^{k}(q(x)f(x))'\ud{x}
 \\
 &= -(x-c)^{k}q(x)f(x)\big|^\omega_\alpha+k\E q(X)(X-c)^{k-1}
 \\
 &= k\delta\E(X-c)^{k+1}+ k q'(c)\E(X-c)^{k}+k q(c)\E(X-c)^{k-1}.
 \end{split}
 \]
 Therefore,
 \[
 \begin{split}
 (1-k\delta)\E(X-c)^{k+1}=(\mu-c+ k q'(c))\E(X-c)^{k}+ k
 q(c)\E(X-c)^{k-1},\\ \quad k=1,2,\ldots,n-1,
 \end{split}
 \]
 and, since the initial conditions are obvious, (\ref{reccur.general}) follows.
 %For (i) write $\E X^k=\mu \int_{\alpha}^{\omega}
 %x^{k-1}f(x)\ud{x} -
 %\int_{\alpha}^{\omega} x^{k-1}(\mu-x)f(x)\ud{x}=
 %\mu\E X^{k-1}-I_k(\mu;q)$.
 %In view of (\ref{limits2}) we get
 %$\E X^k=\mu\E X^{k-1}-I_k(\mu;q)=\mu\E X^{k-1}-\int_{\alpha}^{\omega} x^{k-1}(q(x)f(x))'\ud{x}=
 %\mu\E X^{k-1}-x^{k-1}q(x)f(x)\left|^\omega_\alpha\right.
 %+(k-1)\int_{\alpha}^{\omega} x^{k-2}q(x)f(x)\ud{x}=
 %\mu\E X^{k-1}+(k-1)\int_{\alpha}^{\omega} x^{k-2}(\delta x^2+\beta x+\gamma)f(x)\ud{x}
 %=\mu\E X^{k-1}+(k-1)(\delta \E X^{k}+\beta \E X^{k-1}+\gamma
 %\E X^{k-2})$ and
 %%. Substituting the last expression of $-I_k(\mu;q)$ into
 %%$\E X^k=\mu\E X^{k-1}-I_k(\mu;q)$ we get
 %(\ref{reccur.usual}) follows. For (ii) we have, in view of (\ref{limits3}), that
 %$\E(X-\mu)^k=-\int_{\alpha}^{\omega}
 %(x-\mu)^{k-1}(\mu-x)f(x)\ud{x}
 %=-\int_{\alpha}^{\omega} (x-\mu)^{k-1}(q(x)f(x))'\ud{x}
 %=-(x-\mu)^{k-1}q(x)f(x)\left|^\omega_\alpha\right.
 %+(k-1)\int_{\alpha}^{\omega} (x-\mu)^{k-2}q(x)f(x)\ud{x}
 %=(k-1)\int_{\alpha}^{\omega} (x-\mu)^{k-2}q(x)f(x)\ud{x}=
 %(k-1)\E [q(X)(X-\mu)^{k-2}]$.
 %Writing $q(X)=q(\mu)+q'(\mu)(X-\mu)+\delta(X-\mu)^2$
 %we observe that
 %$\E(X-\mu)^k=(k-1)(q(\mu)\E(X-\mu)^{k-2}+q'(\mu)\E(X-\mu)^{k-1}+\delta\E(X-\mu)^{k})$,
 %which is (\ref{reccur.central}).
 \end{proof}

 \section{Comparison with the ordinary Pearson system}
 \label{sec3}
 The ordinary Pearson family consists of absolutely continuous random
 variables $X$ supported in some (open) interval $(\alpha,\omega)$, such
 that their density $f$, which is assumed strictly positive and differentiable
 in $(\alpha,\omega)$, satisfies the {\it Pearson differential equation}
 \be
 \label{Pearson.ord}
 \frac{f'(x)}{f(x)}=\frac{p_1(x)}{p_2(x)}, \ \ \ \alpha<x<\omega,
 \ee
 where $p_1$ is a polynomial of degree at most one and $p_2$
 is a polynomial of degree at most two. Since we can multiply the
 nominator and the denominator of (\ref{Pearson.ord}) by the same
 nonzero
 constant, it is usually assumed, for convenience, that $p_1$
 is a monic linear polynomial
 of degree one, e.g., $p_1(x)=x+{a}_0$.  Although this
 restriction specifies both $p_1$ and $p_2$ whenever $p_1$ is non-constant,
 it is not satisfactory for our purposes because it eliminates all rectangular
 (uniform over some interval) distributions and several $B(a,b)$ densities
 (those with $a+b=2$)  -- see Table \ref{table densities}, above. Therefore,
 when we say that a function $f$ satisfies the Pearson differential equation
 (\ref{Pearson.ord}) it will be assumed that $p_1$ is {\it any polynomial} of
 degree at most one (the cases $p_1\equiv 0$ and $p_1\equiv c\neq 0$ are allowed)
 and $p_2\not\equiv0$ is {any polynomial of degree at most two}. Note that common
 zeros of $p_1$ and $p_2$ are allowed inside the interval $(\alpha,\omega)$.
 Also, it may happen that $p_1$ and $p_2$  have common zeros outside the interval
 $(\alpha,\omega)$; this is the case of an exponential density.

 Clearly, the ordinary Pearson family contains some random
 variables whose expectation does not exist, e.g.,
 Cauchy. Sometimes it is asserted that,
 under finiteness of the first moment,
 (\ref{qua1}) and (\ref{Pearson.ord}) are equivalent --
 see, e.g.,
 %%%\cite[pp.\ 292--293]{Korw}.
 Korwar (1991), pp.\ 292--293.
 However, this is true only in particular cases,
 i.e.\ when we have made the `correct' choice of $p_2$
 and provided that a solution $f$ of (\ref{Pearson.ord})
 is considered in a maximal subinterval of the support of
 $p_2$, $\{x:p_2(x)\neq 0\}$.
 The following algorithmic procedure will
 always decides correctly if a given Pearson density belongs
 to the Integrated Pearson family.
 The algorithm makes a correct choice of $p_2$,
 if it exists,
 as follows:

 \newpage
 %%\begin{table}[htdp]
 \begin{center}
 %%\caption{\bf The Integrated Pearson Algorithm}
 {\bf The Integrated Pearson Algorithm}
 \centering
 \begin{description}
 \item[Step 0.] Assume that a Pearson density
   $f$ with finite (unknown) mean and (known) support
   $S(f)=\{x:f(x)>0\}=(\alpha,\omega)$
   satisfies $f'/f=\widetilde{p}_1/\widetilde{p}_2$
   for given (real)
   polynomials $\widetilde{p}_1$, $\widetilde{p}_2$ (with
   $\widetilde{p}_2\not\equiv0$), of degree at most one and two,
   respectively.
 \item[Step 1.] Cancel the common factors of $\widetilde{p}_1$ and
   $\widetilde{p}_2$, if any. Then the resulting
   polynomials, say
   $\widetilde{p}_{1}^{(1)}$ and
   $\widetilde{p}_{2}^{(1)}$,
   have become irreducible --
   they do not have any common zeros in $\C$.
   In case $\widetilde{p}_1\equiv 0$ it suffices
   to define
   $\widetilde{p}_{1}^{(1)}\equiv 0$,
   $\widetilde{p}_{2}^{(1)}\equiv 1$.
 \item[Step 2.] If $\alpha>-\infty$ and
   $\widetilde{p}_2^{(1)}(\alpha)\neq 0$
   then multiply both $\widetilde{p}_1^{(1)}$
   and $\widetilde{p}_2^{(1)}$ by $x-\alpha$
   and name the resulting polynomials
   $\widetilde{p}_1^{(2)}$
   and $\widetilde{p}_2^{(2)}$; otherwise (i.e.\ if
   either $\alpha=-\infty$ or
   $\alpha>-\infty$ and $\widetilde{p}_2^{(1)}(\alpha)=0$)
   set $\widetilde{p}_1^{(2)}=\widetilde{p}_1^{(1)}$
   and $\widetilde{p}_2^{(2)}=\widetilde{p}_2^{(1)}$.
 \item[Step 3.] If $\omega<\infty$ and
   $\widetilde{p}_2^{(2)}(\omega)\neq 0$
   then multiply both $\widetilde{p}_1^{(2)}$
   and $\widetilde{p}_2^{(2)}$ by $\omega-x$
   and name  the  resulting polynomials $p_1$
   and $p_2$; otherwise (i.e.\ if either
    $\omega=\infty$ or $\omega<\infty$ and
   $\widetilde{p}_2^{(2)}(\omega)=0$)
   set $p_1=\widetilde{p}_1^{(2)}$
   and $p_2=\widetilde{p}_2^{(2)}$.
 \item[Step 4.] If the resulting polynomials $p_1$ and $p_2$
   satisfy the conditions
   $\deg(p_1)\leq 1$ and $\deg(p_2)\leq 2$
   then $p_2$ is a correct choice and
   $f\sim\IPq$ with $q(x)=\theta p_2(x)$ for some
   $\theta\neq0$; otherwise the given density $f$
   does not belong to the Integrated Pearson system.
 %\item[Step 4.] For the resulting $p_2$ check if the
 %  limiting conditions
 %  $\lim_{x\searrow\alpha}p_2(x)f(x)=0$ and
 %  $\lim_{x\nearrow\omega}p_2(x)f(x)=0$ are satisfied.
 %  If these limits are true then this density is Integrated
 %  Pearson, otherwise it is not.
 \end{description}
 \end{center}
 %\end{table}

 It is clear that the above procedure
 starts with the equation $f'/f=\widetilde{p}_1/\widetilde{p}_2$
 and, at Step 3, it produces two new (real) polynomials
 $p_1,p_2$, of degree at most three and four,
 respectively, such that $f'/f=p_1/p_2$. Moreover,
 the polynomial $p_2$ satisfies the relations $p_2(\alpha)=0$
 if $\alpha>-\infty$, $p_2(\omega)=0$ if $\omega<\infty$
 and $p_2(x)\neq 0$ for all $x\in(\alpha,\omega)$.
 Furthermore, because of Step 1, the polynomials $p_1(z)$ and $p_2(z)$
 cannot have any common zeros in $\C\smallsetminus\{\alpha,\omega\}$.

 The algorithm guarantees that we have
 chosen a correct $p_2$ in each case where
 such a $p_2$ exists. For example, the standard
 exponential density,
 \[
 f(x)=e^{-x}, \ \ x>0,
 \]
 satisfies
 (\ref{Pearson.ord}) when $(p_1,p_2)=(-1,1)$,
 when $(p_1,p_2)=(-x,x)$ and when  $(p_1,p_2)=(-x-1,x+1)$;
 the correct choice is the second one. The
 standard uniform density,
 \[
 f(x)=1, \ \ 0<x<1,
 \]
 satisfies (\ref{Pearson.ord}) for $p_1\equiv 0$
 and for any $p_2$ (with no roots in $(0,1)$), and the
 correct choice is $p_2=x(1-x)$. The power density,
 \[
 f(x)=2x, \ \ 0<x<1,
 \]
 satisfies (\ref{Pearson.ord})
 with $(p_1,p_2)=(2-x,x(2-x))$ and the correct choice
 arises when we multiply both polynomials by $(1-x)/(2-x)$,
 that is, $(p_1,p_2)=(1-x, x(1-x))$.
 The Pareto density,
 \[
 f(x)=\frac{2}{(x+1)^{3}}, \ \ x>0,
 \]
 satisfies (\ref{Pearson.ord})
 when $(p_1,p_2)=(-3,x+1)$,
 when $(p_1,p_2)=(-3x,x(x+1))$
 and when  $(p_1,p_2)=(-3(x+1),(x+1)^2)$;
 the correct choice is the second one.
 The half-Normal density,
 \[
 f(x)=\sqrt{\frac{2}{\pi}}e^{-x^2/2}, \ \ x>0,
 \]
 satisfies (\ref{Pearson.ord}) in its interval support $(\alpha,\omega)=(0,\infty)$,
 although it does not satisfy (\ref{qua1}) -- there not exists a correct choice for
 $p_2$.
 A more natural example is as follows:
 Consider the density
 \[
 f(x)=\frac{C}{\sqrt{1+x^2}}, \ \ \alpha<x<\omega,
 \]
 where $C=C(\alpha,\omega)>0$
 is the normalizing constant.
 This density satisfies, in any finite interval
 $(\alpha,\omega)$, the Pearson differential equation
 (\ref{Pearson.ord})
 with $p_1=-x$, $p_2=1+x^2$,  while
 its integral over unbounded intervals diverges.
 This density does not fulfill (\ref{qua1})
 and thus, it does not belong to the Integrated Pearson family --
 again there does not exist
 a correct choice for $p_2$.

 The algorithm is justified because
 of the following propositions.

 \begin{prop}
 \label{prop.alg1}
 Let $X\sim f$ and assume that the density $f$ satisfies the
 assumptions of Step 0. If $X\sim \IPq$ then the
 polynomials
 $p_1$ and $p_2$ of Step 3 are of degree at most one and
 two, respectively, and $q(x)=\theta p_2(x)$ for some
 $\theta\neq 0$.
 \end{prop}

 \begin{proof}
 Since $X$ is Integrated Pearson,
 $Y=\lambda X+c$
 is also Integrated Pearson for all $\lambda\neq 0$ and $c\in\R$; see
 Proposition \ref{prop.1}(vi).
 Also, its density $f_Y(x)=\frac{1}{|\lambda|}f(\frac{x-c}{\lambda})$
 satisfies, by assumption, the differential equation
 \[
 \frac{f_Y'(x)}{f_Y(x)}=\frac{\widetilde{p}_1^{Y}(x)}{\widetilde{p}_2^{Y}(x)},
 \ \ x\in (\widetilde{\alpha},\widetilde{\omega}),
 %in the interval
 %$(\widetilde{\alpha},\widetilde{\omega})$,
 \ \ \
 \mbox {with} \
 \widetilde{p}_1^{Y}(x)=\lambda\widetilde{p}_1\left(\frac{x-c}{\lambda}\right),
 \
 \widetilde{p}_2^{Y}(x)=\lambda^2
 \widetilde{p}_2\left(\frac{x-c}{\lambda}\right),
 \]
 where $(\widetilde{\alpha},\widetilde{\omega})=(\lambda\alpha+c,\lambda\omega+c)$
 or $(\lambda\omega+c,\lambda\alpha+c)$, according to $\lambda>0$ or $\lambda<0$,
 respectively. It is easily shown that
 the new polynomials $p_1,p_2$
 (those that the algorithm produces at Step 3 for $f$)
 are related to the corresponding  polynomials
 $p_1^Y$, $p_2^Y$ (those that the algorithm produces at Step 3 for
 $f_Y$) by the relationships
 \[
 p_1^{Y}(x)=\lambda^i p_1\left(\frac{x-c}{\lambda}\right), \ \ \
 p_2^{Y}(x)=\lambda^{i+1}p_2\left(\frac{x-c}{\lambda}\right),
 \]
 for some $i\in\{1,2,3\}$. Therefore, it suffices to show
 that $\deg(p_i^Y)\leq i$, $i=1,2$, and that
 the quadratic polynomial $q_Y(x)=\lambda^2 q(\frac{x-c}{\lambda})$
 of $Y$ is related to $p_2^Y$ through $q_Y(x)=\theta p_2^Y(x)$
 for some $\theta\neq0$. Thus, without any loss of generality
 we may assume that $f$ is one of the densities given in
 Table \ref{table densities}.

 Now observe that $(\widetilde{p}_1,\widetilde{p}_2)$
 is always irreducible for types $1,4,5$ (Normal-type,
 Student-type, Reciprocal Gamma-type)
 with $\deg(\widetilde{p}_1)=1$ for all types $1,4,5$,
 while $\deg(\widetilde{p}_2)=0$ for type $1$
 and $\deg(\widetilde{p}_2)=2$ for types $4$ and $5$. Since the
 corresponding supports are $\R$, $\R$ and $(0,\infty)$,
 respectively, and
 since in type $5$, $\widetilde{p}_2(x)=\theta x^2$ for some
 $\theta\neq 0$, it follows that
 $(p_1,p_2)=(\widetilde{p}_1,\widetilde{p}_2)$, $q=\theta p_2$
 for some $\theta\neq 0$, and the assertion follows.

 For
 types $2,3$ and $6$ (Gamma-type, Beta-type and Snedecor-type)
 the irreducibility of $\widetilde{p}_1$ and $\widetilde{p}_2$
 depends on the parameters. Let us see these cases separately.

 If $f\sim\varGamma(a,\lambda)$  with $a\neq 1$ ($\alpha>0$, $\lambda>0$)
 then $\widetilde{p}_1=\theta (a-1-\lambda x)$
 and $\widetilde{p}_2=\theta  x$ for some $\theta\neq 0$, so that
 $\widetilde{p}_1,\widetilde{p}_2$ are irreducible with degree one.
 It follows that $p_i=\widetilde{p}_i$,
 $\deg(p_i)=1$ ($i=1,2$) and
 \[
 q(x)=\frac{x}{\lambda}=\frac{p_2(x)}{\theta\lambda}.
 \]
 If
 $f\sim\varGamma(1,\lambda)$ ($\lambda>0$) then all possible
 choices for $(\widetilde{p}_1,\widetilde{p_2})$ are given
 by $\widetilde{p}_1=-\lambda\theta(x+c)$ and
 $\widetilde{p}_2=\theta(x+c)$
 for $\theta\neq 0$, $c\in\R$. Therefore, Step 3 yields
 $(p_1,p_2)=(-\lambda\theta x,\theta x)$
 and, thus, $\deg(p_i)=1$  ($i=1,2$) and
 \[
 q(x)=\frac{x}{\lambda}=\frac{p_2(x)}{\lambda\theta}.
 \]

 If $f$ is of type $6$ and $b\neq 1$ then
 \[
 (\widetilde{p}_1(x),\widetilde{p}_2(x))=(c((b-1)-(a+1)x),c x(x+\theta))
 \
 \mbox{for some}
 \ c\neq 0;
 \]
 here the parameters are $a,b,\theta$
 with $a>1$, $b>0$ and $\theta>0$. It follows that
 $(p_1,p_2)=(\widetilde{p}_1,\widetilde{p}_2)$,
 $\deg(p_i)=i$ ($i=1,2$) and
 \[
 q(x)=\frac{x(x+\theta)}{a-1}=\frac{p_2(x)}{(a-1)c}.
 \]
 If $f$ is of type $6$ with $b=1$ then
 all possible choices for $(\widetilde{p}_1,\widetilde{p_2})$
 are given by
 \[
 \widetilde{p}_1(x)=-c(a+1)(x+\gamma) \ \mbox{and} \
 \widetilde{p}_2(x)=c(x+\theta)(x+\gamma) \
 \mbox{for some}
 \
 c\neq 0,
 \
 \gamma\in\R.
 \]
 Therefore, Step 3 yields
 $(p_1,p_2)=(-c(a+1)x, cx(x+\theta))$
 and, thus, $\deg(p_i)=i$  ($i=1,2$) and
 \[
 q(x)=\frac{x(x+\theta)}{a-1}=\frac{p_2(x)}{(a-1)c}.
 \]

 Finally, let $f$ be of type $3$ (Beta-type),
 that is, $f\sim B(a,b)$ with $a,b>0$.
 If $a\neq 1$
 and $b\neq 1$ it is easily shown that
 \[
 (\widetilde{p}_1(x),\widetilde{p}_2(x))
 =(\theta(a-1-(a+b-2)x),\theta x(1-x)) \ \
 (\theta\neq0)
 \]
 are irreducible, so that
 $(p_1,p_2)=(\widetilde{p}_1,\widetilde{p}_2)$,
 $\deg(p_i)=i$ ($i=1,2$) and
 \[
 q(x)=\frac{x(1-x)}{a+b}=\frac{p_2(x)}{(a+b)\theta}.
 \]
 If $a=1$, $b\neq 1$, the most general form of $(\widetilde{p}_1,\widetilde{p}_2)$
 is given by
 \[
 (\widetilde{p}_1(x),\widetilde{p}_2(x))
 =(-(b-1)\theta (x+c),\theta
 (1-x)(x+c)),
 \ \
 \mbox{where} \ \theta\neq0, \ c\in\R.
 \]
 Therefore, Step 3 yields $(p_1,p_2)=(-(b-1)\theta x, \theta x(1-x))$
 and, thus, $\deg(p_i)=i$  ($i=1,2$) and
 \[
 q(x)=\frac{x(1-x)}{b+1}=\frac{p_2(x)}{(b+1)\theta}.
 \]
 If $a\neq 1$, $b=1$, the most general form of
 $(\widetilde{p}_1,\widetilde{p}_2)$
 is given by
 \[
 (\widetilde{p}_1(x),\widetilde{p}_2(x))
 =((a-1)\theta(x+c), \theta x (x+c)),
 \ \
 \mbox{where}
 \
 \theta\neq0,
 \
 c\in\R.
 \]
 Therefore, Step 3 yields
 $(p_1,p_2)=((a-1)\theta(1-x), \theta x(1-x))$
 and, thus, $\deg(p_i)=i$  ($i=1,2$) and
 \[
 q(x)=\frac{x(1-x)}{a+1}=\frac{p_2(x)}{(a+1)\theta}.
 \]
 Finally, if $a=b=1$ (standard uniform density,
 $U(0,1)\equiv B(1,1)$)
 then $\widetilde{p}_1\equiv 0$ so that
 $(p_1,p_2)=(0,x(1-x))$, $\deg(p_1)<0$, $\deg(p_2)=2$ and
 \[
 q(x)=\frac{x(1-x)}{2}=\frac{p_2(x)}{2}.
 \]
 This subsumes all cases and completes the proof.
 \end{proof}
 \begin{prop}
 \label{prop.alg2}
 Assume that $X\sim f$ where the density $f$ is differentiable with
 derivative $f'$ in its (known) interval support $(\alpha,\omega)$
 and has finite (unknown) mean. Then, the following are
 equivalent:
 \begin{itemize}
 \item[(A)] $f$ satisfies (\ref{Pearson.ord}) for some
    (real) polynomials $p_1$
    (of degree at most one) and $p_2\not\equiv 0$
    (of degree at most two) with $p_2(\alpha)=0$ if
    $\alpha>-\infty$, $p_2(\omega)=0$ if $\omega<\infty$ and
    $p_2(x)\neq 0$ for all $x\in(\alpha,\omega)$.
    %Also,
    %$p_1$ and $p_2$ do not have any common zeros
    %in $\C\smallsetminus\{\alpha,\omega\}$.
 \item[(B)] $X\sim\IPq$ for some $q(x)=\delta x^2+\beta x+\gamma$
    with $\{x:q(x)>0\}=(\alpha,\omega)$ and some $\mu\in(\alpha,\omega)$.
 \end{itemize}
 Moreover, if (A) and (B) hold, then there exists a constant $\theta\neq 0$
 such that $q(x)=\theta p_2(x)$, $x\in\R$.
 \end{prop}
 \begin{proof}
 Assume first that (B) holds.
 Since $f\sim\IPq$, (\ref{Pearson}) shows that
 $f'/f=\widetilde{p}_1/\widetilde{p}_2$
 where $\widetilde{p}_1=\mu-x-q'$ and $\widetilde{p}_2=q$.
 Putting the polynomials $\widetilde{p}_1=\mu-x-q'$ and $\widetilde{p}_2=q$
 in Step 0
 of the above algorithm and using Proposition
 \ref{prop.alg1}
 we conclude that the resulting polynomials
 $p_1$ and $p_2$ (of Step 3) satisfy the
 requirements of (A); also, $q(x)=\theta p_2(x)$
 for some $\theta\neq 0$.

 Assume now that (A) holds. Using a suitable mapping
 $Y=\lambda X+c$, $\lambda\neq 0$, $c\in\R$, we can
 transform the interval $(\alpha,\omega)$ into
 $(\widetilde{\alpha},\widetilde{\omega})$, where
 $(\widetilde{\alpha},\widetilde{\omega})$ is one
 of the intervals $(0,1)$, $(0,\infty)$ or $(-\infty,\infty)$.
 The polynomials $p_1$ and $p_2$ are transformed to
 $p_1^Y(x)=\lambda p_1(\frac{x-c}{\lambda})$ and
 $p_2^Y(x)=\lambda^2 p_2(\frac{x-c}{\lambda})$, and the
 differential equation (\ref{Pearson.ord}) yields
 $f_Y'(x)/f_Y(x)=p_1^Y(x)/p_2^Y(x)$, $\widetilde{\alpha}
 <x<\widetilde{\omega}$, where $f_Y$ is the density of $Y$
 and $(\widetilde{\alpha},\widetilde{\omega})$ its
 support. Moreover, it is easy to see that
 $p_1^Y$ and $p_2^Y$ satisfy the requirements
 of (A), i.e., $p_2^Y(\widetilde{\alpha})=0$ if
 $\widetilde{\alpha}>-\infty$, $p_2^Y(\widetilde{\omega})=0$ if
 $\widetilde{\omega}<\infty$ and
 $p_2^Y(x)\neq 0$ for all
 $x\in(\widetilde{\alpha},\widetilde{\omega})$.
 Clearly, in view of Proposition \ref{prop.1}(vi),
 it suffices to verify that $Y$ is Integrated
 Pearson. Thus, from now on (and without any loss of
 generality) we shall assume that $(\alpha,\omega)$
 is one of  the intervals $(0,1)$, $(0,\infty)$ or $\R$.

 If $(\alpha,\omega)=(0,1)$ then the assumptions (A)
 show that $p_2(x)=\theta x(1-x)$ for some $\theta\neq 0$.
 Let $p_1(x)=a_0+a_1 x$. Solving (\ref{Pearson.ord}) we get
 \[
 f(x)=Cx^{a_0/\theta}(1-x)^{-(a_0+a_1)/\theta}, \ \ 0<x<1,
 \]
 where, necessarily, $1+a_0/\theta>0$ and
 $1-(a_0+a_1)/\theta>0$. Thus,
 \[
 (1+a_0/\theta)+(1-(a_0+a_1)/\theta)=(2\theta-a_1)/\theta>0,
 \]
 so that $2\theta-a_1\neq 0$. It follows that $f\sim B(a,b)$
 with $a=1+a_0/\theta$, $b=1-(a_0+a_1)/\theta$ and,
 therefore,
 \[
 q(x)=\frac{x(1-x)}{a+b}=\frac{x(1-x)}{2-a_1/\theta}
 =\frac{p_2(x)}{2\theta-a_1}.
 \]

 Assume that $(\alpha,\omega)=(0,\infty)$. Then,
 assumptions (A) show that the possible forms of
 $p_2$ are either $p_2=\theta x$ or $p_2=\theta x^2$
 or  $p_2=\theta x (x+c)$ for some $\theta\neq 0$ and
 $c>0$. If $p_2=\theta x$ set $p_1=a_0+a_1 x$ and solve
 (\ref{Pearson.ord}) to obtain
 \[
 f(x)=Cx^{a_0/\theta}\exp({a_1 x/\theta}), \ \ x>0,
 \]
 where,
 necessarily, $a_0/\theta>-1$ and $a_1/\theta<0$; thus,
 $X\sim\varGamma(a,\lambda)$ with $a=\frac{a_0}{\theta}-1>0$
 and $\lambda=-\frac{a_1}{\theta}>0$. Therefore, $a_1\neq 0$
 and
 \[
 q(x)=\frac{x}{\lambda}=\frac{p_2(x)}{-a_1}.
 \]
 If $p_2=\theta x^2$, set $p_1=a_0+a_1 x$ and solve
 (\ref{Pearson.ord}) to obtain
 \[
 f(x)=Cx^{a_1/\theta}\exp(-a_0/(\theta x)), \ \ x>0,
 \]
 where,
 necessarily, $a_0/\theta>0$ and $a_1/\theta<-2$;
 these conditions are necessary and sufficient for
 $\int_{0}^{\infty} f(x)\ud{x}$ and $\int_{0}^{\infty}
 xf(x)\ud{x}$ to be finite. Therefore,
 $f(x)=Cx^{-a-1}e^{-\lambda/x}$, $x>0$,
 where $a=-1-\frac{a_1}{\theta}>1$ and
 $\lambda=\frac{a_0}{\theta}>0$. Observe now that $f$ is of
 Reciprocal Gamma type (type 5) and $q(x)=\frac{x^2}{a-1}$.
 Since $a=-1-\frac{a_1}{\theta}>1$
 it follows that $\frac{-a_1-2\theta}{\theta}>0$ and,
 finally, $a_1+2\theta\neq 0$. Thus,
 \[
 q(x)=\frac{x^2}{a-1}=\frac{\theta x^2}{-a_1-2\theta}
 =\frac{p_2(x)}{-a_1-2\theta}.
 \]
 Assume now that
 $p_2=\theta x(x+c)$, $\theta\neq 0$, $c>0$ and let
 $p_1=a_0+a_1 x$. Solving (\ref{Pearson.ord}) we obtain
 \[
 f(x)=Cx^{\frac{a_0}{c\theta}}(x+c)^\frac{a_1c-a_0}{c\theta},
 \ \ x>0,
 \]
 where,
 necessarily, $\frac{a_0}{c\theta}>-1$ and $\frac{a_1}{\theta}<-2$;
 these conditions are necessary and sufficient for
 $\int_{0}^{\infty} f(x)\ud{x}$ and $\int_{0}^{\infty}
 xf(x)\ud{x}$ to be finite. Now observe that $f(x)=C
 x^{b-1}(x+c)^{-a-b}$ ($x>0$)
 is of Snedecor-type (type 6) with
 $a=-\frac{a_1}{\theta}-1>1$ and
 $b=1+\frac{a_0}{c\theta}>0$. From $\frac{a_1}{\theta}<-2$
 we get $a_1+2\theta\neq 0$ and, thus, we conclude that
 (see Table \ref{table densities})
 \[
 q(x)=\frac{x(x+c)}{a-1}=\frac{x(x+c)}{-2-a_1/\theta}
 =\frac{\theta x(x+c)}{-a_1-2\theta}=\frac{p_2(x)}{-a_1-2\theta}.
 \]

 Finally, assume that $(\alpha,\omega)=\R$. In this case
 assumptions (A) imply that either $p_2\equiv \theta\neq 0$
 or $p_2=\pm(\theta(x-c)^2+\lambda)$ with $\theta>0$,
 $\lambda>0$ and $c\in\R$. Assume first that $p_2\equiv\theta\neq0$
 and let $p_1=a_0+a_1x$. Then, it is easily seen from
 (\ref{Pearson.ord}) that
 \[
 f(x)=C\exp\left(\frac{a_1}{2\theta}x^2+\frac{a_0}{\theta}x\right), \ \ x\in\R.
 \]
 This can represents a density if and only if
 $\frac{a_1}{2\theta}<0$; in this case it is easily seen
 that $f\sim N(\mu,\sigma^2)$ with $\mu=\frac{-a_0}{a_1}$,
 $\sigma=\sqrt{-\frac{\theta}{a_1}}$, and thus,
 \[
 q(x)\equiv\sigma^2=-\frac{\theta}{a_1}=\frac{p_2(x)}{-a_1}.
 \]
 For the last remaining case it suffices to consider
 \[
 p_2(x)=\theta(x-c)^2+\lambda
 \ \ \mbox{and} \ \
 p_1(x)=a_0+a_1(x-c) \ \
 \mbox{where} \  \theta>0,
 \ \lambda>0 \ \mbox{and}
 \ \ a_0,a_1,c\in\R.
 \]
 Also, using the transformation $X\mapsto X-c$,
 the general case is simplified to $p_2=\theta x^2+\lambda$
 and $p_1=a_0+a_1 x$. Now, the differential equation (\ref{Pearson.ord})
 has the general solution
 \[
 f(x)=C(\theta x^2+\lambda)^{\frac{a_1}{2\theta}}\exp\left[\frac{a_0}{\sqrt{\theta\lambda}}
 \tan^{-1}(x\sqrt{\theta/\lambda})\right], \ \ x\in\R.
 \]
 The
 necessary and sufficient condition for this $f$ to
 represent a density with finite mean is
 $-\frac{a_1}{2\theta}-1>0$ or, equivalently, $a_1+2\theta<0$.
 Therefore, setting
 \[
 \delta=\frac{\theta}{-a_1-2\theta}>0, \ \
 \gamma=\frac{\lambda}{-a_1-2\theta}>0 \ \
  \mbox{and} \ \
 \mu=\frac{a_0}{-a_1-2\theta}\in\R
 \]
 we see that this
 is a Student-type density (type $4$); see Table \ref{table densities}.
 Consequently,
 \[
 q(x)=\delta x^2+\gamma
 =\frac{\theta
 x^2+\lambda}{-a_1-2\theta}=\frac{p_2(x)}{-a_1-2\theta},
 \]
 and the proof is complete.
 \end{proof}

 Eventually, Proposition \ref{prop.alg2}
 says that for a particular choice of $p_2$ to be correct
 it is necessary and sufficient that $p_2$ remains
 nonzero in $(\alpha,\omega)$ and vanishes
 at all (if any)
 finite endpoints of $(\alpha,\omega)$.

 If the mean $\mu$ is known, then another
 simple criterion for an ordinary Pearson variate to
 belong to the Integrated Pearson family is provided by
 the following proposition.

 \begin{prop}
 \label{prop.2}
 Let $X$ be a random variable with density $f$ and finite mean
 $\mu$. Assume that the set $\{x:f(x)>0\}$ is the
 (bounded or unbounded) interval
 $J(X)=(\alpha,\omega)$ and that $f$ is differentiable in
 $(\alpha,\omega)$ with derivative
 $f'(x)$, $\alpha<x<\omega$.
 Then the following are equivalent:
 \begin{itemize}
 \item[(A)] $X\sim\IPq$.
 \item[(B)] The density $f$ satisfies (\ref{Pearson.ord})
 and the polynomials $p_1$ ($p_1\equiv0$ is allowed)
 and $p_2$ can be chosen in such a way that (i) and (ii),
 below, hold:
 \begin{itemize}
 %\item[(i)] $p_2(x)\neq 0$ for $\alpha<x<\omega$,
 %\item[(ii)] $p_2(\alpha)=0$ if $\alpha>-\infty$ and,
 %similarly, $p_2(\omega)=0$ if $\omega<\infty$,
 \item[(i)] there exist a constant $\theta\neq 0$
 such that $p_1(x)+p_2'(x)=(\mu-x)/\theta$, $x\in\R$, and
 \item[(ii)] either $\lim_{x\searrow\alpha} p_2(x)f(x)=0$
 or $\lim_{x\nearrow\omega} p_2(x)f(x)=0$.
 \end{itemize}
 \end{itemize}
 If (i) and (ii) are true then the polynomials $p_2$ and $q$
 are related through
 $q(x)=\theta p_2(x)$ where $\theta\neq 0$ is as in (i).
 Moreover, if (\ref{Pearson.ord}) is satisfied in an {\it
 unbounded} interval $(\alpha,\omega)$ then (ii)
 is unnecessary since it
 is implied by (i).
 \end{prop}
 \begin{proof}
 If $X\sim\IPq$ then we see from (\ref{Pearson})
 that (\ref{Pearson.ord}) is satisfied for the polynomials
 $p_1(x)=\mu-x-q'(x)$ and $p_2(x)=q(x)$. With this choice
 of $p_1$, $p_2$, Proposition \ref{prop.1} shows that (i) (with $\theta=1$)
 is valid. Also,
 (ii) reduces to $p_2(x)f(x)=q(x)f(x)\to 0$ as $x\nearrow\omega$ or
 $x\searrow\alpha$; this follows by an obvious application of
 dominated convergence since the mean exists and, by
 assumption, $p_2(x)f(x)=q(x)f(x)=\int_{\alpha}^x (\mu-t)f(t)\ud{t}$ --
 see (\ref{qua1}). Conversely, (\ref{Pearson.ord})
 and (i) imply that $[\theta p_2(t)f(t)]'= (\mu-t)f(t)$,
 $\alpha<t<\omega$. Integrating this
 equation over the interval $[x,y]\subset (\alpha,\omega)$
 we get
 \be
 \label{double}
 \int_{x}^y (\mu-t)f(t)\ud{t} = \theta p_2(y)f(y)
 -\theta p_2(x) f(x), \ \ \ \alpha<x<y<\omega.
 \ee
 Now, let us take into account the first assumption in (ii),
 $\lim_{x\searrow\alpha} p_2(x)f(x)=0$.
 Taking limits in (\ref{double}) and using dominated convergence
 for the l.h.s.\
 we conclude that
 \[
 \int_{\alpha}^y (\mu-t)f(t)\ud{t} = {\theta} p_2(y)f(y),
 \ \ \ \alpha<y<\omega;
 \]
 that is, $X\sim\IPq$ with $q(x)=\theta p_2(x)$. Clearly
 we get the same conclusion if we use the second assumption in (ii),
 $\lim_{y\nearrow\omega} p_2(y)f(y)=0$,
 and evaluate  the limits as $y\nearrow\omega$ in (\ref{double});
 in this case we get the identity
 %\[
 $\int_{x}^\omega (t-\mu)f(t)\ud{t} = {\theta} p_2(x)f(x)=q(x)f(x)$,
 %\ \ \
 $\alpha<x<\omega$,
 %\]
 which is equivalent to (\ref{qua1}),
 since $\int_{\alpha}^{\omega} (\mu-t)f(t)\ud{t}=0$.

 It is clear that, in the presence of (i), both assumptions in (ii) are
 equivalent. In fact, (\ref{double}) shows that both limits
 $\lim_{y\nearrow\omega}p_2(y)f(y)$ and $\lim_{x\searrow\alpha}p_2(x)f(x)$
 exist (in $\R$) and are equal. Indeed,
 \[
 {\theta}p_2(y)f(y)={\theta}p_2(x)f(x)+\int_{x}^y (\mu-t) f(t)\ud{t}, \ \
 \ \alpha<x<y<\omega,
 \]
 and the existence of the first moment implies that, as $y\nearrow\omega$,
 the r.h.s.\ has the well-defined finite limit
 $C(x)={\theta}p_2(x)f(x)+\int_{x}^{\omega}
 (\mu-t) f(t)\ud{t}$; the l.h.s, however, is independent of
 $x$ and, certainly, the same is true for its limit, so that
 $C(x)\equiv C$. In other words,
 \[
 {\theta}p_2(x)f(x)=C+\int_{x}^{\omega}
 (t-\mu)f(t)\ud{t}, \ \ \ \alpha<x<\omega,
 \]
 and since $\lim_{x\searrow\alpha}\int_{x}^{\omega}
 (t-\mu)f(t)\ud{t}=\int_{\alpha}^{\omega} (t-\mu)f(t)\ud{t}=0$ we
 conclude that
 \[
 \lim_{x\searrow\alpha}p_2(x)f(x)=\lim_{y\nearrow\omega}p_2(y)f(y)
 =\frac{C}{\theta}\in\R.
 \]
 It remains to verify that if (\ref{Pearson.ord}) holds in an unbounded
 interval $(\alpha,\omega)$ and $X$ has finite first moment then (i) implies (ii).
 To this end assume that $\omega=\infty$ so that $J(X)=(\alpha,\infty)$
 with $\alpha\in[-\infty,\infty)$. It follows that $f'(x)=p_1(x)f(x)/p_2(x)$
 does not change sign for large enough $x$, and thus,
 $f'(x)<0$ for $x>x_0$. Therefore,
 for $x>\max\{2x_0,0\}$,
 \[
 0<x^2 f(x)=\frac{8}{3}f(x)
 \int_{x/2}^{x}t\ud{t}< \frac{8}{3}\int_{x/2}^{x} tf(t)\ud{t}<\frac{8}{3}\int_{x/2}^{\infty}
 tf(t)\ud{t}\to 0,
 \]
 as
 $x\to\infty$, i.e.\ $f(x)=o(x^{-2})$ as $x\to\infty$.
 Thus, $p_2(x)f(x)\to 0$ as $x\to\infty$. The case $\alpha=-\infty$
 is similar and the proof is complete.
 \end{proof}

 \section{Are the Rodrigues-type polynomials orthogonal
 in the ordinary Pearson system?}
 \label{sec4}
 Associated with any Pearson density $f$
 is a (unique) sequence of polynomials,
 defined by a Rodrigues-type formula.
 Actually, these polynomials are by-products of the pair
 $(p_1,p_2)$ that appears in the nominator and the
 denominator of the differential equation
 (\ref{Pearson.ord}); that is, they have nothing to do
 either with $f$ or with the interval
 $(\alpha,\omega)$.

 These considerations will become more clear
 if we slightly relax the form of differential equation
 (\ref{Pearson.ord}) and permit more solutions,
 as follows:

 \begin{defi}
 \label{def.compatible}
 Let $\varnothing\neq(\alpha,\omega)\subseteq \R$, and consider
 a pair of real polynomials $(p_1,p_2)=(a_0+a_1 x,b_0+b_1x+b_2 x^2)$
 such that $p_2\not\equiv0$ (i.e., $|{b}_0|+|{b}_1|+|{b}_2|>0$).
 The pair $(p_1,p_2)$ is called {\it Pearson-compatible}
 in $(\alpha,\omega)$, or simply {\it compatible},
 if there exists a differentiable function $f:(\alpha,\omega)\to\R$,
 $f\not\equiv 0$
 ($f$ is not assumed nonnegative or integrable),
 %differentiable in $(\alpha,\omega)$,
 such that the following {\it generalized {Pearson} differential
 equation} is fulfilled:
 \be
 \label{Pearson2}
 p_2(x) f'(x)=p_1(x)f(x), \ \ \ \alpha<x<\omega.
 \ee
 In other words, $(p_1,p_2)$ is compatible if (\ref{Pearson2})
 has non-trivial solutions for $f$.
 \end{defi}

 It is easily seen that $(p_1,p_2)$ is compatible
 whenever $p_2$ has no roots in $(\alpha,\omega)$; in this case,
 the general solution $f$ is $C^{\infty}(\alpha,\omega)$ and can be
 chosen to be strictly positive in $(\alpha,\omega)$.
 The presence of a zero of $p_2$ in $(\alpha,\omega)$,
 however, may results in incompatibility; e.g.,
 in the interval $(\alpha,\omega)=(-2,2)$ the pair
 $(p_1,p_2)=(4x,x^2-1)$ is compatible,
 in contrast to the pair $(p_1,p_2)=(x,x^2-1)$.

 If $(p_1,p_2)$ is compatible in $(\alpha,\omega)$ then
 we can find the general solution as follows: First we
 solve (\ref{Pearson2}) separately in any open
 subinterval of $(\alpha,\omega)\cap \{x:p_2(x)\neq 0\}$;
 clearly, there are at most three subintervals and,
 in the worst case, the three general solutions
 for the distinct intervals
 $(J_1,J_2,J_3)=((\alpha,\rho_1),(\rho_1,\rho_2),(\rho_2,\omega))$
 will be of the form $f_i=C_ie^{g_i}$ for some $g_i\in C^{\infty}(J_i)$,
 $i=1,2,3$, with $C_i$ being arbitrary constants.
 Next, we match the solutions and their first derivatives at
 the common endpoints of any two $J_i$; any such point is,
 necessarily, a zero of $p_2$. The compatibility of $(p_1,p_2)$
 guarantees that this procedure will success
 in producing
 some solution
 $f\not\equiv0$ (in which case, $|f|\geq 0$ will
 be also a non-trivial solution),
 but it may happen that $f_i\equiv 0$
 in some $J_i$. The following proposition describes
 all possible cases for the support of $f$.
 \begin{prop}
 \label{prop.support}
 Assume that the function $f:(\alpha,\omega)\to\R$,
 $f\not\equiv 0$ (not necessarily positive or integrable)
 is differentiable in $(\alpha,\omega)$ and satisfies the
 differentiable equation (\ref{Pearson2})
 for some real polynomials $p_1(x)=a_0+a_1x$
 and $p_2(x)=b_0+b_1x+b_2 x^2$ with
 $|{b}_0|+|{b}_1|+|{b}_2|>0$.
 Then, the support of
 $f$, $S(f):=\{x\in(\alpha,\omega):f(x)\neq 0\}$,
 is either of the form
 $(\widetilde{\alpha},\widetilde{\omega})\subseteq
 (\alpha,\omega)$ with
 $\alpha\leq\widetilde{\alpha}<\widetilde{\omega}\leq\omega$,
 or of the form $(\widetilde{\alpha},\rho_1)
 \cup(\rho_2,\widetilde{\omega})\subseteq (\alpha,\omega)$
 with $\alpha\leq\widetilde{\alpha}<\rho_1\leq\rho_2
 <\widetilde{\omega}\leq\omega$,
 or, finally, of the form
 $(\alpha,\rho_1)\cup(\rho_1,\rho_2)\cup(\rho_2,\omega)$,
 with $\alpha<\rho_1<\rho_2<\omega$.
 Moreover, the boundary of
 $S(f)$ is contained  in the set
 $\{\alpha,\omega\}\cup \{x\in (\alpha,\omega):p_2(x)=0\}$,
 that is, $\partial S(f)\subseteq\{\alpha,\omega\}\cup \{x\in
 (\alpha,\omega):p_2(x)=0\}$. Finally, for any solution
 $f$, $f(\rho)=0$ (that is, $\rho\notin S(f)$)
 whenever $\rho$ is a zero of $p_2$ which is not a
 zero of $p_1$.
 \end{prop}
 \begin{cor}
 \label{cor.support}
 The differential equation (\ref{Pearson2}) has a
 nontrivial and nonnegative solution if and only if
 the pair $(p_1,p_2)$ is compatible in $(\alpha,\omega)$.
 Moreover, assuming that $(p_1,p_2)$ is compatible in
 $(\alpha,\omega)$, it follows that:
 \begin{itemize}
 \item[(a)] any nonnegative solution is of the form $|f|$ for some solution
 $f$;
 \item[(b)] the support $S(f)=\{x\in(\alpha,\omega):f(x)\neq 0\}$
 of any nontrivial solution $f$ of (\ref{Pearson2})
 is a union of one, two or three disjoint open intervals of positive length,
 and the same is true for any nonnegative and nontrivial solution;
 \item[(c)] the boundary points of $S(f)=S(|f|)$ of any nontrivial solution $f$
 of (\ref{Pearson2}) are either roots of $p_2$ or boundary points of
 $(\alpha,\omega)$.
 \end{itemize}
 \end{cor}
 %\noindent
 We now turn to the corresponding Rodrigues polynomials.
 It is well-known that
 the (generalized) Pearson differential equation
 (\ref{Pearson2}) produces a sequence of polynomials
 $\{h_{k}, k=1,2,\ldots\}$, defined by a
 {\it Rodrigues-type formula},
 as follows:
 %%%\pagebreak

 \begin{theo}
 %[\textrm{\cite[p.\ 401]{Hild}, \cite[pp.\ 99-100]{Beale2}, \cite[p.\ 295]{DZ}}]
 [Hildebrandt (1931), p.\ 401; Beale (1941), pp.\ 99--100;
 Diaconis and Zabell (1991), p.\ 295]
 \label{theo.polynomials}
 Assume that a function $f:(\alpha,\omega)\to\R$
 (not necessarily positive or integrable)
 does not vanish identically in $(\alpha,\omega)$ and
 satisfies the differential equation (\ref{Pearson2})
 %\be
 %\label{Pearson2-b}
 %p_2(x) f'(x)=p_1(x)f(x), \ \ \ \alpha<x<\omega,
 %\ee
 for some polynomials $p_1(x)=a_0+a_1x$
 and $p_2(x)=b_0+b_1x+b_2 x^2$, with
 $|{b}_0|+|{b}_1|+|{b}_2|>0$.
 Then, the set $\{x\in(\alpha,\omega):f(x)\neq 0\}$
 contains some interval of positive length
 and the function
 \be
 \label{polynomials}
 h_k(x):=\frac{1}{f(x)}\frac{\ud^k}{\ud{x}^k} [p_2^{k}(x)f(x)], \
 \ \ x\in (\alpha,\omega)\smallsetminus\{x:f(x)=0\}, \ \ \  k=0,1,2,\ldots
 \ee
 is a polynomial (more precisely, $h_k$ is
 the restriction in $(\alpha,\omega)\smallsetminus\{x:f(x)=0\}$
 of a polynomial $\widetilde{h}_k:\R\to\R$)
 with
 \be
 \label{lead1}
 \deg(h_k)\leq k\quad \textrm{and}
 \quad \lead(h_k)=\prod_{j=k+1}^{2k}({a}_1+j{b}_2),
 \quad k=0,1,2,\ldots,
 \ee
 where $\lead(h_k):=\lim_{x\to\infty}\widetilde{h}_k(x)/x^k$
 denotes the coefficient of $x^k$ in
 $h_k(x)$.
 \end{theo}
 Hildebrandt (1931)
 %\cite{Hild},
 actually showed that the relation $p_2 f'= p_1 f$
 implies that $D^k[p_2^k f]=\widetilde{h}_k f$,
 \mbox{$k=0,1,2,\ldots$,} where the polynomials $\widetilde{h}_k$ (with
 $\deg(\widetilde{h}_k)\leq k$) are defined inductively.
 Each polynomial $\widetilde{h}_k$ can be viewed as the
 value of a functional ${\cal R}_k$ that maps any pair
 $(p_1,p_2)$ to a real polynomial of degree at most $k$.
 The form of this functional is
 \[
 (p_1,p_2)\mapsto {\cal R}_k(p_1,p_2)
 := \widetilde{h}_k=\sum_{r,i,j} C_{k;rij}^{a_1,b_2}
 (p_1)^r (p_2')^i (p_2)^j
 \]
 where the sum ranges over all integers $r,i,j\geq 0$
 with $r+i+2j\leq k$, and the constant $C_{k;rij}^{a_1,b_2}$
 depends only on $k,r,i,j,p_1'=a_1$ and $p_2''=2b_2$.
 On the other hand it is clear that, given an arbitrary pair
 $(p_1,p_2)$ with $p_2\not\equiv0$,
 we can fix an interval $(\alpha,\omega)$ containing
 no roots of $p_2$. With the help of a positive solution $f$
 of the differential equation (\ref{Pearson2}) we can determine
 $h_k(x)$, $\alpha<x<\omega$, using the Rodrigues-type formula
 (\ref{polynomials}). Obviously, this $h_k$ extends uniquely
 to $\widetilde{h}_k$.

 To give an idea about the nature of the polynomials in (\ref{polynomials})
 we expand the first four:
 \[
 \begin{split}
 &h_0=1;\\
 &h_1 = p_1 + p_2' = ({a}_1+2{b}_2)x+({a}_0+{b}_1);\\
 &h_2= p_1^2 +  3 p_1 p_2' + p_1' p_2 +  2 p_2 p_2'' + 2 (p_2')^2 \\
  &\hspace{10ex}=({a}_1+3{b}_2)({a}_1+4{b}_2)x^2+2(a_0+2b_1)(a_1+3b_2)x\\
  &\hspace{12.5ex}+(a_0+b_1)(a_0+2b_1)+b_0(a_1+4b_2);\\
 &h_3=p_1^3
    + 6 p_1^2 p_2' + 3 p_1 p_1' p_2
    + 8 p_1 p_2 p_2''+ 11 p_1 (p_2')^2 + 7 p_1' p_2 p_2'
    + 18 p_2 p_2' p_2'' +  6(p_2')^3\\
    &
    \hspace{10ex}
    =({a}_1+4{b}_2)({a}_1+5{b}_2)(a_1+6b_2)x^3+ 3(a_0+3b_1)(a_1+4b_2)(a_1+5b_2)x^2\\
    &
    \hspace{12.5ex}+3(a_1+4b_2)[(a_0+2b_1)(a_0+3b_1)+b_0(a_1+6b_2)]x\\
    &\hspace{12.5ex}+a_0^3+6a_0^2b_1+a_0[11b_1^2+b_0(3a_1+16b_2)]+b_1[6b_1^2+b_0(7a_1+36b_2)].
 \end{split}
 \]
 Provided that the solution $f$ of (\ref{Pearson2}) is a probability density
 in $(\alpha,\omega)$, the polynomials $h_k$ are candidate to
 form an orthogonal system for $f$. Indeed, Hildebrandt
 (1931), pp.\ 404--405,
 %, \cite[pp.\ 404--405]{Hild},
 showed
 that each $h_k$ satisfies a specific second order differential
 equation in $(\alpha,\omega)$.
 Using this
 differential equation
 Diaconis and Zabell (1991)
 %\cite{DZ},
 proved that the $h_k$ are eigenfunctions of
 a particular self-adjoint, second order Sturm-Liouville differential
 equation; thus, their orthogonality with respect to the density $f$
 is a consequence of the Sturm-Liouville theory.
 Specifically, it is shown in Theorem 1
 of \cite{DZ} (see p.\ 295)
 %
 %\cite[Theorem 1, p.\ 295]{DZ}
 %showed
 that each polynomial $h_k$
 satisfies the equation
 \be
 \label{sl}
 [f(x)p_2(x) h_k'(x)]'
 =k({a}_1+(k+1){b}_2) f(x) h_k(x),\quad \alpha<x<\omega,\quad
 k=0,1,2,\ldots \  .
 \ee
 An adaption of the Diaconis-Zabell approach to the present
 general case reveals that the orthogonality is valid only
 when a number of regularity conditions is satisfied.
 It will be proved here that these regularity conditions
 consist of an equivalent definition of
 the Integrated Pearson system. In fact,
 it will be shown that the Rodrigues polynomials
 (\ref{polynomials}) are orthogonal with respect to the
 corresponding density $f$ if and only if this
 $f$ belongs to Integrated Pearson family,
 provided that we have chosen a correct $p_2$
 in the differential equation (\ref{Pearson2}),
 i.e.\ provided that $p_2=q/\theta$ for some $\theta\neq 0$.
 We mention here that, even for Integrated Pearson densities,
 a wrong choice of $p_2$ results in
 non-orthogonality of the
 Rodrigues polynomials; see, e.g.,
 the polynomials $h_k=P_k^2$ given in
 \cite{DZ}, p.\ 297,
 %\cite[p.\ 297]{DZ}
 for the Beta-type density $f(x)=Cx^N$, $0<x<x_0$.
 In light
 of Proposition \ref{prop.alg2} (and Table \ref{table densities}),
 a correct choice for this density is given by
 $p_2=x(x_0-x)$.

 In order to discuss the orthogonality of $h_k$
 we first show the following lemma.
 \begin{lem}
 \label{lem.orth}
 Let $f$ be a density satisfying (\ref{Pearson2})
 and for fixed $k,m\in\{0,1,\ldots\}$, $k\neq m$,
 consider the polynomials $h_k$ and $h_m$,
 given by (\ref{polynomials}). Assume that
 \smallskip

 \noindent
 (a) The density $f$ process a suitable
 number of moments so that
 \smallskip
 $\int_{\alpha}^{\omega} |h_k(t)
 h_m(t)| f(t)\ud{t}<\infty$;
 %equivalently, for
 %$r=\deg(h_k)+\deg(h_m)$,
 %\[
 %\int_{\alpha}^{\omega} |t|^r f(t)\ud{t}<\infty;
 %\]
 %\item[(a)] The interval
 %$(\alpha,\omega)$ is a maximal subinterval of
 %$\{x:p_2(x)\neq 0\}$ -- cf.\ Remark \ref{rem.maximal};

 \noindent
 (b) ${a}_1+(k+m+1)b_2 \neq 0$;
 %for $k_1\neq k_2$;
 %\item[(b)] The function $f$ must be positive and integrable in
 %$(\alpha,\omega)$, so that its integral can be taken to be $1$
 %-- a density;
 %\item[(d)] The density $f$ must process a suitable
 %number of moments;
 \smallskip

 \noindent
 (c)
 %\mbox{~}
 %\vspace{-2em}
 %\begin{eqnarray*}
 %&& \hspace{-10ex}
 $\displaystyle\lim_{x\nearrow\omega} \{ p_2(x)f(x)
 [h_k'(x)h_m(x)-h_k(x) h_m'(x)]\}=
 \smallskip
 \lim_{x\searrow\alpha} \{ p_2(x)f(x)[h_k'(x)h_m(x)-h_k(x)h_m'(x)]\}.$
 %\vspace{.5ex}
 %\\
 %&&
 %
 %\mbox{~}\hspace{25ex}$=\lim_{x\searrow\alpha} \{ p_2(x)f(x)[h_k'(x)h_m(x)-h_k(x)h_m'(x)]\}.$
 %\end{eqnarray*}

 \noindent
 Then,
 %$h_k$ and $h_m$ are orthogonal
 %to each other with respect to $f$,
 %that is,
 \[
 \int_{\alpha}^{\omega} h_k(x)h_m(x)f(x)\ud{x}=0.
 \]
 [We shall show that, under (a) and (b), both limits in (c) exist
 (in $\R$), but it is not guaranteed that they are equal; in fact,
 their difference equals to
 $(k-m)(a_1+(k+m+1)b_2)\times\int_{\alpha}^{\omega}h_k(t)h_m(t)f(t)\ud{t}$.]
 \end{lem}
 \begin{proof}
 Multiply both hands of (\ref{sl}) by $h_m$,
 interchange the roles of $k$ and $m$ and subtract the
 resulting equations to get
 \be
 \label{a}
 \lambda h_k(t)h_m(t)f(t)
 =h_m(t)[f(t)p_2(t)
 h_k'(t)]'-h_k(t)[f(t)p_2(t) h_m'(t)]',
 \ \ \ \alpha<t<\omega,
 \ee
 where $\lambda=(k-m)(a_1+(k+m+1)b_2)\neq 0$, by (b).
 Now, it is easy to verify the Lagrange identity:
 \be\label{b}
 \begin{split}
 \big\{\left[f(t)p_2(t)h_k'(t)\right]h_m(t)-\big[f&(t)p_2(t)h_m'(t)\big]h_k(t)\big\}'\\
 &=h_m(t)[f(t)p_2(t)h_k'(t)]'-h_k(t)[f(t)p_2(t) h_m'(t)]'.
 \end{split}
 \ee
 Thus, integrating (\ref{a}) over
 $[x,y]\subseteq(\alpha,\omega)$,
 and in view of (\ref{b}),
 we conclude that
 \[
 \begin{split}
 \int_{x}^{y}h_k(t)h_m(t)f(t) \ud{t}
 =& \ \frac{1}{\lambda}p_2(y)f(y)[h_k'(y)h_m(y)-h_k(y) h_m'(y)]
 \\
 &-\frac{1}{\lambda}p_2(x)f(x)[h_k'(x)h_m(x)-h_k(x)h_m'(x)].
 \end{split}
 \]
 Therefore, taking limits as
 $x\searrow\alpha$ and $y\nearrow\omega$
 and using (a) and (c) we get the result.
 Working as in the proof of Proposition
 \ref{prop.2} it is easily seen that
 both limits in (c) exist in $\R$, whenever (a) and (b)
 hold. In fact, it is true that under (a),
 \be\label{limits}
 \begin{split}
 (k-m)(a_1+(k+m+&1)b_2)\int_{\alpha}^{\omega} h_k(t)h_m(t)f(t)\ud{t}\\
 &=\lim_{y\nearrow\omega} \{ p_2(y)f(y)[h_k'(y)h_m(y)-h_k(y) h_m'(y)]\}\\
 &\hspace{10ex}-\lim_{x\searrow\alpha} \{ p_2(x)f(x)[h_k'(x)h_m(x)-h_k(x)
 h_m'(x)]\}.
 \end{split}
 \ee
 \vspace*{-1.3cm}

 \end{proof}

 \noindent
 The following result is an immediate consequence
 of Lemma \ref{lem.orth}.
 \begin{theo}
 \label{theo.orthogonal}
 Let $f$ be a density in $(\alpha,\omega)$ which satisfies
 (\ref{Pearson2}).
 For some (fixed) $n\in\{1,2,\ldots\}$
 consider the set ${\cal H}_n:=\{h_0,h_1,\ldots,h_n\}$,
 formed by the first $n+1$ polynomials in
 (\ref{polynomials}).
 Then the set ${\cal H}_n$ is an orthogonal system
 (containing only non-zero elements) with respect
 to
 $f$ if and only if the following conditions
 are satisfied:
 \begin{itemize}
 \item[(i)] The density $f$ process $2n-1$
 finite moments;
 \item[(ii)] $\prod_{j=2}^{2n}({a}_1+j b_2)\neq 0$;
 \item[(iii)]
 $\lim_{x\nearrow\omega} x^j p_2(x)f(x)=\lim_{x\searrow\alpha}
 x^j p_2(x)f(x)$ for each $j\in\{0,1,\ldots,2n-2\}$.
 \end{itemize}
 \end{theo}
 \begin{proof} Let $X\sim f$ and assume first that
 (i)--(iii) are satisfied.
 Condition (ii) shows, in view of
 (\ref{lead1}), that $\deg(h_k)=k$ for all
 $k\in\{0,1,\ldots,n\}$. Fix $k,m\in\{0,1,\ldots,n\}$ with $m\neq k$.
 Since $\E|X|^{2n-1}<\infty$ by (i), it follows that
 $\E|h_k(X)h_m(X)|<\infty$, i.e.\
 the integral $\int_{\alpha}^{\omega} h_k(x)h_m(x)f(x)\ud{x}$ is
 (well-defined and) finite. Finally, since $h_k'h_m-h_k h_m'$
 is a polynomial of degree $k+m-1$ (observe that
 $\lead(h_k'h_m-h_k h_m')=(k-m)\lead(h_k)\lead(h_m)\neq 0$),
 (iii) ensures that assumption (c) of Lemma \ref{lem.orth}
 is also fulfilled and, hence,
 \[
 \int_{\alpha}^{\omega} h_k(x)h_m(x)
 %\times
 f(x)\ud{x}=0.
 \]

 Conversely, assume that the set ${\cal H}_n=\{h_0,h_1,\ldots,h_n\}$
 is orthogonal with respect to $f$; that is,
 $\E|h_k(X)h_m(X)|=\int_{\alpha}^\omega |h_k(x)h_m(x)|f(x) \ud{x}<\infty$
 for all $k,m\in\{0,1,\ldots,n\}$
 with $m\neq k$, and $\int_{\alpha}^{\omega}h_k(x)h_m(x)f(x)\ud{x}=0$.
 It follows that, necessarily, $\deg(h_k)=k$ for all
 $k=1,2,\ldots,n$; for if $k$ is the smallest integer
 in $\{1,2,\ldots,n\}$ for which $\lead(h_k)=0$ then we can write
 $h_{k}(x)=\sum_{j=0}^{k-1}c_j h_j(x)$ for some
 constants $c_j$, and this implies that
 \[
 h_k^2(x)f(x)
 =\left|
 %\mbox{$
 \sum_{j=0}^{k-1}
 %$}
 c_j h_j(x)h_k(x)\right|f(x)
 \leq
 %\mbox{$
 \sum_{j=0}^{k-1}
 %$}
 |c_j| \ |h_j(x)h_k(x)|f(x).
 \]
 Subsequently, the inequality
 \[
 \int_{\alpha}^\omega h_k^2(x)f(x) \ud{x} \leq \sum_{j=0}^{k-1}|c_j|
 \int_{\alpha}^\omega |h_k(x)h_j(x)|f(x) \ud{x}<\infty
 \]
 shows that $h_k\in L_f^2(\alpha,\omega)$ and, finally,
 \[
 \int_{\alpha}^\omega h_k^2(x)f(x) \ud{x} = \sum_{j=0}^{k-1}c_j
 \int_{\alpha}^\omega h_k(x)h_j(x)f(x) \ud{x}=0,
 \]
 by
 the orthogonality assumption. Since $h_k$ is
 continuous (a polynomial) and $f$ is  positive
 in a subinterval of $(\alpha,\omega)$ with positive
 length, it follows that $h_k\equiv 0$, which contradicts
 %(\ref{polynomials}).
 the assumption that ${\cal H}_n$ contains only non-zero
 elements. Therefore,
 $\prod_{k=0}^n\lead(h_k)\neq 0$,
 %for all $k\in\{1,2,\ldots,n\}$,
 and (\ref{lead1}) yields (ii). Obviously,
 $\E|h_n(X) h_{n-1}(X)|<\infty$ is equivalent to
 $\E|X|^{2n-1}<\infty$ and (i) follows.
 Since $g_{k,m}=h_k'h_m-h_k h_m'$
 is a polynomial of degree exactly $k+m-1$ (for $k\neq m$),
 we can form a linearly independent
 set
 \[
 \{g_0,g_1,\ldots,g_{2n-2}\}\subseteq \{g_{k,m}: k,m=0,1,\ldots,n, \ k\neq m\},
 \]
 with $\deg(g_j)=j$ for each $j$.
 Applying (\ref{limits})
 inductively to $g_0,g_1,\ldots,g_{2n-2}$
 we get (iii).
 \end{proof}
 \begin{exam}
 \label{ex.zero}
 It may happen that $h_k\equiv 0$ for all $k\geq 1$. For instance
 consider the density $f(x)=C/x$, $1<x<2$; this density satisfies
 (\ref{Pearson2}) with $(p_1,p_2)=(-1,x)$. Although
 $\int_{1}^{2} h_k h_m f=0$ for $m\neq k$, the
 trivial system ${\cal H}_n=\{1,0,\ldots,0\}$
 is not considered as orthogonal in this case.
 Condition (ii) of Theorem \ref{theo.orthogonal}
 eliminates such trivial cases.
 \end{exam}

 \begin{exam}
 \label{ex.nonorth}
 The density $f(x)=\frac{3}{2}x^2$,
 $-1<x<1$, satisfies (\ref{Pearson2}) in $(\alpha,\omega)=(-1,1)$. The
 choice $(p_1,p_2)=(2,x)$ leads to constant polynomials,
 $h_k\equiv(k+2)!/2$. A set $\{h_k,h_m\}$ can never be
 orthogonal; this explains that condition (b)
 of Lemma \ref{lem.orth} is necessary.
 On the other hand, the choice $(p_1,p_2)=(2x,x^2)$
 yields the polynomials $h_k=c_k x^k$ with $c_k=(2k+2)!/(k+2)!$.
 The limits in Lemma \ref{lem.orth}(c) are $\frac{3}{2}c_kc_m(k-m)$
 and $\frac{3}{2}c_kc_m(k-m)(-1)^{k+m+1}$; they are equal
 if and only if $k+m$ is odd, in which case $h_k$ and $h_m$
 are, obviously, orthogonal. Clearly, any set containing three
 (or more) polynomials cannot be an orthogonal set.
 \end{exam}

 \begin{rem}
 \label{rem.question}
 While the density $f$ of Example \ref{ex.nonorth}
 satisfies the (generalized) Pearson differential equation
 (\ref{Pearson2}) and has finite moments of any order,
 the system $\{h_0,h_1,h_2\}$ fails to be orthogonal.
 The same is true for the Pearson density
 \[
 f(x)=\frac{C}{\sqrt{1+x^2}}, \ \
 -\infty<\alpha<x<\omega<\infty.
 \]
 Now $(p_1,p_2)=(-x,1+x^2)$
 and $\{h_0,h_1,h_2\}=\{1,x,3+6x^2\}$ so that $h_0h_2\geq 3$
 and the system $\{h_0,h_1,h_2\}$ cannot be orthogonal
 (with respect to any measure).
 Does this happen because these $f$ lie outside
 the Integrated Pearson family? In other words,
 it is natural to state the following question:
 \begin{quote}
 If a density $f$
 %(with $S(f)\subseteq(\alpha,\omega)$)
 has finite moments up to order $2n-1$
 (for some fixed $n\geq 2$) and satisfies
 (\ref{Pearson2}),
 and if the system $\{h_0,h_1,\ldots,h_n\}$
 of the first $n+1$ Rodrigues polynomials
 %obtained by
 %(\ref{polynomials}))
 is orthogonal with respect to $f$, does it follow
 that this $f$ belongs to the Integrated Pearson family?
 \end{quote}
 \end{rem}

 The answer is in the affirmative.
 In particular, the following result holds.
 \begin{theo}
 \label{theo.orth2}
 Assume that a differentiable density
 $f$ with $S(f)=\{x:f(x)>0\}\subseteq(\alpha,\omega)$
 has finite third moment and satisfies (\ref{Pearson2}).
 Let $h_0\equiv 1,h_1,h_2$ be the first three Rodrigues polynomials
 given by (\ref{polynomials}),
 consider the system ${\cal H}_2=\{h_0,h_1,h_2\}$ and assume that
 ${\cal H}_2$ is non-trivial, i.e.,
 $h_1\not\equiv0$ and $h_2\not\equiv0$.
 If the system ${\cal H}_2$ is orthogonal with respect to
 $f$ then there exists a subinterval
 $({\alpha}',{\omega}')\subseteq(\alpha,\omega)$,
 a quadratic polynomial
 \[
 q(x)=\delta x^2+\beta x+\gamma, \ \
 \mbox{with} \ \
 \{x:q(x)>0\}=({\alpha}',{\omega}'),
 \]
 and a number $\mu\in({\alpha}',{\omega}')$
 such that $f\sim\IPq\equiv\IP$. Moreover, there exists a
 constant $\theta\neq 0$ such that $q(x)=\theta p_2(x)$,
 $x\in\R$.
 \end{theo}
 \begin{proof}
 In view of Theorem
 \ref{theo.orthogonal} and the fact that $f$ has finite third moment,
 the orthogonality assumption
 is equivalent to
 \be
 \label{non.trivial}
 (a_1+2b_2)(a_1+3b_2)(a_1+4b_2)\neq 0
 \ee
 and
 \be
 \label{lim.eq}
 L_j(\alpha)=L_j(\omega),\quad   j=0,1,2,
 \ee
 where
 \be
 \label{lim}
 L_j(\alpha):=\lim_{x\searrow \alpha} x^j p_2(x)f(x), \ \ \
 L_j(\omega):=\lim_{x\nearrow \omega} x^j p_2(x)f(x).
 \ee
 To simplify cases we
 can apply an affine transformation
 $x\mapsto \lambda x+c$ ($\lambda\neq 0$, $c\in\R$) to $f$.
 By considering
 $\widetilde{f}(x)=\frac{1}{|\lambda|}f(\frac{x-c}{\lambda})$
 in place of $f$ it is easily seen that (\ref{Pearson2})
 is satisfied in the translated interval
 $(\widetilde{\alpha},\widetilde{\omega})$ for
 $\widetilde{p}_1(x)=\lambda p_1(\frac{x-c}{\lambda})$
 and
 $\widetilde{p}_2(x)=\lambda^2 p_2(\frac{x-c}{\lambda})$;
 since $\widetilde{a}_1=a_1$ and $\widetilde{b}_2=b_2$,
 (\ref{non.trivial}) remains unchanged. Obviously
 $f$ has finite third moment if and only if $\widetilde{f}$
 does. Moreover, it is easily seen from
 (\ref{polynomials}) that the translated polynomials $\widetilde{h}_k$
 are related to $h_k$ by $\widetilde{h}_k(x)=\lambda^k h_k(\frac{x-c}{\lambda})$;
 thus, $\lead(\widetilde{h}_k)=\lead(h_k)$ and, in particular,
 the system ${\cal H}_2$ is non-trivial if and only if the same is true for
 the system
 ${\cal\widetilde{H}}_2
 :=\{\widetilde{h}_0,\widetilde{h}_1,\widetilde{h}_2\}$.
 The orthogonality of the system ${\cal\widetilde{H}}_2$ with respect to
 $\widetilde{f}$ is equivalent to the orthogonality  of the system
 ${\cal{H}}_2$ with respect to $f$; indeed,
 \[
 \int_{\widetilde{\alpha}}^{\widetilde{\omega}}
 \widetilde{h}_k (x)  \widetilde{h}_m (x)  \widetilde{f}
 (x) \ud{x}= \lambda^{k+m}
 \int_{{\alpha}}^{{\omega}}
 {h}_k (x) {h}_m (x)  {f} (x) \ud{x}.
 \]
 It remains to verify that (\ref{lim.eq}) are equivalent to
 $\widetilde{L}_j(\widetilde{\alpha})=\widetilde{L}_j(\widetilde{\omega})$
 $(j=0,1,2)$,
 where $\widetilde{L}_j(\widetilde{\alpha}):= \lim_{x\searrow \widetilde{\alpha}}
 x^j \widetilde{p}_2(x)\widetilde{f}(x)$,
 $\widetilde{L}_j(\widetilde{\omega}):=
 \lim_{x\nearrow \widetilde{\omega}} x^j
 \widetilde{p}_2(x)\widetilde{f}(x)$.
 To this end, it
 suffices to observe the relations
 %for $\lambda>0$,
 \[
 \begin{split}
 &\sum_{i=0}^j
 {j\choose i} \lambda^{i+1} c^{j-i} L_i(\alpha)
 =
 \left\{
 \begin{array}{cc}
 \widetilde{L}_j(\widetilde{\alpha}), &
 \textrm{if } \ \lambda>0,
 \\
 -\widetilde{L}_j(\widetilde{\omega}), &
 \textrm{if } \ \lambda<0,
 \end{array}
 \right.
 \\
 &\hspace{15ex}
 \sum_{i=0}^j
 {j\choose i} \lambda^{i+1} c^{j-i} L_i(\omega)
 =
 \left\{
 \begin{array}{cc}
 \widetilde{L}_j(\widetilde{\omega}), &
 \mbox{if } \ \lambda>0,
 \\
 -\widetilde{L}_j(\widetilde{\alpha}), &
 \mbox{if } \ \lambda<0.
 \end{array}
 \right.
 \end{split}
 \]
 Thus, it is easily seen that  $L_j(\alpha)=L_j(\omega)$ ($j=0,1,2$)
 if and only if $\widetilde{L}_j(\widetilde{\alpha})
 =\widetilde{L}_j(\widetilde{\omega})$ ($j=0,1,2$).

 It is clear from the above considerations, and in view of
 Proposition \ref{prop.1}(vi), that we can freely
 apply any affine transformation, either to the polynomial
 $p_2$ or to the density $f$ and its support
 $(\alpha,\omega)$; under such
 transformations, the conclusions as well as
 the assumptions of our theorem remain unchanged.

 The rest of the proof is easy but tedious since we just have
 to examine all possible non-equivalent cases
 by solving the differential equation (\ref{Pearson2})
 in each case.
 We shall try to give a somewhat complete approach as
 follows:

 Assume first that $\deg(p_2)=2$ and that its discriminant,
 $\Delta$, is strictly negative. Applying an affine
 transformation and dividing both $p_1$ and $p_2$ by
 $\lead(p_2)\neq 0$ we may assume that $p_2=x^2+b_0$ for
 some $b_0>0$. If $p_1\equiv 0$ then, necessarily,
 $(\alpha,\omega)$ is finite and $f\sim U(\alpha,\omega)$;
 but in this case, $h_2(x)=6x^2+4 b_0\geq 4b_0>0$
 cannot be orthogonal to $h_0\equiv 1$.
 If $\deg(p_1)=0$, that is,
 $p_1\equiv a_0\neq 0$, then the density
 \[
 f(x)=C\exp\left(\frac{a_0}{b_0} \tan^{-1}\left(\frac{x}{\sqrt{b_0}}\right)\right)
 \]
 is bounded away from zero,
 so that $(\alpha,\omega)$ must be again finite.
 Then, the assumed orthogonality of ${\cal H}_2$
 fails because (\ref{lim.eq}) shows that $\alpha=\omega$.
 Finally, assume that $\deg(p_1)=1$ i.e.\ $p_1=a_0+a_1 x$
 with $a_1\neq 0$. In this case, $a_1\not\in\{-2,-3,-4\}$
 because of (\ref{non.trivial}). Since
 \[
 f(x)=C(x^2+b_0)^{\frac{a_1}{2}}\exp\left(\frac{a_0}{b_0}
 \tan^{-1}\left(\frac{x}{\sqrt{b_0}}\right)\right),
 \]
 it follows that either $(\alpha,\omega)$ is finite or, otherwise,
 $a_1<-4$ (for the third moment to exists).
 If $(\alpha,\omega)$ is finite, the assumed orthogonality
 fails because (\ref{lim.eq}) shows that $\alpha=\omega$.
 If $\alpha>-\infty$, $\omega=\infty$ then
 the assumed orthogonality fails again
 from (\ref{lim.eq}) since $L_0(\alpha)>0$,
 $L_0(\infty)=0$. The case $\alpha=-\infty$,
 $\omega<\infty$ is similar to the previous one
 (we can also make the transformation $x\mapsto -x$).
 Therefore, the unique case where ${\cal H}_2$
 is indeed orthogonal is when $(\alpha,\omega)=\R$.
 Then,
 \[
 \E h_1(X)=(a_1+2b_2)\mu+(a_0+2 b_1)=0 \ \
 \mbox{implies that} \ \
 \mu=\frac{a_0}{-2-a_1}
 \]
 (note that $b_2=1$, $b_1=0$)
 and, hence, $p_1+p_2'=(-2-a_1)(\mu-x)$. In view of
 Proposition \ref{prop.2} we see
 that
 \[
 f\sim\IPq
 \ \
 \mbox{with}
 \ \
 \mu=\frac{a_0}{-2-a_1}
 \ \
 \mbox{and}
 \ \
 q(x)=\frac{x^2+b_0}{-2-a_1}=\frac{p_2(x)}{-2-a_1}.
 \]

 Next, assume that $\deg(p_2)=2$ and $\Delta=0$.
 Applying an affine transformation and dividing both
 $p_1$ and $p_2$ by
 $\lead(p_2)\neq 0$ we may further assume that $p_2=x^2$.
 If $p_1\equiv 0$ then, necessarily,
 $(\alpha,\omega)$ is finite and $f\sim U(\alpha,\omega)$;
 but in this case, $h_2(x)=12x^2\geq0$
 cannot be orthogonal to $h_0\equiv 1$.
 Let $\deg(p_1)=0$, that is,
 $p_1\equiv a_0\neq 0$.
 With the map $x\mapsto-x$, if necessary,
 we may further translate the density to
 have either
 the form
 \[
 f(x)=C e^{-\frac{a_0}{x}},
 \ \
 0\leq \alpha<x<\omega<\infty,
 \]
 or the form
 \[
 f(x)=\left\{
 \begin{array}{ll}
 C_1 e^{-\frac{a_0}{x}},
 & 0<x<\omega<\infty,
 \\
 0, & -\infty\leq \alpha<x\leq 0,
 \end{array}
 \right.
 \]
 %and
 %$f(x)=0$, $-\infty\leq \alpha<x\leq0$,
 where, necessarily, $a_0>0$ in the second case.
 %if $\alpha\leq 0$.
 In both cases the assumed
 orthogonality fails because of (\ref{lim.eq}).
 Finally, assume that $\deg(p_1)=1$ i.e.\ $p_1=a_0+a_1 x$
 with $a_1\neq 0$. In this case, $a_1\not\in\{-2,-3,-4\}$
 because of (\ref{non.trivial}).
 With the map $x\mapsto-x$, if necessary,
 we may further translate the density to
 have either
 the form
 \[
 f(x)=C x^{a_1}e^{-\frac{a_0}{x}},
 \ \ 0\leq \alpha<x<\omega\leq \infty,
 \]
 or the form
 \[
 f(x)
 =\left\{
 \begin{array}{ll}
 C_1 x^{a_1}e^{-\frac{a_0}{x}},
 &
 0<x<\omega\leq \infty,
 \\
 0, & -\infty\leq \alpha<x\leq0,
 \end{array}
 \right.
 \]
 %and
 %$f(x)=0$, $-\infty\leq \alpha<x\leq0$,
 where, necessarily, $a_0>0$ in the second case.
 %whenever $\alpha\leq 0$.
 If $\omega<\infty$
 then, due to (\ref{lim.eq}),
 the assumed orthogonality fails for both cases.
 If $\omega=\infty$ and $\alpha>0$ then we must take
 $a_1<-4$ for the finiteness of the third moment
 (note that in this case, $a_0\in\R$ can be arbitrary since
 $\alpha>0$), but the orthogonality fails because of
 (\ref{lim.eq}), since $L_0(\alpha)>0$, $L_0(\infty)=0$.
 In the last case where $\alpha\leq 0$ and $\omega=\infty$
 (thus, $a_0>0$ and $a_1<-4$) the orthogonality
 is indeed satisfied. This is so because it is easy
 to verify both (\ref{lim.eq}) and (\ref{non.trivial}).
 On the other hand, since we have assumed that
 $\E h_1(X)=(a_1+2b_2)\mu+(a_0+b_1)=0$, it follows that
 $\mu=\frac{a_0}{-a_2-2}$ (note that $b_2=1$, $b_1=0$)
 and $p_1+p_2'=(-2-a_1)(\mu-x)$. In view of Proposition
 \ref{prop.2}, this density belongs to the Integrated
 Pearson system with
 \[
 \mu=\frac{a_0}{-a_2-2}
 \ \
 \mbox{and}
 \ \
 q(x)=\frac{p_2(x)}{-2-a_1}=\frac{x^2}{-2-a_1}.
 \]
 Moreover, observe that its support,
 $(\alpha',\omega)
 =(0,\infty)\subseteq (\alpha,\omega)$, is
 different than $(\alpha,\omega)$, whenever
 $\alpha<0$.

 Next, assume that $\deg(p_2)=2$ and $\Delta>0$.
 Applying an affine transformation and dividing both
 $p_1$ and $p_2$ by
 $\lead(p_2)\neq 0$ we may further assume that
 $p_2=x(1-x)$.
 Solving the differential
 equation (\ref{Pearson2}) for arbitrary $p_1$
 and for all $x\in\R\smallsetminus\{0,1\}$ we see that
 the general solution has the form
 \[
 f(x)=\left\{
 \begin{array}{ll}
 C_1 (-x)^A (1-x)^B, & \mbox{if } \ x<0,  \\
 C_2 x^A (1-x)^B,   & \mbox{if } \ 0<x<1, \\
 C_3 x^A (x-1)^B,   & \mbox{if } \ x>1, \\
 \end{array}
 \right.
 \]
 where $A$ and $B$ are arbitrary parameters and
 $C_1,C_2,C_3\geq 0$ are arbitrary constants,
 not all zero.
 The restrictions on $A$ and $B$ depend on the
 interval $(\alpha,\omega)$ that we consider
 and the positivity or vanishing of each branch;
 they have to be chosen in such a way that the resulting
 function is differentiable and integrable in
 $(\alpha,\omega)$.
 For example, if $[0,1]\subseteq (\alpha,\omega)$
 and $C_1,C_2,C_3>0$ then, in order that $f$ is
 (continuous and) differentiable at the points $0$ and $1$,
 we must take $A>1$ and $B>1$; but then it is necessary for
 $(\alpha,\omega)$ to be bounded, since, otherwise, the resulting
 $f$ could not be integrable.
 The several possibilities can be classified according to
 the number of roots of $p_2$ that fall into
 $(\alpha,\omega)$, as follows:

 (1) Let $\{0,1\}\cap (\alpha,\omega)=\varnothing$. Then,
 either $(\alpha,\omega)\subseteq (0,1)$ or
 $(\alpha,\omega)\subseteq (-\infty,0)$ or
 $(\alpha,\omega)\subseteq (1,\infty)$.
 For the first case we observe that
 (\ref{lim.eq}) fails whenever $\alpha>0$ or $\omega<1$;
 if $(\alpha,\omega)=(0,1)$ then $A>-1$, $B>-1$,
 $p_1=A-(A+B)x$ and the
 orthogonality assumption yields
 $\E h_1(X)=(a_1+2b_2)\mu+(a_0+b_1)=-(A+B+2)\mu+(A+1)=0$.
 Hence, $p_1+p_2'=A+1-(A+B+2)x=(A+B+2)(\mu-x)$ and
 $p_2(x)f(x)=Cx^{A+1}(1-x)^{B+1}\to 0$ as $x\nearrow 1$;
 thus, Proposition \ref{prop.2} shows that $f\sim\IPq$
 with
 \[
 \mu=\frac{A+1}{A+B+2}
 \ \
 \mbox{and}
 \ \
 q(x)=\frac{p_2(x)}{A+B+2}.
 \]

 Using the map $x\mapsto-x$ for the second case and the map
 $x\mapsto x-1$ for the third case it is seen that both cases
 are reduced to $(\alpha,\omega)\subseteq(0,\infty)$ and
 translate $p_2$ to $p_2=-x(x+1)$;  equivalently, we can
 take $p_2=x(x+1)$.
 Moreover, the general solution in this case
 takes the form
 \[
 f(x)=Cx^{\theta}
 (x+1)^{\lambda},
 \ \
 0\leq\alpha<x<\omega\leq \infty.
 \]
 If
 $\omega<\infty$ or $\alpha>0$ it is easily seen that
 (\ref{lim.eq}) fails. In the remaining case where
 $(\alpha,\omega)=(0,\infty)$ we must have $\theta>-1$ (for
 integrability close to zero) and $\theta+\lambda<-4$
 (for finiteness of the third moment). Since
 $p_1=a_0+a_1x=\theta+(\theta+\lambda)x$, $p_2=b_0+b_1
 x+b_2 x^2=x+x^2$ and $h_1=(a_1+2b_2)x+(a_0+b_1)
 =(\theta+\lambda+2)x+(\theta+1)$, the assumed
 orthogonality yields
 $\E h_1(X)=(\theta+\lambda+2)\mu+\theta+1=0$;
 thus,
 $p_1+p_2'=(\theta+\lambda+2)x+(\theta+1)
 =-(\theta+\lambda+2)(\mu-x)$
 and Proposition \ref{prop.2} shows that
 \[
 f\sim\IPq
 \ \
 \mbox{with}
 \ \
 \mu=\frac{\theta+1}{-(\theta+\lambda+2)}
 \ \
 \mbox{and}
 \ \
 q(x)=\frac{p_2(x)}{-(\theta+\lambda+2)}.
 \]

 (2) Let $\{0,1\}\cap (\alpha,\omega)=\{1\}$ or
 $\{0,1\}\cap (\alpha,\omega)=\{0\}$,
 that is, $0\leq \alpha<1<\omega\leq\infty$ or $-\infty\leq \alpha<0<\omega\leq 1$.
 Clearly the map $x\mapsto 1-x$ translates the second case
 to the first one and leaves $p_2$ unchanged;
 thus, it suffices to consider only the first case.
 If $0<\alpha<1<\omega<\infty$ it is easily
 seen that (\ref{lim.eq}) fails for all choices of
 $(C_2,C_3)\in\{(+,+),(+,0),(0,+)\}$, where $(C_2,C_3)=(+,0)$
 means $C_2>0$, $C_3=0$, etc. If $\alpha=0$ and
 $1<\omega<\infty$
 then (\ref{lim.eq}) fails for all choices of
 $(C_2,C_3)\in\{(+,+),(0,+)\}$,
 while it is satisfied when $C_2>0$ and $C_3=0$.
 Similarly, if $0<\alpha<1$ and
 $\omega=\infty$ then (\ref{lim.eq})
 fails for all choices of $(C_2,C_3)\in\{(+,+),(+,0)\}$,
 while it is satisfied when $C_2=0$ and $C_3>0$.
 Finally, if $\alpha=0$ and $\omega=\infty$ then
 (\ref{lim.eq}) is satisfied for all choices
 of $(C_2,C_3)\in\{(0,+),(+,0)\}$, while  $C_2>0$, $C_3>0$
 is not a permissible choice because $f$
 is not integrable.
 Therefore, the two distinct situations
 where orthogonality can be
 verified are given by
 \[
 f_1(x)=\left\{
 \begin{array}{ll}
 C_2x^A(1-x)^B, & 0<x<1, \\
 0, & 1\leq x<\omega,
 \end{array}
 \right.
 \quad \textrm{and} \quad
 f_2(x)=\left\{
 \begin{array}{ll}
 0, & \alpha<x\leq 1, \\
 C_3x^A(x-1)^B, & 1< x<\infty,
 \end{array}
 \right.
 \]
 where $C_2>0$, $A>-1$, $B>1$ and $1<\omega\leq \infty$
 for $f_1$;
 $C_3>0$, $B>1$, $A+B<-4$ and $0\leq \alpha<1$ for $f_2$.
 Now it is easily seen that both $f_1$ and $f_2$ belong to
 the Integrated Pearson family. Specifically,
 Proposition \ref{prop.2} shows that $f_1\sim\IPq$
 with
 \be
 \label{eq.two}
 \mu=\frac{A+1}{A+B+2}
 \ \
 \mbox{and}
 \ \
 q(x)=\frac{p_2(x)}{A+B+2}=\frac{x(1-x)}{A+B+2},
 \ee
 while  $f_2\sim\IPq$ with
 \be
 \label{eq.three}
 \mu=\frac{A+1}{A+B+2}=1+\frac{B+1}{-A-B-2}
 \ \
 \mbox{and}
 \ \
 q(x)=\frac{p_2(x)}{A+B+2}=\frac{x(x-1)}{-A-B-2}.
 \ee

 (3) Let $\{0,1\}\subseteq (\alpha,\omega)$,
 that is, $-\infty\leq\alpha<0<1<\omega \leq\infty$.
 We have to study the following cases:
 (3a): $\alpha=-\infty$, $\omega=\infty$;
 (3b): $-\infty<\alpha<0$, $\omega=\infty$;
 (3b$'$): $\alpha=-\infty$, $1<\omega<\infty$;
 (3c):   $-\infty<\alpha<0$, $1<\omega<\infty$.
 Clearly the map $x\mapsto 1-x$ translates the case
 (3b$'$) to (3b) and leaves $p_2$ unchanged;
 thus, it suffices to consider only the cases (3a), (3b)
 and (3c).

 Assume first (3a). If
 $(C_1,C_2,C_3)\in\{(+,+,+),(+,0,+),(0,+,+)\}$ (where, e.g.,
 \linebreak $(C_1,C_2,C_3)=(+,0,+)$ means $C_1>0$, $C_2=0$,
 $C_3>0$ etc.) it follows that $A>1$ and $B>1$ and, thus,
 $f$ fails to be integrable (at a neighborhood of $+\infty$).
 The case
 $(C_1,C_2,C_3)=(+,+,0)$ is equivalent to $(C_1,C_2,C_3)=(0,+,+)$
 (by the map $x\mapsto 1-x$) and, again, $f$ fails to be
 integrable. By the same map, the cases $(+,0,0)$ and $(0,0,+)$
 are also equivalent. Assuming, e.g., $(C_1,C_2,C_3)=(0,0,+)$
 it is easily seen that $B>1$, $A+B<-4$ are necessary and
 sufficient for $f$ being integrable, differentiable
 at $0$ and $1$ and with finite third moment.
 In this case both (\ref{lim.eq}) and (\ref{non.trivial})
 are satisfied  so that
 the system $\{h_0,h_1,h_2\}$ is indeed orthogonal.
 Finally, if we assume that $(C_1,C_2,C_3)=(0,+,0)$
 then, necessarily, $A>1$, $B>1$ (for differentiability
 of $f$ at $0$ and $1$) and it follows that the system
 $\{h_0,h_1,h_2\}$ is indeed orthogonal, since both
 (\ref{lim.eq}) and (\ref{non.trivial})
 are satisfied.

 Next, assume  (3b). If
 $(C_1,C_2,C_3)\in\{(+,+,+),(+,0,+),(0,+,+)\}$
 it follows that $A>1$ and $B>1$ and, thus,
 $f$ fails to be integrable.
 If $(C_1,C_2,C_3)=(+,+,0)$ then $A>1$, $B>1$
 and (\ref{lim.eq}) fails. Also, if
 $(C_1,C_2,C_3)=(+,0,0)$ then $B>1$ and
 (\ref{lim.eq}) again fails.
 Assuming  $(C_1,C_2,C_3)=(0,0,+)$
 it is easily seen that $B>1$, $A+B<-4$ are necessary and
 sufficient for $f$ being integrable, differentiable
 at $0$ and $1$ and with finite third moment.
 In this case both (\ref{lim.eq}) and (\ref{non.trivial})
 are satisfied  so that
 the system $\{h_0,h_1,h_2\}$ is indeed orthogonal.
 Finally, if we assume that $(C_1,C_2,C_3)=(0,+,0)$
 then, necessarily, $A>1$, $B>1$ (for differentiability
 of $f$ at $0$ and $1$) and it follows that the system
 $\{h_0,h_1,h_2\}$ is indeed orthogonal, since both
 (\ref{lim.eq}) and (\ref{non.trivial})
 are satisfied.

 Finally, assume (3c). If
 $(C_1,C_2,C_3)\in\{(+,+,+),(+,0,+),(0,+,+),(+,+,0)\}$
 it follows that $A>1$ and $B>1$ and (\ref{lim.eq})
 fails.
 By the map $x\mapsto 1-x$ it is easily seen that
 the cases $(+,0,0)$ and $(0,0,+)$ are equivalent.
 Assuming, e.g., $(C_1,C_2,C_3)=(0,0,+)$
 it is easily seen that $B>1$ is necessary and
 sufficient for $f$ being integrable, differentiable
 at $0$ and $1$ and with finite third moment; but then,
 (\ref{lim.eq}) fails.
 Finally, if we assume that $(C_1,C_2,C_3)=(0,+,0)$
 then, necessarily, $A>1$, $B>1$ (for differentiability
 of $f$ at $0$ and $1$) and it follows that the system
 $\{h_0,h_1,h_2\}$ is indeed orthogonal, since both
 (\ref{lim.eq}) and (\ref{non.trivial})
 are satisfied.

 Therefore, the two distinct situations
 where orthogonality can be
 verified are given by
 \[
 f_1(x)=\left\{
 \begin{array}{ll}
 0, & \alpha<x\leq 0,
 \\
 C_2x^A(1-x)^B, & 0<x<1, \\
 0, & 1\leq x<\omega,
 \end{array}
 \right.
 \quad \textrm{and} \quad
 f_2(x)=\left\{
 \begin{array}{ll}
 0, & \alpha<x\leq 1, \\
 C_3x^A(x-1)^B, & 1< x<\infty,
 \end{array}
 \right.
 \]
 where $C_2>0$, $A>1$, $B>1$ and $-\infty\leq\alpha<0$,
 $1<\omega\leq \infty$
 for $f_1$;
 $C_3>0$, $B>1$, $A+B<-4$ and $-\infty\leq \alpha<0$
 for $f_2$.
 Now it is easily seen that both $f_1$ and $f_2$ belong to
 the Integrated Pearson family. Specifically,
 Proposition \ref{prop.2} shows that
 $f_1\sim\IPq$ with $\mu$ and $q$ as in
 (\ref{eq.two}), while
 $f_2\sim\IPq$ with $\mu$ and $q$ as in
 (\ref{eq.three}).
 %\[ \ \  \mbox{with}  \ \  \mu=\frac{A+1}{A+B+2}  \ \  \mbox{and}
 %\ \  q(x)=\frac{p_2(x)}{A+B+2}=\frac{x(1-x)}{A+B+2},
 % \]   while  \[   f_2\sim\IPq   \ \  \mbox{with}
 %\ \  \mu=\frac{A+1}{A+B+2}=1+\frac{B+1}{-A-B-2}  \ \
 %\mbox{and}   \ \  q(x)=\frac{p_2(x)}{A+B+2}=\frac{x(x-1)}{-A-B-2}.  \]

 Next, assume that $\deg(p_2)=1$ and, without loss of generality
 (by using an affine map) we shall further assume that $p_2=x$.
 If $p_1=a_0+a_1x$, the general solution of (\ref{Pearson2}) is
 \[
 f(x)=\left\{
 \begin{array}{cc}
 C_1x^{a_0}e^{a_1 x} & \mbox{if }\ x<0,
 \\
 C_2x^{a_0}e^{a_1 x} & \mbox{if }\ x>0,
 \end{array}
 \right.
 \]
 where $a_0$ and $a_1$ are arbitrary parameters and
 $C_1,C_2\geq 0$ are arbitrary constants,
 not both zero.
 The restrictions on $a_0$ and $a_1$ depend on the
 interval $(\alpha,\omega)$ that we consider
 and the positivity or vanishing of each branch;
 they have to be chosen in such a way that the resulting
 function is differentiable and integrable in
 $(\alpha,\omega)$.
 Assuming that $0<\alpha<\omega<\infty$ we readily see that
 any values of $a_0,a_1\in\R$  are admissible but
 (\ref{lim.eq}) fails. If $0<\alpha<\omega=\infty$ then
 either $a_1=0$ and $a_0<-3$ (for finiteness of the third
 moment) or $a_1<0$ and $a_0\in\R$. In the first case
 both (\ref{lim.eq}) and (\ref{non.trivial}) are violated:
 the limits are unequal although
 \[
 \int_{\alpha}^\infty h_k(x) h_m(x) f(x)\ud x=0
 \ \
 \mbox{for}
 \ \
 k\neq m,
 \ \
 k,m\in\{0,1,2\},
 \]
 because
 $h_1=h_2\equiv0$. In the second case, (\ref{lim.eq}) fails.
 If $\alpha=0<\omega<\infty$ then $a_0>-1$ and $a_1\in\R$;
 it follows that (\ref{lim.eq}) fails. Finally, if $\alpha=0$ and
 $\omega=\infty$ then $a_0>-1$ and $a_1<0$. In this case
 both (\ref{lim.eq}) and (\ref{non.trivial}) are satisfied
 and the system $\{h_0,h_1,h_2\}$ is, indeed, orthogonal.
 Also we see that $\E h_1(X)=a_1\mu+a_0+1=0$ so that
 $p_1+p_2'=a_1x+a_0+1=-a_1(\mu-x)$. Now, from Proposition
 \ref{prop.2} it follows that
 \be
 \label{eq.one}
 f\sim\IPq
 \ \
 \mbox{with}
 \ \
 \mu=\frac{a_0+1}{-a_1}
 \ \
 \mbox{and}
 \ \
 q(x)=\frac{x}{-a_1}=\frac{p_2(x)}{-a_1}.
 \ee
 By the map $x\mapsto-x$ we can transform the cases
 $-\infty\leq \alpha<\omega\leq0$ to the previous ones,
 since $p_2=x$ is transformed to $p_2=-x$. It remains to
 investigate the cases $-\infty\leq
 \alpha<0<\omega\leq\infty$; then, necessarily,
 $a_0>1$. Assuming that
 $-\infty<\alpha<0<\omega<\infty$ it is easily seen that
 (\ref{lim.eq}) fails for all choices of
 $(C_1,C_2)\in\{(+,+),(+,0),(0,+)\}$. Assuming that $\alpha=-\infty$,
 $\omega=\infty$ we see that for $f$ to be integrable
 it is necessary and sufficient that $a_1<0$ if $C_2>0$
 and $a_1>0$ if $C_1>0$; therefore, if $(C_1,C_2)=(+,+)$
 then $f$ is not integrable. The case $(C_1,C_2)=(+,0)$
 is transformed (by $x\mapsto-x$) to $(C_1,C_2)=(0,+)$.
 In the last case we can see that $a_0>1$ and $a_1<0$ are
 necessary and sufficient for $f$ to be differentiable
 (in $(\alpha,\omega)=\R$) and to have finite third moment.
 As before we can easily check that
 both (\ref{lim.eq}) and (\ref{non.trivial}) are satisfied,
 that $\{h_0,h_1,h_2\}$ is orthogonal and
 that $f\sim\IPq$ with $\mu$ and $q$ as in (\ref{eq.one}).
 %\[
 %f\sim\IPq
 %\ \
 %\mbox{with}
 %\ \
 %\mu=\frac{a_0+1}{-a_1},
 %\ \
 %q(x)=\frac{x}{-a_1}=\frac{p_2(x)}{-a_1}.
 %\]
 The map $x\mapsto-x$ shows that the last two cases, $\alpha=-\infty$,
 $0<\omega<\infty$,
 and $-\infty<\alpha<0$, $\omega=\infty$, are equivalent.
 By considering the second one we see that $a_0>1$
 and $a_1<0$ are necessary and sufficient for $f$ to be differentiable
 (in $(\alpha,\infty)$) and to have finite third moment.
 However, if $(C_1,C_2)\in\{(+,+),(+,0)\}$ it is easily
 seen that (\ref{lim.eq}) is violated because the limits
 as $x\searrow\alpha$ are nonzero. In the remaining case
 $(C_1,C_2)=(0,+)$ we can easily check, as before, that
 both (\ref{lim.eq}) and (\ref{non.trivial}) are satisfied,
 that $\{h_0,h_1,h_2\}$ is orthogonal and
 that $f\sim\IPq$ with
 $\mu$ and $q$ as in (\ref{eq.one}).
 %$\mu=\frac{a_0+1}{-a_1}$, $q(x)=\frac{x}{-a_1}=\frac{p_2(x)}{-a_1}$.

 Finally, assume that $\deg(p_2)=0$ or, equivalently,
 $p_2\equiv1$. Then, if $p_1=a_0+a_1x$, it follows that
 \[
 f(x)=C\exp(a_0x+a_1 x^2/2),
 \ \ \alpha<x<\omega.
 \]
 If the support $(\alpha,\omega)$ is bounded then it is
 easily seen that (\ref{lim.eq}) fails. The cases
 $-\infty<\alpha<\omega=\infty$ and $-\infty=\alpha<\omega<\infty$
 are, obviously, equivalent (by the map $x\mapsto-x$, which leaves
 $p_2$ unchanged). Assuming that $-\infty<\alpha<\omega=\infty$
 we see that either $a_1=0$, $a_0<0$ or $a_1<0$,
 $a_0\in\R$; in the first case both (\ref{lim.eq}) and
 (\ref{non.trivial}) fail, while (\ref{lim.eq}) fails in
 the second one. Finally, in the last remaining case where $(\alpha,\omega)=\R$
 we see that, necessarily, $a_1<0$. Then, for any value of
 $a_0\in\R$ we check that  both (\ref{lim.eq}) and (\ref{non.trivial})
 are satisfied so that $\{h_0,h_1,h_2\}$ is, indeed,
 orthogonal. Observe that, by assumption,
 $\E h_1(X)=a_1\mu+a_0=0$; thus, $p_1+p_2'=a_0+a_1
 x=-a_1(\mu-x)$. Proposition \ref{prop.2} shows
 that $f\sim\IPq$ with
 \[
 \mu=\frac{a_0}{-a_1}
 \ \
 \mbox{and}
 \ \
 q(x)=\frac{p_2(x)}{-a_1}=\frac{1}{-a_1};
 \ \
 \mbox{in fact},
 \ \
 f\sim N\left(a_0/(-a_1),(1/\sqrt{-a_1})^2\right).
 \]

 This subsumes all possible  cases and completes the proof.
 \end{proof}

 \section{Orthogonality of the Rodrigues-type polynomials and of
 their derivatives
 within the Integrated Pearson family}
 \label{sec5}
 Assume that $f$ is the density of a random
 variable $X\sim\IPq\equiv\IP$ with support
 $(\alpha,\omega)$.
 From Theorem \ref{theo.polynomials}
 it follows that the function
 \be
 \label{Rodrigues}
 P_k(x):=\frac{(-1)^k}{f(x)}\frac{\ud^k}{\ud{x}^k} [q^{k}(x)f(x)],
 \ \ \
 \alpha<x<\omega,
 \ \ \
 k=0,1,2,\ldots
 \ee
 is a polynomial with
 \be
 \label{lead}
 \deg(P_k)\leq k
 \ \
 \textrm{and}
 \ \
 \lead(P_k)=\prod_{j=k-1}^{2k-2}(1-j\delta):=c_k(\delta),
 \ \
 k=0,1,2,\ldots \ .
 \ee
 Obviously $c_0(\delta):=1$, i.e.\ an empty product should be
 treated as one.

 The polynomials $P_k$
 %defined by (\ref{Rodrigues})
 are
 special cases of the polynomials $h_k$
 defined by (\ref{polynomials});
 in fact, $P_k=(-1)^kh_k$.
 %, when $f\sim\IPq$.
 They are particularly important because under natural moment conditions
 they are, indeed, orthogonal with respect to the density $f$;
 see, e.g.,
 \cite{DZ} (pp.\ 295--296),
 %\cite[pp.\ 295--296]{DZ},
 \cite{John},
 \cite{Pap2},
 \cite{APP2}.
 The orthogonality follows immediately from
 Theorems \ref{theo.orthogonal} and
 \ref{theo.orth2}.
 %; see, also,
 %Diaconis and Zabell (1991, p.p.\ 295--296); Johnson (1993),
 %Papathanasiou (1995), Afendras et al.\ (2011).
 %This is true because, under natural moment conditions, they
 %satisfy the assumptions of Lemma \ref{lem.orth}
 %in Appendix \ref{app.c}.
 Moreover, the polynomials $P_k$ and their derivatives satisfy a number
 of useful properties that will be reviewed here.
 The first three are
 \be
 \label{three}
 \begin{split}
 &P_0(x)=1,\\
 &P_1(x)=x-\mu,\\
 &P_2(x)=(1-\delta)(1-2\delta)x^2-2(1-\delta)(\mu+\beta)x+\mu^2+\beta\mu-(1-2\delta)\gamma.
 \end{split}
 \ee

 An alternative simple proof of the orthogonality of the
 polynomials defined by (\ref{Rodrigues})
 can be derived by means of the following covariance
 identity, which extends Stein's identity
 for the Normal distribution and has independent interest
 in itself.
 %%%\pagebreak

 \begin{theo}[\textrm{\cite{APP2}, pp.\ 515--516}]
 %[\textrm{\cite{APP2}, pp.\ 515--516}]
 \label{theo.cov}
 Let $X\sim \IP\equiv\IPq$ with density $f$
 and support $(\alpha,\omega)$. Assume that $X$ has $2k$
 finite moments for some fixed $k\in\{1,2,\ldots\}$.
 Let $g:(\alpha,\omega)\to\R$ be any function such that
 $g\in C^{k-1}(\alpha,\omega)$, and assume that the function
 \[
 g^{(k-1)}(x):=\frac{d^{k-1}}{\ud{x}^{k-1}}g(x)
 \]
 is absolutely
 continuous in $(\alpha,\omega)$ with a.s.\ derivative $g^{(k)}$.
 %
 %\noindent
 If $\E q^k(X)|g^{(k)}(X)|<\infty$ then
 $\E|P_k(X)g(X)|<\infty$
 and the following covariance identity holds:
 \be
 \label{cov.id}
 \E P_k(X) g(X) = \E q^{k}(X)g^{(k)}(X).
 \ee
 \end{theo}

 \noindent
 It should be noted that when we claim that
 $h:(\alpha,\omega)\to\R$ is an absolutely continuous
 function with a.s.\ derivative $h'$ we mean that there exists
 a Borel measurable function $h':(\alpha,\omega)\to \R$ such that
 $h'$ is integrable in every finite subinterval
 $[x,y]$ of $(\alpha,\omega)$
 such that
 \[
 \int_{x}^y h'(t)\ud{t}=h(y)-h(x) \quad \textrm{for all}
 \quad  [x,y]\subseteq (\alpha,\omega).
 \]

 \begin{cor}[\textrm{\cite{APP2}, p.\ 516}]
 %[\textrm{\cite[pp.\ 516]{APP2}}]
 \label{cor.orth}
 Let $X\sim \IP\equiv\IPq$.
 %with density $f$
 %and support $(\alpha,\omega)$.
 Assume that for some
 $n\in\{1,2,\ldots\}$,
 $\E |X|^{2n}<\infty$ or, equivalently,
 $\delta<\frac{1}{2n-1}$.
 Then
 \be
 \label{orth}
 \begin{split}
 \E[P_k(X) P_m(X)] = \delta_{k,m} k! \E q^k(X)\!
 \prod_{j=k-1}^{2k-2}\!(1-j\delta) = \delta_{k,m} k! c_k(\delta)\E q^k(X),
 \\
 k,m\in\{0,1,\ldots,n\},
 \end{split}
 \ee
 where $\delta_{k,m}$ is Kronecker's delta and where an empty product
 should be treated as one.
 \end{cor}
 It should be noted that the orthogonality of $P_k$ and
 $P_m$, $k\neq m$, $k,m\in\{0,1,\ldots,n\}$,
 remains valid even if
 $\delta\in[\frac{1}{2n-1},\frac{1}{2n-2})$; in this case, however,
 $P_n\not\in L^2(\R,X)$ since $\lead(P_n)>0$ and $\E|X|^{2n}=\infty$.
 On the other hand, in view of Corollary \ref{cor.moments},
 the assumption $\E|X|^{2n}<\infty$ is
 equivalent to the condition $\delta<\frac{1}{2n-1}$.
 Therefore, for each $k\in\{0,1,\ldots,n\}$
 and for all $j\in\{k-1,\ldots,2k-2\}$,
 $1-j\delta>0$ because $\{k-1,\ldots,2k-2\}\subseteq\{0,1,\ldots,2n-2\}$.
 Thus, $c_k(\delta)>0$. Since $\Pr[q(X)>0]=1$, $\deg(q)\leq 2$ and
 $\E|X|^{2n}<\infty$ we conclude that $0<\E q^k(X)<\infty$ for all
 $k\in\{0,1,\ldots,n\}$. It follows that the set
 $\{\phi_0,\phi_1,\ldots,\phi_n\}\subset L^2(\R,X)$,
 where
 \be
 \label{orthonormal}
 \phi_k(x):=\frac{P_k(x)}{\left(k! c_k(\delta)
 \E q^k(X)\right)^{1/2}}=
 \frac{\frac{(-1)^k}{f(x)}\frac{\ud^k}{\ud x^k}[q^k(x)f(x)]}{\left(k!
 \E q^k(X)\prod_{j=k-1}^{2k-2}(1-j\delta)\right)^{1/2}}, \ \ k=0,1,\ldots,n,
 \ee
 is an orthonormal basis of all polynomials with degree at most $n$.
 Moreover, (\ref{lead}) shows that the leading coefficient is given by
 \be
 \label{lead.orth}
 \begin{split}
 \lead(\phi_k):=d_k(\mu;q)&=
 \left(\frac{\prod_{j=k-1}^{2k-2}(1-j\delta)}{k! \E q^k(X)}\right)^{1/2}\\
 &=\left(\frac{c_k(\delta)}{k! \E q^k(X)}\right)^{1/2}>0,
 \ \ \ k=0,1,\ldots,n.
 \end{split}
 \ee
 Let $X$ be any random variable with $\E|X|^{2n}<\infty$ and
 assume that the support of $X$ is not concentrated on a finite subset
 of $\R$. It is well known that we can always construct
 an orthonormal set of real polynomials up to order $n$.
 This construction is based on the first $2n$ moments of $X$
 and is a by-product of the Gram-Schmidt orthonormalization
 process, applied to the linearly independent system
 $\{1,x,x^2,\ldots,x^n\}\subset L^2(\R,X)$. The orthonormal
 polynomials are then uniquely defined, apart from the fact that
 we can multiply each polynomial by $\pm1$. It follows that
 the standardized Rodrigues polynomials $\phi_k$ of
 (\ref{orthonormal}) are the unique orthonormal polynomials
 that can be defined for a density $f\sim\IP$,
 provided that $\lead(\phi_k)>0$. Therefore,
 it is useful to express the $L^2$-norm of each $P_k$
 in terms of the parameters $\delta,\beta,\gamma$ and
 $\mu$ and, in view of (\ref{orth}) and (\ref{orthonormal}),
 it remains to obtain an expression for
 $\E q^k(X)$.
 To this end, we first
 recall a definition from \cite{PP}; cf. \cite{GR1}.

 \begin{defi}
 \label{def.star}
 Let $X\sim f$ and assume that $X$ has support
 $J(X)=(\alpha,\omega)$ and belongs to
 the integrated Pearson family, that is, $f\sim \IPq\equiv\IP$.
 Furthermore, assume that $\E X^2<\infty$ (i.e.\ $\delta<1$).
 Then we define $X^*$ to be the random variable
 with density $f^*$ given by
 \be
 \label{star}
 f^*(x):=\frac{q(x)f(x)}{\E q(X)}, \ \ \ \alpha<x<\omega.
 \ee
 \end{defi}
 Since $P_1=x-\mu$, setting $k=1$ in the covariance identity (\ref{cov.id})
 we get (see \cite{CP2}, \cite{PP})
 \be
 \label{cov.id2}
 \E[(X-\mu)g(X)]=\Cov[X,g(X)]=\E[q(X)g'(X)].
 \ee
 This identity is valid for all absolutely continuous functions
 $g:(\alpha,\omega)\to\R$ with a.s.\ derivative $g'$
 such that $\E q(X)|g'(X)|<\infty$.
 Thus, applying (\ref{cov.id2}) to the identity function
 $g(x)=x$ it is easily seen that $\E q(X)=\Var X=\sigma^2$,
 so that (cf. \cite{GR1})
 \[
 X^*\sim f^*(x)=\frac{1}{\sigma^2}q(x)f(x), \ \ \ \alpha<x<\omega.
 \]
 The following lemma shows that $X^*$ is integrated Pearson
 whenever $X$ is integrated Pearson and has finite third moment.
 \begin{lem}
 \label{lem.star}
 {\rm
 If $X\sim\IP\equiv\IPq$ with support $J(X)=(\alpha,\omega)$ and
 $\E|X|^3<\infty$ then $X^*\sim \IPqb$ with the same support
 $J(X^*)=J(X)=(\alpha,\omega)$,
 \be
 \label{eq.qstar}
 \mu^*=\frac{\mu+\beta}{1-2\delta},
 \quad \textrm{and} \quad q^*(x)=\frac{q(x)}{1-2\delta},
 \ \ \alpha<x<\omega.
 \ee
 }
 \end{lem}
 \begin{proof} From Corollary \ref{cor.moments}
 it follows that the assumption $\E|X|^3<\infty$ is equivalent to
 $\delta<\frac{1}{2}$. Let
 $X^*\sim f^*(x)=q(x)f(x)/\E q(X)=q(x)f(x)/\sigma^2$, $\alpha<x<\omega$,
 where
 $\sigma^2$ is the variance of $X$. Then, it follows that
 \[
 \mu^*=\E X^*=\frac{\E[Xq(X)]}{\sigma^2}.
 \]
 Define $P_1(x)=x-\mu$ and
 $P_2(x)=(x-\mu)^2-(x-\mu)q'(x)-(1-2\delta)q(x)$.
 We have
 $\E P_1(X)=0$ and $\E P_2(X)=\sigma^2-\Cov[X,q'(X)]-(1-2\delta)\E q(X)$.
 Applying the covariance identity (\ref{cov.id2}) to
 $g(x)=x$ and to $g(x)=q'(x)$ we see that
 $\E P_2(X)=2\delta\sigma^2-\E[q(X)q''(X)]=2\delta\sigma^2-2\delta\E q(X)=0$.
 Also,
 \begin{eqnarray*}
 \E [P_1(X)P_2(X)] & = & \E (X-\mu)^3-\E[(X-\mu)^2 q'(X)]-(1-2\delta)\E[(X-\mu)q(X)]
 \\
 &=&
 \Cov[X,(X-\mu)^2]-\Cov[X,(X-\mu)q'(X)]-(1-2\delta)\Cov[X,q(X)]
 \end{eqnarray*}
 and, once again, (\ref{cov.id2}) shows that $\E[P_1(X)P_2(X)]=0$.
 Now observe that
 \[
 x=\left(\frac{1}{2(1-\delta)(1-2\delta)}P_2(x)
 +\frac{\mu+\beta}{1-2\delta} P_1(x)\right)'=g'(x), \quad \textrm{say},
 \]
 so that
 \[
 \begin{split}
 \E X q(X)&=\E q(X)g'(X)=\Cov[X,g(X)]=\E (X-\mu)g(X)=\E P_1(X)g(X)
 \\
 &=\frac{1}{2(1-\delta)(1-2\delta)}
 \
 \E P_1(X)P_2(X)+\frac{\mu+\beta}{1-2\delta} \ \E P_1^2(X)
 \\
 &=0+\frac{\mu+\beta}{1-2\delta} \
 \E (X-\mu)^2=\frac{\mu+\beta}{1-2\delta} \ \sigma^2.
 \end{split}
 \]
 It follows that $\mu^*
 =\E[Xq(X)]/\sigma^2=(\mu+\beta)/(1-2\delta)$.

 It remains to show that $q^*(x)=q(x)/(1-2\delta)$ is the
 quadratic polynomial of $X^*$, i.e.\ that
 \[
 \int_{-\infty}^x (\mu^*-t)
 f^*(t)\ud{t}=\frac{1}{1-2\delta}\ q(x) f^*(x), \ \ \ x\in\R.
 \]
 Equivalently, it suffices to verify the identity
 \be
 \label{qstar}
 \int_{-\infty}^x \{\mu+\beta-(1-2\delta)t\}
 q(t)f(t)\ud{t}= q^2(x) f(x), \ \ \ x\in\R.
 \ee
 Since $f(x)=0$ for $x\notin (\alpha,\omega)$ it follows that the l.h.s.\
 of (\ref{qstar}) equals to zero for $x\leq \alpha$ (if
 $\alpha>-\infty$). Also, if $\omega<\infty$ and $x\geq \omega$ then the
 l.h.s.\ of (\ref{qstar}) is equal to
 $(\mu+\beta)\E
 q(X)-(1-2\delta)\E Xq(X) =(\mu+\beta)\sigma^2
 -(1-2\delta)\frac{\mu+\beta}{1-2\delta}\sigma^2
 =0$.
 Thus, (\ref{qstar}) takes the form $0=0$ whenever $x\notin
 (\alpha,\omega)$.
 For $x\in (\alpha,\omega)$ it is easily seen that
 \[
 \begin{split}
 \Big(&q^2(x)f(x)-\int_{-\infty}^x\{\mu+\beta-(1-2\delta)t\}q(t)f(t)\ud{t}\Big)'\\
 &=(q(x)\cdot q(x)f(x))'-\{\mu+\beta-(1-2\delta)x\}q(x)f(x)\\
 &=q'(x)q(x)f(x)+q(x)(\mu-x)f(x)-q(x)f(x)\{\mu+\beta-(1-2\delta)x\}\\
 &=q(x)f(x)\ [q'(x)+(\mu-x)-(\mu+\beta)+(1-2\delta)x]=q(x)f(x)\ [q'(x)-2\delta x-\beta]=0.
 \end{split}
 \]
 Thus, there exists a constant $c\in\R$ such that
 \be
 \label{eq.de}
 \int_{-\infty}^x
 \{\mu+\beta-(1-2\delta)t\}q(t)f(t)\ud{t}=q^2(x)f(x)+c, \ \ \
 \alpha<x<\omega.
 \ee
 Now observe that
 $\lim_{x\nearrow\omega}\int_{-\infty}^x \{\mu+\beta-(1-2\delta)t\}q(t)f(t)\ud{t}
 =\lim_{x\nearrow\omega} q^2(x)f(x)=0$. Indeed, the first
 limit follows from dominated convergence and the fact that
 $\E q(X)=\sigma^2$ and
 $\E[Xq(X)]=(\mu+\beta)\sigma^2/(1-2\delta)$,
 while the second one is obvious when $\omega<\infty$
 because $q(\omega)=0$ and $q(x)f(x)\to \E(\mu-X)=0$ as
 ${x\nearrow\omega}$.
 Finally, if $\omega=\infty$ we have
 $q(x)f(x)=o(x^{-2})$ as $x\to\infty$ because
 $\E|X|^3<\infty$ and for large enough $x$,
 \[
 x^2 q(x)f(x)=x^2 \int_{x}^{\infty} (t-\mu)f(t)\ud{t}
 \leq \int_{x}^{\infty} t^2 (t-\mu) f(t)\ud{t} \to 0,
 \ \
 \mbox{as}
 \ \ x\to\infty,
 \]
 by
 dominated convergence. This shows that $\lim_{x\nearrow\omega}q^2(x)f(x)=0$
 in all cases. Therefore,
 taking limits as $x\nearrow\omega$ in (\ref{eq.de}) we conclude
 that $c=0$ and (\ref{qstar})
 follows.
 \end{proof}

 \begin{theo}
 \label{theo.star}
 Let $X$ be a random variable with density
 $f\sim \IPq\equiv\IP$,
 supported in $J(X)=(\alpha,\omega)$.
 Furthermore, assume that $\E |X|^{2n+1}<\infty$ (i.e.\ $\delta<\frac{1}{2n}$)
 for some fixed $n\in\{0,1,\ldots\}$. Define the random variable
 $X_k$ with density $f_k$ given by
 \be
 \label{star2}
 f_k(x):=\frac{q^k(x)f(x)}{\E q^k(X)},
 \ \ \
 \alpha<x<\omega,
 \ \ \
 k=0,1,\ldots,n.
 \ee
 Then, $f_k\sim\IPqk$
 %and $X_k$ has the same support
 with (the same) support
 $J(X_k)=J(X)=(\alpha,\omega)$,
 \be
 \label{eq.qstar2}
 \mu_k=\frac{\mu+k\beta}{1-2k\delta},
 \ \
 \textrm{and}
 \ \
 q_k(x)=\frac{q(x)}{1-2k\delta},
 \ \
 \alpha<x<\omega,
 \ \
 k=0,1,\ldots,n.
 \ee
 Moreover, $X_0=X$, $X_1=X_0^*=X^*$, $X_2=X_1^*$ and, in
 general, $X_k=X_{k-1}^*$ for $k\in\{1,\ldots,n\}$.
 \end{theo}
 \begin{proof}
 For $k=0$ the assertion is obvious while for $k=1$ (and thus, $n\geq 1$)
 the assertion follows from Lemma \ref{lem.star} since $\E|X|^3<\infty$ and,
 by definition, $f_1=f^*$, $\mu_1=\mu^*$ and $q_1=q^*$.
 Assume now that the assertion has been proved for some
 $k\in\{1,\ldots,n-1\}$. Then,
 \[
 \E|X_k|^3=\frac{\E q^k(X)|X|^3}{\E q^k(X)}<\infty,
 \]
 because $\E |X|^{2k+3}<\infty$ since $k\leq n-1$.
 Therefore, we can apply Lemma \ref{lem.star} to the random
 variable
 $X_k\sim\IPqk\equiv\mbox{IP}(\mu_k;\delta_k,\beta_k,\gamma_k)$
 %\equiv\mbox{IP}(\frac{\mu+k\beta}{1-2k\delta};
 %\frac{\delta}{1-2k\delta},\frac{\beta}{1-2k\delta},\frac{\gamma}{1-2k\delta})$,
 obtaining $X_k^*\sim\IPqkb\equiv\IPkb$ where
 \[
 \mu_k^*=\frac{\mu_k+\beta_k}{1-2\delta_k}
 =\frac{\frac{\mu+k\beta}{1-2k\delta}
 +\frac{\beta}{1-2k\delta}}{1-2\frac{\delta}{1-2k\delta}}
 =\frac{\mu+(k+1)\beta}{1-2(k+1)\delta}=\mu_{k+1}
 \]
 and
 \[
 q_k^*(x)=\frac{q_k(x)}{1-2\delta_k}
 =\frac{\frac{q(x)}{1-2k\delta}}{1-2\frac{\delta}{1-2k\delta}}
 =\frac{q(x)}{1-2(k+1)\delta}=q_{k+1}(x), \quad
 \alpha<x<\omega.
 \]
 On the other hand, since
 $\E q(X_k)=\frac{\E q^{k+1}(X)}{\E q^{k}(X)}$
 and $X_k^*\sim f_k^*$ we get
 \[
 f_k^*(x)=\frac{q_k(x)f_k(x)}{\E q_k(X_k)}
 =\frac{\frac{q(x)}{1-2k\delta}
 \frac{q^k(x)f(x)}{\E q^k(X)}}{\frac{\E q(X_k)}{1-2k\delta}}
 =\frac{\frac{q^{k+1}(x)f(x)}{\E q^k(X)}}
 {\frac{\E q^{k+1}(X)}{\E q^{k}(X)}}
 =\frac{q^{k+1}(x)f(x)}{\E q^{k+1}(X)}=f_{k+1}(x),
 \quad
 \alpha<x<\omega,
 \]
 that is, $X_k^*=X_{k+1}\sim f_{k+1}\sim \mbox{IP}(\mu_{k+1};q_{k+1})$,
 and the proof is complete.
 \end{proof}
 \begin{cor}
 \label{cor.Eq_k}
 If $X\sim\IPq$ and $\E|X|^{2n+2}<\infty$
 (equivalently, if $\delta<\frac{1}{2n+1}$)
 then for each
 $k\in\{0,1,\ldots,n\}$,
 \be
 \label{Eq_k}
 \sigma_k^2:=\Var X_k=\E
 q_k(X_k)=\frac{q(\frac{\mu+k\beta}{1-2k\delta})}{1-(2k+1)\delta},
 \ee
 where $q_k(x)=\delta_k x^2+\beta_k x+\gamma_k$
 and $X_k$ are as in Theorem \ref{theo.star}.
 In particular, if $\delta<1$ then
 \be
 \label{Eq_k2}
 \sigma^2:=\Var X=\E q(X)=\frac{q(\mu)}{1-\delta}.
 \ee
 \end{cor}
 \begin{proof}
 First observe that for any $k\in\{0,1,\ldots,n\}$, $\E |X_k|^2<\infty$
 (and thus, $\E q^k(X_k)<\infty$) since
 $\delta_k=\frac{\delta}{1-2k\delta}<1$
 because $\delta<\frac{1}{2n+1}\leq \frac{1}{2k+1}$.
 Note that it suffices to show only (\ref{Eq_k2}). Indeed,
 since $X_k\sim \IPqk$ it follows from
 (\ref{cov.id2}) (applied to the random variable $X_k$ and to the function $g(x)=x$)
 that $\sigma_k^2=\Var X_k=\E q_k(X_k)$. On the other hand,
 if we manage to show that $\Var X=\frac{q(\mu)}{1-\delta}$ for any
 $X\sim \IPq$ with $\delta<1$ then, by (\ref{Eq_k2}) applied to
 $X_k$, we get
 \[
 \Var X_k = \frac{q_k(\mu_k)}{1-\delta_k}.
 \]
 Since
 \[
 \mu_k=\frac{\mu+k\beta}{1-2k\delta},
 \ \
 q_k(x)=\frac{q(x)}{1-2k\delta}
 \ \
 \mbox{and}
 \ \
 \delta_k=\frac{\delta}{1-2k\delta}<1,
 \]
 (\ref{Eq_k2})
 yields the identity (\ref{Eq_k}) as follows:
 \[
 \E q_k(X_k)=\Var X_k =\frac{q_k(\mu_k)}{1-\delta_k}
 =\frac{\frac{q(\mu_k)}{1-2k\delta}}{1-\frac{\delta}{1-2k\delta}}
 =\frac{q(\mu_k)}{1-(2k+1)\delta}
 =\frac{q(\frac{\mu+k\beta}{1-2k\delta})}{1-(2k+1)\delta}.
 \]
 It remains to verify that $\Var X =\sigma^2=\frac{q(\mu)}{1-\delta}$
 whenever $X\sim\IPq$ and $\delta<1$. To this end, write
 \[
 q(X)=q(\mu)+q'(\mu)(X-\mu)+\delta (X-\mu)^2
 \]
 and take
 expectations to get $\sigma^2=q(\mu)+\delta\sigma^2$,
 which is equivalent to (\ref{Eq_k2}).
 \end{proof}
 \begin{cor}
 \label{cor.Eq^k}
 If $X\sim\IPq$ and $\E|X|^{2n}<\infty$ for some $n\geq 1$
 (i.e.\ $\delta<\frac{1}{2n-1}$)
 then for each
 $k\in\{1,\ldots,n\}$,
 \be
 \label{Eq^k}
 A_k=A_k(\mu;q):=\E q^{k}(X)
 =\frac{\prod_{j=0}^{k-1}(1-2j\delta)}{\prod_{j=0}^{k-1}(1-(2j+1)\delta)}
 \prod_{j=0}^{k-1}q\left(\frac{\mu+j\beta}{1-2j\delta}\right).
 \ee
 \end{cor}
 \begin{proof}
 Observe that
 \[
 (1-2j\delta)\E q_{j}(X_j)=\E q(X_j)=\frac{A_{j+1}}{A_j},
 \ \
 j=0,1,\ldots,n-1,
 \]
 where $A_0:=1$, $q_0=q$, $X_0=X$.
 Multiplying  these relations for $j=0,1,\ldots,k-1$ and using
 (\ref{Eq_k}) we get (\ref{Eq^k}).
 \end{proof}
 \begin{rem}
 \label{rem.complete}
 (a) It is important to note that the identity (\ref{cov.id}) enables
 a convenient calculation of the Fourier coefficients
 of any smooth enough function $g$ with $\Var g(X)<\infty$
 (i.e., $g\in L^2(\R,X)$). Indeed, if $X\sim\IP\equiv\IPq$
 and $\E|X|^{2n}<\infty$
 then the Fourier coefficients $c_k=\E \phi_k(X)g(X)$
 are given by $c_0=\E g(X)$ and
 \be
 \label{Fourier2}
 c_k=\frac{\E q^k(X)g^{(k)}(X)}{(k!c_k(\delta)A_k(\mu;q))^{1/2}},
 \quad k=1,2,\ldots,n,
 \ee
 where $c_k(\delta)$ and $A_k(\mu;q)$ are given by (\ref{lead}) and
 (\ref{Eq^k}), respectively, provided that $g$ is smooth enough so that
 $\E q^k(X) |g^{(k)}(X)|<\infty$ for
 \medskip
 $k\in\{1,2,\ldots,n\}$.
 \\
 (b)
 Obviously, if $X\sim\IP$ and $\delta\leq 0$ (i.e.\ if $X$
 is of Normal, Gamma or Beta-type) then $\E|X|^n<\infty$
 for all $n$. Moreover, since there exist an $\epsilon>0$
 such that $\E e^{tX}<\infty$ for $|t|<\epsilon$
 it follows that the corresponding polynomials $\{\phi_k\}_{k=0}^{\infty}$,
 given by (\ref{orthonormal}), form a complete orthonormal
 system in $L^2(\R;X)$; see, e.g., \cite{Riesz}, \cite{BC}, \cite{APP2}.
 Therefore, for smooth
 enough $g$ with $\Var g(X)<\infty$ and $\E q^k(X) |g^{(k)}(X)|<\infty$
 for all $k\geq 1$, the Fourier coefficients are given by
 \be
 \label{Fourier}
 c_k=\E \phi_k(X)g(X)
 =\frac{\E q^k(X)g^{(k)}(X)}{(k!c_k(\delta)A_k(\mu;q))^{1/2}},
 \quad k=0,1,2,\ldots,
 \ee
 and the variance of $g$ can be calculated as
 (see \cite{APP2}, Theorem 5.1, pp.\ 522--523)
 %\cite[Theorem 5.1, pp.\ 522--523]{APP2})
 \be
 \label{variance}
 \Var g(X)=\sum_{k=1}^{\infty}
 \frac{\E^2 q^k(X)g^{(k)}(X)}{k!c_k(\delta)A_k(\mu;q)}.
 \ee
 Furthermore, the completeness of the Rodrigues polynomials
 (when $X\sim\IP$ and $\delta\leq 0$) enables one to write
 (\cite{APP2}, Theorem 5.2, p.\ 523)
 %\cite[Theorem 5.2, p.\ 523]{APP2})
 \be
 \label{covariance}
 \Cov[g_1(X),g_2(X)]=\sum_{k=1}^{\infty}
 \frac{\E[q^k(X)g_1^{(k)}(X)]\E[q^k(X)g_2^{(k)}(X)]}{k!c_k(\delta)A_k(\mu;q)},
 \ee
 provided that for $i=1,2$, $g_i\in L^2(\R,X)$ and
 $\E q^k(X)|g_i^{(k)}(X)|<\infty$ for
 all $k\geq 1$. The important thing in (\ref{variance}) and (\ref{covariance})
 is that we do not need explicit forms for the polynomials;
 in view of (\ref{lead}) and (\ref{Eq^k}), everything is calculated
 from the four numbers $(\mu;\delta,\beta,\gamma)$ and the derivatives
 of $g$ or $g_i$ ($i=1,2$).
 In particular, for the first three types of Table
 \ref{table densities}, (\ref{variance})
 yields the formulae
 \begin{eqnarray}
 &&\hspace{-1.7cm}\xC{$\ds\Var g(X)
 =\sum_{k=1}^{\infty}\frac{\sigma^{2k}}{k!}
 \E^2 g^{(k)}(X),$\hfill if $X\sim N(\mu,\sigma^2)$,}
 \label{var.normal}
 \\
 &&\hspace{-1.7cm}\xC{$\ds\Var g(X)
 =\sum_{k=1}^{\infty}\frac{\varGamma(a)}{k!\varGamma(a+k)}
 \E^2 X^k g^{(k)}(X),$\hfill if $X\sim \varGamma(a,\lambda)$,}
 \label{var.gamma}
 \\
 &&\hspace{-1.7cm}\xC{$\ds\Var g(X)
 =\sum_{k=1}^{\infty}\frac{(a+b+2k-1)\varGamma(a)
 \varGamma(b)\varGamma(a+b+k-1)}{k!\varGamma(a+b)
 \varGamma(a+k)\varGamma(b+k)}\E^2 X^k (1-X)^kg^{(k)}(X),$\hfill\,}
 \label{var.beta}
 \\
 [-1ex]
 &&\hspace{-1.7cm}\xC{\,\hfill if $X\sim B(a,b)$.}
 \nonumber
 \end{eqnarray}
 \end{rem}

 Turn now to the orthogonal polynomial system
 $\{P_k; \ k=0,1,\ldots,n\}$, of
 (\ref{Rodrigues}), obtained for a random variable $X\sim\IP$
 with support $J(X)=(\alpha,\omega)$ and $\E |X|^{2n}<\infty$ for some $n\geq 2$,
 i.e.\ with $\delta<\frac{1}{2n-1}$.
 By Lemma \ref{lem.star} the random variable
 $X^*=X_1\sim\mbox{IP}(\mu_1;q_1)\equiv
 \mbox{IP}(\mu_1;\delta_1,\beta_1,\gamma_1)$ with
 \[
 \mu_1=\frac{\mu+\beta}{1-2\delta}
 \ \
 \mbox{and}
 \ \
 q_1(x)=\frac{q(x)}{1-2\delta}
 \]
 and has support $(\alpha,\omega)$. Since $\delta<\frac{1}{2n-1}$
 is equivalent to
 $\delta_1=\frac{\delta}{1-2\delta}<\frac{1}{2n-3}$ we
 conclude that $\E|X_1|^{2n-2}<\infty$ and, in particular,
 $\Var X_1<\infty$.
 Therefore, we can define the orthogonal
 polynomial system
 \[
 \{P_{k,1}; \ k=0,1,\ldots,n-1\},
 \]
 by applying (\ref{Rodrigues}) to the density $f_1$ and
 to the quadratic polynomial
 $q_1$ of $X_1$, that is (recall that $f_1(x)=q(x)f(x)/\E q(X)$)
 \be
 \label{polyn.orthogonal.star}
 \begin{split}
 P_{k,1}(x):=\frac{(-1)^k}{f_1(x)}\frac{\ud^k}{\ud x^k}[q_1^k(x)f_1(x)]
 =\frac{(-1)^k}{(1-2\delta)^k q(x)f(x)}\frac{\ud^k}{\ud x^k}[q^{k+1}(x)f(x)],&
 \\
 \alpha<x<\omega,
 \ \  k=0,1,\ldots&,n-1.
 \end{split}
 \ee
 Clearly the system $\{P_{k,1}; \ k=0,1,\ldots,n-1\}$ is
 orthogonal with respect to $X_1$, but the important
 observation is that we can reobtain it by differentiating
 the polynomials $P_k$ (which are orthogonal with respect to $X$).
 In fact, the following lemma holds.
 \begin{lem}
 \label{lem.derivatives}
 If $X\sim\IPq$ and $\E|X|^{2n}<\infty$ for some $n\geq 1$
 then the polynomials $P_k$ of (\ref{Rodrigues}) and $P_{k,1}$
 of (\ref{polyn.orthogonal.star}) are related through
 \be
 \label{derivatives.polynomials}
 \begin{split}
 P_{k+1}'(x)=C_k(\delta)&P_{k,1}(x),
 \ \ k=0,1,\ldots,n-1,\\
 &\textrm{where}
 \ \
 C_k(\delta):=(k+1)(1-k\delta)(1-2\delta)^k.
 \end{split}
 \ee
 \end{lem}
 \begin{proof} First we show that the polynomials
 $P_{k+1}'$ are orthogonal with respect to $X_1$. Indeed,
 $\deg(P_{k+1}')=k$ (for $k=0,1,\ldots,n-1$) and for
 $k,m\in\{0,1,\ldots,n-1\}$ with $k<m$ we have
 \[
 \begin{split}
 \E P_{k+1}'(X_1)P_{m+1}'(X_1)
 &=\frac{1}{\sigma^2}\int_{\alpha}^{\omega}
 P_{m+1}'(x) P_{k+1}'(x) q(x)f(x)\ud{x}\\
 &=\frac{1}{\sigma^2}\bigg\{P_{m+1}(x)
 P_{k+1}'(x) q(x)f(x)\Big|_{\alpha}^{\omega}\\
 &\hspace{10ex}-\int_{\alpha}^{\omega} P_{m+1}(x)
 [P_{k+1}'(x)q(x)f(x)]'\ud{x}\bigg\}.
 \end{split}
 \]
 Now observe that, in view of Lemma \ref{lem.limits2},
 \[
 P_{m+1}(x) P_{k+1}'(x) q(x)f(x)\big|_{\alpha}^{\omega}=0,
 \]
 because
 %\linebreak
 $P_{m+1} P_{k+1}'$ is a polynomial of degree $m+k+1\leq 2n-2$
 and $\E|X|^{2n}<\infty$. Moreover,
 %observe that
 \[
 [P_{k+1}'(x)
 q(x)f(x)]'=P_{k+1}''(x)q(x)f(x)+P_{k+1}'(x)(\mu-x)f(x)=H_{k+1}(x)f(x),
 \]
 where $H_{k+1}(x)=P_{k+1}''(x)q(x)+(\mu-x)P_{k+1}'(x)$ is a polynomial
 in $x$ of degree at most
 $k+1<m+1$. Therefore,
 \[
 \E P_{k+1}'(X_1)P_{m+1}'(X_1)=-\frac{1}{\sigma^2}\E P_{m+1}(X) H_{k+1}(X)=0,
 \]
 since $P_{m+1}$ is orthogonal (with respect to $X$) to any
 polynomial of degree lower than $m+1$.
 Note that the same orthogonality conditions are also valid for
 $\{P_{k,1}\}_{k=0}^{n-1}$, that is,
 \[
 \E P_{k,1}(X_1)
 %\times
 P_{m,1}(X_1)=0
 \ \
 \mbox{for}
 \ \
 k,m\in\{0,1,\ldots,n-1\}
 \ \
 \mbox{with}
 \ \
 k\neq m.
 \]
 Since $\deg(P_{k+1}')=\deg(P_{k,1})=k$,
 $k=0,1,\ldots,n-1$, the uniqueness of the orthogonal polynomial system
 implies that there exist constants $C_k\neq 0$ such
 that $P_{k+1}'(x)=C_k P_{k,1}(x)$. Equating the leading
 coefficients we obtain
 $\lead(P_{k+1}')=C_k\lead(P_{k,1})$, that is (see (\ref{lead})),
 \[
 \small
 \begin{split}
 C_k&=\frac{\lead(P_{k+1}')}{\lead(P_{k,1})}
     =\frac{(k+1)\lead(P_{k+1})}{\lead(P_{k,1})}
     =\frac{(k+1)c_{k+1}(\delta)}{c_k(\delta_1)}
     =\frac{(k+1)\prod_{j=k}^{2k}(1-j\delta)}{\prod_{j=k-1}^{2k-2}(1-j\delta_1)}\\
    &=\frac{(k+1)\prod_{j=k}^{2k}(1-j\delta)}{\prod_{j=k-1}^{2k-2}(1-j\frac{\delta}{1-2\delta})}
     =\frac{(k+1)(1-2\delta)^k \prod_{j=k}^{2k}(1-j\delta)}{\prod_{j=k+1}^{2k}(1-j\delta)}
     =(k+1)(1-k\delta)(1-2\delta)^k.
 \end{split}
 \]
 \vspace*{-1.5cm}

 \end{proof}
 \medskip

 \begin{rem}
 \label{rem.Beale}
 We note that the recurrence (\ref{derivatives.polynomials})
 is
 contained in
 %\cite[eq.\ (2), p.\ 207]{Beale1}.
 %\cite{Beale1}
 Beale (1937), eq.\ (2), p.\ 207.
 Actually,
 Beale's recurrence (which is stated in a much different notation)
 is valid for the polynomials $h_k$ of (\ref{polynomials})
 and for all $k\geq 0$; thus, orthogonality is not, at all, needed for
 deriving it. Specifically, if $p_1=a_0+a_1 x$, $p_2=b_0+b_1x+b_2 x^2$,
 and if $h_k$ are the polynomials in (\ref{polynomials}) and $h_{k,1}$ are the
 polynomials given by
 \[
 h_{k,1}(x):=\frac{1}{p_2(x)f(x)}\frac{\ud^k}{\ud x^k}
 [p_2^{k+1}(x)f(x)],
 \]
 then, with Beale's notation, $h_{k+1}(x)=P_{k+1}(k+1,x)$
 and $h_{k,1}(x)=P_k(k+1,x)$; see also
 %\cite[p.\ 401]{Hild}.
 \cite{Hild}, p.\ 401.
 Therefore, Beale's identity is equivalent to
 (cf.\
 %\cite[eq.\ (2), p.\ 207]{Beale1})
 \cite{Beale1}, eq.\ (2), p.\ 207)
 \be
 \label{Beale}
 h_{k+1}'(x)=(k+1)[a_1+(k+2)b_2]h_{k,1}(x).
 \ee
 On the other hand, the current definition of $P_k$ and $P_{k,1}$
 can be translated  to Beale's notation as follows:
 Since $X\sim\IP\equiv\IPq$ we have from Proposition \ref{prop.1}
 that $f'/f=p_1/p_2$ with $p_2=q$ and $p_1=\mu-x-q'$, that
 is, $a_0=\mu-\beta$, $a_1=-(1+2\delta)$, $b_0=\gamma$, $b_1=\beta$
 and $b_2=\delta$. Furthermore,
 \[
 P_{k+1}(x)
 =\frac{(-1)^{k+1}}{f(x)}\frac{\ud^{k+1}}{\ud x^{k+1}}[q^{k+1}(x)f(x)]
 =\frac{(-1)^{k+1}}{f(x)}\frac{\ud^{k+1}}{\ud x^{k+1}}[p_2^{k+1}(x)f(x)]
 =(-1)^{k+1}h_{k+1}(x)
 \]
 and
 \[
 P_{k,1}(x)
 =\frac{(-1)^{k}}{f_1(x)}\frac{\ud^{k}}{\ud x^{k}}[q_1^{k}(x)f_1(x)]
 =\frac{(-1)^{k}}{q(x)f(x)}\frac{\ud^{k}}{\ud x^{k}}
 \left[\frac{q^{k}(x)}{(1-2\delta)^k}q(x)f(x)\right]
 =\frac{(-1)^k}{(1-2\delta)^k}h_{k,1}(x).
 \]
 Thus, $h_{k+1}=(-1)^{k+1}P_{k+1}$,
 $h_{k,1}=(-1)^k(1-2\delta)^k P_{k,1}$ and
 (\ref{Beale}) yields
 \[
 (-1)^{k+1}P_{k+1}'=(k+1)[-(1+2\delta)+(k+2)\delta](-1)^k(1-2\delta)^k
 P_{k,1}.
 \]
 That is, $P_{k+1}'=(k+1)[(1+2\delta)-(k+2)\delta](1-2\delta)^k P_{k,1}=
 (k+1)(1-k\delta)(1-2\delta)^k P_{k,1}$, which shows that
 (\ref{derivatives.polynomials})
 holds for all  $k\in\{0,1,\ldots\}$.
 \end{rem}
 Applying Lemma \ref{lem.derivatives} inductively it is
 easy to verify the following result.
 \begin{theo}
 \label{theo.derivatives}
 If $X\sim\IP$
 with
 support $J(X)=(\alpha,\omega)$ and
 $\E|X|^{2n}<\infty$ for some $n\geq 1$
 (i.e.\ $\delta<\frac{1}{2n-1}$) then
 \be
 \label{derivatives.higher}
 P_{k+m}^{(m)}(x)=C^{(m)}_k(\delta) P_{k,m}(x), \quad
 %\alpha<x<\omega, \ \
 m=1,2,\ldots,n, \quad k=0,1,\ldots,n-m,
 \ee
 where
 \be
 \label{derivatives.higher.polynomials}
 C^{(m)}_k(\delta):=\frac{(k+m)!}{k!}(1-2m\delta)^k
 \prod_{j=k+m-1}^{k+2m-2}(1-j\delta).
 \ee
 Here, $P_k$ are the polynomials given by (\ref{Rodrigues})
 associated with $f$, and $P_{k,m}$ are the corresponding
 Rodrigues polynomials of (\ref{Rodrigues}),
 associated with the density $f_m(x)=\frac{q^m(x)f(x)}{\E q^m(X)}$,
 $\alpha<x<\omega$,
 of the random variable $X_m\sim\mbox{IP}(\mu_m;q_m)$
 of Theorem \ref{theo.star}, i.e.,
 \be
 \label{polyn.orthogonal.gen.star}
 \begin{split}
 P_{k,m}(x):=\frac{(-1)^k}{f_m(x)}
 \frac{\ud^k}{\ud x^k}[q_m^k(x)f_m(x)]
 =\frac{(-1)^k}{(1-2m\delta)^k q^m(x)f(x)}
 \frac{\ud^k}{\ud x^k}[q^{k+m}(x)f(x)],
 \\
 \alpha<x<\omega, \quad k=0,1,\ldots,n-m.
 \end{split}
 \ee
 \end{theo}
 \begin{proof}
 %Since $X_m\sim\mbox{IP}(\mu_m;\delta_m,\beta_m,\gamma_m)$
 %and $\delta_m=\frac{\delta}{1-2m\delta}$ it is easily seen
 %that $\delta<\frac{1}{2n-1}$ is equivalent to $\delta_m<\frac{1}{2n-2m-1}$
 %so that $\E|X_m|^{2n-2m}<\infty$ and, thus, $\E|P_{k,m}(X_m)|^2<\infty$
 %for all $k=0,1,\ldots,n-m$.
 Apply first Lemma
 \ref{lem.derivatives} to get
 \[
 P_{k+m}'=P_{(k+m-1)+1}'=(k+m)(1-(k+m-1)\delta)(1-2\delta)^{k+m-1}P_{k+m-1,1}.
 \]
 Now, since $P_{k+m-1,1}$ are the Rodrigues polynomials of
 $f_1$ we can apply again Lemma \ref{lem.derivatives} to
 $X_1$ with $\delta_1=\frac{\delta}{1-2\delta}$.
 It follows that
 \[
 P_{k+m-1,1}'=P_{(k+m-2)+1,1}'
 =(k+m-1)(1-(k+m-2)\delta_1)(1-2\delta_1)^{k+m-2}P_{k+m-2,2}.
 \]
 Combining the above equations we see that
 \[
 \begin{split}
 P_{k+m}''&=(k+m)(1-(k+m-1)\delta)(1-2\delta)^{k+m-1}P_{k+m-1,1}'\\
 &=(k+m)(k+m-1)(1-(k+m-1)\delta)(1-(k+m-2)\delta_1)\\
 &\hspace{10em}\times(1-2\delta)^{k+m-1}(1-2\delta_1)^{k+m-2}P_{k+m-2,2}.
 \end{split}
 \]
 By the same argument it follows that for any $j\in\{0,1,\ldots,m-1\}$,
 \medskip
 \\
 \centerline{\small
 $P_{k+m-j,j}'=P_{(k+m-j-1)+1,1}'
 =(k+m-j)(1-(k+m-j-1)\delta_j)(1-2\delta_j)^{k+m-j-1}
 \medskip
 P_{k+m-j-1,j+1}$,}

 \noindent
 where $\delta_j={\delta}/(1-2j\delta)$.
 Thus, we can easily show, using (finite) induction
 on $s$, that for any $s\in\{1,2,\ldots,m\}$,
 {\small
 \[
 P_{k+m}^{(s)}=
 \left\{
 \left(
 \prod_{j=0}^{s-1} (k+m-j)
 \right)\hspace{-.5ex}
 \left(
 \prod_{j=0}^{s-1}
 (1-(k+m-j-1)\delta_j)
 \right)\hspace{-.5ex}
 \left(
 \prod_{j=0}^{s-1}
 (1-2\delta_j)^{k+m-j-1}
 \right)
 \right\}
 P_{k+m-s,s}.
 \]}

 \noindent
 Setting $s=m$ it follows that
 (\ref{derivatives.higher}) is satisfied with
 \[
 C^{(m)}_k(\delta)=
 \left(
 \prod_{j=0}^{m-1} (k+m-j)
 \right)
 \left(
 \prod_{j=0}^{m-1}
 (1-(k+m-j-1)\delta_j)
 \right)
 \left(
 \prod_{j=0}^{m-1}
 (1-2\delta_j)^{k+m-j-1}
 \right).
 \]
 Now it suffices to observe that
 $\prod_{j=0}^{m-1}
 (k+m-j)=\frac{(k+m)!}{k!}$,
 that
 \[
 \begin{split}
 \prod_{j=0}^{m-1} (1-(k+m-j-1)\delta_j)&=\prod_{j=0}^{m-1}
 \left(1-(k+m-j-1)\frac{\delta}{1-2j\delta}\right)
 \\
 &=\frac{\prod_{j=0}^{m-1} (1-(k+m+j-1)\delta)}{\prod_{j=0}^{m-1}
 (1-2j\delta)}
 =\frac{\prod_{j=k+m-1}^{k+2m-2}
 (1-j\delta)}{\prod_{j=0}^{m-1} (1-2j\delta)},
 \end{split}
 \]
 and that
 \[
 \begin{split}
 \prod_{j=0}^{m-1}(1-2\delta_j)^{k+m-j-1}
 &=\prod_{j=0}^{m-1}\left(1-2\frac{\delta}{1-2j\delta}\right)^{k+m-j-1}
  =\prod_{j=0}^{m-1}\left(\frac{1-2(j+1)\delta}{1-2j\delta}\right)^{k+m-j-1}
  \\
 &=\frac{\prod_{j=0}^{m-1}(1-2(j+1)\delta)^{k+m-j-1}}
 {\prod_{j=0}^{m-1}(1-2j\delta)^{k+m-j-1}}
  =\frac{\prod_{j=1}^{m}(1-2j\delta)^{k+m-j}}
  {\prod_{j=1}^{m-1}(1-2j\delta)^{k+m-j-1}}
  \\
 &=(1-2m\delta)^{k}\frac{\prod_{j=1}^{m-1}(1-2j\delta)^{k+m-j}}
 {\prod_{j=1}^{m-1}(1-2j\delta)^{k+m-j-1}}
 \\
 &=(1-2m\delta)^{k}\prod_{j=1}^{m-1}(1-2j\delta).
 \end{split}
 \]
 \vspace{-3.5em}

 \end{proof}
 \begin{rem}
 \label{rem.Beale.general.derivatives}
 (a)
 An alternative calculation of the constant
 $C_k=C^{(m)}_k(\delta)$ can be given as follows.
 Lemma \ref{lem.derivatives} guarantees that
 $P_{k+m}^{(m)}(x)=C_k P_{k,m}(x)$ for some constant
 $C_k$. Arguing as in the proof of Lemma \ref{lem.derivatives}
 we see that $C_k$  can be derived from the corresponding leading
 coefficients. Indeed, since
 $\lead(P_{k+m}^{(m)})=C_k \lead(P_{k,m})$,
 we get, in view of (\ref{lead}), that
 \[
 \begin{split}
 C_k&=\frac{\lead(P_{k+m}^{(m)})}{\lead(P_{k,m})}
     =\frac{\frac{(k+m)!}{k!}\lead(P_{k+m})}{\lead(P_{k,m})}
     =\frac{\frac{(k+m)!}{k!}c_{k+m}(\delta)}{c_k(\delta_m)}
     =\frac{\frac{(k+m)!}{k!}\prod_{j=k+m-1}^{2k+2m-2}(1-j\delta)}
     {\prod_{j=k-1}^{2k-2}(1-j\delta_m)}
     \\
    &=\frac{\frac{(k+m)!}{k!}\prod_{j=k+m-1}^{2k+2m-2}(1-j\delta)}
    {\prod_{j=k-1}^{2k-2}(1-j\frac{\delta}{1-2m\delta})}
     =\frac{\frac{(k+m)!}{k!}(1-2m\delta)^k \prod_{j=k+m-1}^{2k+2m-2}(1-j\delta)}
     {\prod_{j=k+2m-1}^{2k+2m-2}(1-j\delta)}
     \\
    &=\frac{(k+m)!}{k!}(1-2m\delta)^k\prod_{j=k+m-1}^{k+2m-2}(1-j\delta).
 \end{split}
 \]
 (b)
 We note that the recurrence (\ref{derivatives.higher})
 is contained in Beale (1937), eq.\ (4), p.\ 207,
 %\cite[eq.\ (4), p.\ 207]{Beale1},
 although it is stated in a quite different notation there.
 %Beale's recurrence
 %(which is stated in a much different notation)
 %is valid for the polynomials $h_k$ of (\ref{polynomials})
 %and for all $k\geq 0$; thus, orthogonality is not, at all, needed for
 %deriving it.
 Specifically, if $p_1=a_0+a_1 x$, $p_2=b_0+b_1x+b_2 x^2$,
 and if $h_k$ are the polynomials in (\ref{polynomials}) and $h_{k,m}$ are the
 polynomials given by
 \[
 h_{k,m}(x):=\frac{1}{p_2^m(x)f(x)}\frac{\ud^k}{\ud x^k}
 [p_2^{k+m}(x)f(x)],
 \]
 then, with Beale's notation, $h_{k+m}(x)=P_{k+m}(k+m,x)$
 and $h_{k,m}(x)=P_k(k+m,x)$.
 %; see also Hildebrandt (1931), p.\ 401.
 Therefore, putting $q\to m$, $k\to k+m$, $n\to k+m-1$,
 $N'\to -(1+2\delta)$ and
 $D''\to 2\delta$ in eq.\ (4) of
 \cite{Beale1},
 %\cite[eq.\ (4), p.\ 207]{Beale1},
 we get
 \be\label{Beale2}
 \begin{split}
 h_{k+m}^{(m)}(x)&=\left(\prod_{i=0}^{m-1}(k+m-i)((k+m+i-1)\delta-1)\right)
 h_{k,m}(x)
 \\
 &=(-1)^m\frac{(k+m)!}{k!}\left(\prod_{j=k+m-1}^{k+2m-2}(1-j\delta)\right)h_{k,m}(x),
 \qquad k=0,1,2,\ldots \ .
 \end{split}
 \ee
 On the other hand it is easy to see that
 $P_{k+m}(x)
 %=\frac{(-1)^{k+1}}{f(x)}\frac{d^{k+1}}{d x^{k+1}}[q^{k+1}(x)f(x)]
 %=\frac{(-1)^{k+1}}{f(x)}\frac{d^{k+1}}{d x^{k+1}}[p_2^{k+1}(x)f(x)]
 =(-1)^{k+m}h_{k+m}(x)$ and, with $p_2=q$,
 \[
 P_{k,m}(x)
 =\frac{(-1)^{k}}{f_m(x)}\frac{\ud^{k}}{\ud x^{k}}[q_m^{k}(x)f_m(x)]
 =\frac{(-1)^{k}}{p_2^m(x)f(x)}\frac{\ud^{k}}{\ud x^{k}}\left[\frac{p_2^{k}(x)}{(1-2m\delta)^k}p_2^m(x)f(x)\right]
 =\frac{(-1)^kh_{k,m}(x)}{(1-2m\delta)^k}.
 \]
 Thus, $h_{k+m}=(-1)^{k+m}P_{k+m}$,
 $h_{k,m}=(-1)^k(1-2m\delta)^k P_{k,m}$, and
 (\ref{Beale2}) becomes
 \[
 (-1)^{k+m}P_{k+m}^{(m)}
 =(-1)^m \frac{(k+m)!}{k!}\left(\prod_{j=k+m-1}^{k+2m-2}(1-j\delta)\right)
 (-1)^k(1-2m\delta)^k P_{k,m}, \quad k=0,1,\ldots;
 %h_{k,m}(x) =(k+1)[-(1+2\delta)+(k+2)\delta](-1)^k(1-2\delta)^k P_{k,1},
 \]
 equivalently,
 $P_{k+m}^{(m)}=\frac{(k+m)!}{k!}(1-2m\delta)^k
 \left(\prod_{j=k+m-1}^{k+2m-2}(1-j\delta)\right)P_{k,m}$,
 which shows that (\ref{derivatives.higher})
 holds for all
 \medskip
 $k\in\{0,1,\ldots\}$.
 \\
 (c) Krall \cite{Krall1}, \cite{Krall2} characterizes the Pearson system from
 the fact that the derivatives of orthogonal polynomials
 are orthogonal polynomials.
 \end{rem}
 We can now adapt the preceding
 results to the corresponding orthonormal polynomial
 systems. Notice that the following
 corollary contains the main interest regarding
 Fourier expansions within the Pearson family and,
 to our knowledge, it is not stated elsewhere in the
 present simple, unified, explicit form.

 \begin{cor}
 \label{cor.derivatives}
 Let $X\sim\IP\equiv\IPq$ with support
 $(\alpha,\omega)$,
 and assume that
 %support $J(X)=(\alpha,\omega)$ and
 $\E|X|^{2n}<\infty$ for some fixed $n\geq 1$
 (equivalently, $\delta<\frac{1}{2n-1}$).
 Let $\{\phi_k\}_{k=0}^{n}$ be the orthonormal
 polynomials associated with $X$ (with $\lead(\phi_k)>0$ for
 all $k$;
 see (\ref{orthonormal}), (\ref{lead.orth})),
 fix a number $m\in\{0,1,\ldots,n\}$,
 and consider the corresponding orthonormal polynomials
 $\{\phi_{k,m}\}_{k=0}^{n-m}$,
 with $\lead(\phi_{k,m})>0$,
 associated with
 \[
 X_m\sim f_m(x)=\frac{q^m(x) f(x)}{\E q^{m}(X)},
 \ \ \
 \alpha<x<\omega.
 \]
 Then there exist constants $\nu_k^{(m)}=\nu_k^{(m)}(\mu;q)>0$
 such that
 \be
 \label{orthonormal.derivatives}
 \phi_{k+m}^{(m)}(x)=\nu_k^{(m)} \phi_{k,m}(x),
 \ \ \
 \alpha<x<\omega,
 \ \ \
 %\ \
 %\alpha<x<\omega, \ \
 %m=1,2,\ldots,n,
 k=0,1,\ldots,n-m.
 \ee
 Specifically, the constants $\nu_k^{(m)}$ have the explicit  form
 \be
 \label{orthonormal.constants}
 \nu_k^{(m)}=\nu^{(m)}_k(\mu;q):
 =\left\{
 \frac{\frac{(k+m)!}{k!}\prod_{j=k+m-1}^{k+2m-2}(1-j\delta)}{A_m(\mu;q)}
 \right\}^{1/2},
 \ee
 where $A_m(\mu;q)=\E q^m(X)$ is given by (\ref{Eq^k}).
 In particular, setting $\sigma^2=\Var X$ we have
 \be
 \label{eq.paragwgos}
 \phi_{k+1}'(x)=\frac{\sqrt{(k+1)(1-k\delta)}}{\sigma}
 \phi_{k,1}(x)=\sqrt{\frac{(k+1)(1-\delta)(1-k\delta)}{q(\mu)}}\phi_{k,1}(x),
 \ \
 k=0,1,\ldots,n-1.
 \ee
 %and
 %$P_k$ are the polynomials given by (\ref{Rodrigues})
 %corresponding to $f$ and $P_{k,m}$ are the corresponding
 %Rodrigues polynomials of (\ref{Rodrigues}),
 %associated with the density $f_m(x)=\frac{q^m(x)f(x)}{\E q^m(X)}$,
 %$\alpha<x<\omega$,
 %of the random variable $X_m\sim\mbox{IP}(\mu_m;q_m)$
 %of Theorem \ref{theo.star}, i.e.,
 %\begin{eqnarray}
 %\label{polynomial.orthonormal.star}
 %&& P_{k,m}(x):=\frac{(-1)^k}{f_m(x)}\frac{d^k}{d x^k}[q_m^k(x)
 % f_m(x)]=\frac{(-1)^k}{(1-2m\delta)^k q^m(x)f(x)}\frac{d^k}{d x^k}[q^{k+m}(x)
 % f(x)], \\
 %&&
 %\hspace{50ex}
 %\alpha<x<\omega, \ \ k=0,1,\ldots,n-m.
 %\nonumber
 %\end{eqnarray}
 \end{cor}
 \begin{proof}
 Observe that
 \[
 \phi_{k+m}(x)=\frac{P_{k+m}(x)}{\sqrt{\E |P_{k+m}(X)|^2}}
 \ \
 \mbox{and}
 \ \
 \phi_{k,m}(x)=\frac{P_{k,m}(x)}
 {\sqrt{\E|P_{k,m}(X_m)|^2}},
 \ \
 \alpha<x<\omega,
 \]
 where $P_{k+m}$ and $P_{k,m}$ are as in Theorem
 \ref{theo.derivatives}.
 %Moreover, from Theorem \ref{theo.derivatives} we have that
 Since
 \[
 P_{k+m}^{(m)}(x)=C_k^{(m)}(\delta)P_{k,m}(x),
 \ \ \
 \alpha<x<\omega,
 \]
 we conclude that
 %Thus,
 there exists a constant $\nu_k^{(m)}$ such that
 $\phi_{k+m}^{(m)}(x)=\nu_k^{(m)}\phi_{k,m}(x)$. Hence,
 \[
 \begin{split}
 \nu_k^{(m)}&=\frac{\lead(\phi_{k+m}^{(m)})}{\lead(\phi_{k,m})}
             =\frac{\frac{(k+m)!}{k!}\lead(\phi_{k+m})}{\lead(\phi_{k,m})}
             =\frac{\frac{(k+m)!}{k!}\frac{\lead(P_{k+m})}
             {\sqrt{\E |P_{k+m}(X)|^2}}}{\frac{\lead(P_{k,m})}
             {\sqrt{\E|P_{k,m}(X_m)|^2}}}
             \\
            &=\frac{(k+m)!\ \lead(P_{k+m})\sqrt{\E|P_{k,m}(X_m)|^2}}{k!\
            \lead(P_{k,m})\sqrt{\E|P_{k+m}(X)|^2}}
             =\frac{(k+m)!\ c_{k+m}(\delta) \sqrt{\E|P_{k,m}(X_m)|^2}}{k!\
             c_{k}(\delta_m)\sqrt{\E|P_{k+m}(X)|^2}},
 \end{split}
 \]
 where, by (\ref{lead}),
 $c_{k+m}(\delta)=\prod_{j=k+m-1}^{2k+2m-2}(1-j\delta)$ and
 \[
 \begin{split}
 c_{k}(\delta_m)&=\prod_{j=k-1}^{2k-2}(1-j\delta_m)
 =\prod_{j=k-1}^{2k-2}(1-j\frac{\delta}{1-2m\delta})\\
 &=\frac{\prod_{j=k-1}^{2k-2}(1-(2m+j)\delta)}{(1-2m\delta)^k}
 =\frac{\prod_{j=k+2m-1}^{2k+2m-2}(1-j\delta)}{(1-2m\delta)^k}.
 \end{split}
 \]
 From (\ref{orth}) we see that
 $\E|P_{k+m}(X)|^2=(k+m)!c_{k+m}(\delta)\E q^{k+m}(X)$
 and
 \[
 \E|P_{k,m}(X_m)|^2=k!c_{k}(\delta_m)
 \E q_m^{k}(X_m)=k!c_{k}(\delta_m)\frac{\E q_m^{k}(X)q^m(X)}{\E q^m(X)}
 =\frac{k!c_{k}(\delta_m)\E q^{k+m}(X)}{(1-2m\delta)^k\E q^m(X)}.
 \]
 Combining the preceding relations we obtain
 \[
 \begin{split}
 \nu_k^{(m)}&=\frac{(k+m)!\ c_{k+m}(\delta)
 \sqrt{\E|P_{k,m}(X_m)|^2}}{k!\ c_{k}(\delta_m)\sqrt{\E|P_{k+m}(X)|^2}}
             =\frac{(k+m)!\ c_{k+m}(\delta)
             \sqrt{\frac{k!c_{k}(\delta_m)\E q^{k+m}(X)}{(1-2m\delta)^k
             \E q^m(X)}}}{k!\ c_{k}(\delta_m)\sqrt{(k+m)!c_{k+m}(\delta)
             \E q^{k+m}(X)}}
             \\
            &=\frac{(k+m)!\ c_{k+m}(\delta) \sqrt{k!c_{k}(\delta_m)
            \E q^{k+m}(X)}}{k!\ c_{k}(\delta_m)\sqrt{(k+m)!c_{k+m}(\delta)
            \E q^{k+m}(X)(1-2m\delta)^k\E q^m(X)}}
            \\
            &=\frac{\sqrt{(k+m)!c_{k+m}(\delta)}}{\sqrt{k!c_{k}(\delta_m)(1-2m\delta)^k
            \E q^m(X)}}=\sqrt{\frac{(k+m)!}{k!\E q^m(X)}} \sqrt{\frac{c_{k+m}(\delta)}
            {c_{k}(\delta_m)(1-2m\delta)^k}}
            \\
            &=\sqrt{\frac{(k+m)!}{k!\E q^m(X)}}
            \sqrt{\frac{\prod_{j=k+m-1}^{2k+2m-2}(1-j\delta)}
            {\frac{\prod_{j=k+2m-1}^{2k+2m-2}(1-j\delta)}{(1-2m\delta)^k}(1-2m\delta)^k}}
            =\sqrt{\frac{(k+m)!}{k!\E q^m(X)}\prod_{j=k+m-1}^{k+2m-2}(1-j\delta)},
 \end{split}
 \]
 and the proof is complete.
 \end{proof}

 {\sc Acknowledgements.} We would like to thank
 M.C.\ Jones and H.\ Papageorgiou for a number of
 suggestions and helpful comments that improved the presentation.

 \begin{appendix}

 \end{appendix}
\end{document}